\documentclass{article}
\usepackage{amssymb, amsmath, amsthm, mathtools, bbm, tikz-cd,stmaryrd,enumerate,enumitem}
\usepackage{fancyhdr}
\newcommand{\TT}{{T\oplus T^*}}
\newcommand{\JJ}{\mathcal{J}}
\newcommand{\KK}{\mathcal{K}}
\newcommand{\GG}{\mathcal{G}}
\newcommand{\Aa}{{\mathcal{A}}}

\newcommand{\RR}{\mathbb{R}}
\newcommand{\XX}{\mathfrak{X}}
\newcommand{\HH}{\mathcal{H}}
\newcommand{\FF}{\mathcal{F}}
\newcommand{\II}{\mathcal{I}}
\newcommand{\Lb}{\mathbb{L}}
\newcommand{\Lbt}{{\widetilde{\Lb}}}

\newcommand{\id}{\mathbbm{1}}

\newcommand{\lc}{\mathring{\n}}

\newcommand{\Lie}{\mathcal{L}}
\newcommand{\PS}{\mathcal{P}}
\newcommand{\ap}{\alpha}
\newcommand{\bt}{\beta}
\def\w{\wedge}
\newcommand{\p}{\partial}
\newcommand{\pt}{\tilde{\partial}}
\newcommand{\xt}{{\tilde{x}}}
\newcommand{\n}{\nabla}
\newcommand{\rd}{\mathrm{d}}

\newcommand{\Lt}{\tl{L}}

\newcommand{\se}{\Gamma}
\newcommand{\Endo}{\text{End}}
\newcommand{\ellt}{{\tl{\ell}}}

\newcommand{\la}{\langle}
\newcommand{\ra}{\rangle}
\newcommand{\lara}{\la\ ,\ \ra}
\newcommand{\brac}{[\ ,\ ]}
\newcommand{\bl}{[\![}
\newcommand{\br}{]\!]}
\newcommand{\bracd}{\bl \ ,\ \br}
\newcommand{\yt}{\tl{y}}

\newcommand{\pb}{{(+)}}
\newcommand{\mb}{{(-)}}
\newcommand{\pbmb}{{(\pm)}}

\def\tl{\tilde}

\newtheorem{theorem}{Theorem}[section]
\newtheorem{proposition}[theorem]{Proposition}
\newtheorem{lemma}[theorem]{Lemma}
\newtheorem{corollary}[theorem]{Corollary}

\newtheorem*{theorem*}{Theorem}
\newtheorem*{lemma*}{Lemma}
\newtheorem*{proposition*}{Proposition}

\theoremstyle{definition}
\newtheorem{Def}[theorem]{Definition}

\theoremstyle{definition}
\newtheorem*{notation*}{Notation}
\newtheorem*{definition*}{Definition}

\theoremstyle{remark}
\newtheorem*{remark*}{Remark}
\newtheorem{remark}[theorem]{Remark}
\newtheorem{Ex}[theorem]{Example}

\input xy

\xyoption{all}

\DeclareMathOperator{\End}{End}

\title{Commuting Pairs, Generalized para-K\"ahler Geometry and Born Geometry}
\author{\\ \ \\ Shengda Hu\footnote{Department of Mathematics, Wilfrid Laurier University; {\tt shu@wlu.ca}}, Ruxandra Moraru\footnote{Department of Pure Mathematics, University of Waterloo; {\tt moraru@uwaterloo.ca}}, David Svoboda\footnote{Perimeter Institute for Theoretical Physics; {\tt dsvoboda@perimeterinstitute.ca}}}
\date{}
\begin{document}
%
\maketitle
\begin{abstract}
In this paper, we study the geometries given by commuting pairs of generalized endomorphisms $\Aa \in \End(\TT)$ with the property that their product defines a generalized metric. There are four types of such commuting pairs: generalized K\"ahler (GK), generalized para-K\"ahler (GpK), generalized chiral and generalized anti-K\"ahler geometries. We show that GpK geometry is equivalent to a pair of para-Hermitian structures and we derive the integrability conditions in terms of these. From the physics point of view, this is the geometry of $2D$ $(2,2)$ twisted supersymmetric sigma models. The generalized chiral structures are equivalent to a pair of tangent bundle product structures that also appear in physics applications of $2D$ sigma models. We show that the case when the two product structures anti-commute corresponds to Born geometry. Lastly, the generalized anti-K\"ahler structures are equivalent to a pair of anti-Hermitian structures (sometimes called Hermitian with Norden metric). The generalized chiral and anti-K\"ahler geometries do not have isotropic eigenbundles and therefore do not admit the usual description of integrability in terms of the Dorfman bracket. We therefore use an alternative definition of integrability in terms of the generalized Bismut connection of the corresponding metric, which for GK and GpK commuting pairs recovers the usual integrability conditions and can also be used to define the integrability of generalized chiral and anti-K\"ahler structures. In addition, it allows for a weakening of the integrability condition, which has various applications in physics.

\end{abstract}
\newpage
%
%
%
\setcounter{tocdepth}{2}
\tableofcontents

\section{Introduction}

Let $M$ be a smooth real manifold, and denote by $T$ its tangent bundle and $T^*$ its cotangent bundle. In this paper, we are interested in geometries on both the tangent bundle $T$ and the {\em generalized tangent bundle} $\TT$, and how they are related. A first example of this is Dirac geometry.

Dirac geometry was introduced as an elegant way of unifying presymplectic and Poisson geometry via Dirac structures \cite{Courant-Weinstein,Dorfman}. More precisely, the data of a presymplectic two-form 
$\omega \in \se(\Lambda^2 T^*)$ or a Poisson bi-vector $\Pi \in \se(\Lambda^2 T)$ on the manifold $M$ is encoded in terms of Dirac structures, which are subbundles of $\TT$ that are maximally isotropic with respect to the pairing
\begin{align}\label{intro:pairing}
\la X+\ap,Y+\bt\ra=\ap(Y)+\bt(X),
\end{align}
where $X,Y \in \Gamma(T)$ and $\alpha,\beta \in \Gamma(T^*)$,
and are involutive with respect to the {\em Dorfman bracket}. Note that the integrability of the structures $\omega$ and $\Pi$ (which corresponds to $\rd \omega = 0$ and $[\Pi,\Pi]=0$, respectively) is repackaged as the involutivity of these subbundles of $\TT$. 

In the seminal thesis \cite{Gualtieri:2003dx}, it was shown that one can continue this idea beyond pre-symplectic and Poisson structures and incorporate in this framework many other structures, in particular, complex and holomorphic Poisson structures. It was also shown that complex Dirac structures $\Lb\subset (\TT)\otimes \mathbb{C}$ such that  $\Lb\oplus \overline{\Lb}=(\TT)\otimes \mathbb{C}$ are equivalent to {\it generalized complex structures}, that is, integrable endomorphisms $\II$ of $\TT$ that satisfy
\begin{align*}
\II^2=-\id,\quad \la \II\cdot,\II\cdot\ra=\la\cdot,\cdot\ra.
\end{align*}
(The integrability of the generalized complex structure $\II$ was again defined as closure of the corresponding Dirac structure $\Lb$ under the Dorfman bracket.)

A natural next step was then to study pairs of generalized complex structures, which was also done in \cite{Gualtieri:2003dx} under the name of {\it generalized K\"ahler} (GK) geometry. A GK structure is a {\it pair} of commuting generalized complex structures $\II_\pm$ whose product $\GG:=-\II_+\II_-$, $\GG^2=\id$, is a generalized metric. This means, in particular, that $\GG$ is {\it non-degenerate} in the sense that the $\pm 1$-eigenbundles $C_\pm$ of $\GG$  are both isomorphic to the tangent bundle $T$ via the projections 
\[\pi_\pm:C_\pm \subset T \oplus T^* \rightarrow T.\] 
The isomorphisms $\pi_\pm$ induce a Riemannian metric on $M$ given by
\[ g(X,Y) := \pm \la \pi_\pm^{-1}(X),\pi_\pm^{-1}(Y) \ra, \]
$X,Y \in \Gamma(T)$.
Moreover, one can define a pair of endomorphisms $I_\pm$ of $T$ by
\begin{align*}
I_+:=\pi_+ \II_\pm \pi^{-1}_+, \quad I_-:=\pm\pi_- \II_\pm \pi^{-1}_-
\end{align*}
such that $I_\pm^2 = -\id_T$ and $g(I_\pm \cdot, I_\pm \cdot) = g(\cdot,\cdot)$. In other words, the GK structure $\II_\pm$ induces an almost {\em bi-Hermitian} structure $(g,I_\pm)$ on $M$. Moreover, the integrability of the generalized complex structures $\II_\pm$ is equivalent to both almost complex structures $I_\pm$ being integrable and both fundamental forms $\omega_\pm=g I_\pm$ of $g$ satisfying the condition
\begin{align*}
\rd^c_\pm \omega_\pm=\pm H
\end{align*}
for some closed $3$-form $H \in \Gamma(\Lambda^3 T^*)$, where $\rd^c_\pm = I_\pm^*\circ \rd \circ I_\pm^*$ are the twisted differentials with respect to $I_\pm$. Interestingly, bi-Hermitian geometry first appeared in the 80's in a famous physics paper by Gates, Hull and Rocek \cite{Gates-hull-rocek-biherm}, as the target geometry of $(2,2)$ supersymmetric non-linear sigma models. Bi-Hermitian manifolds have since then been extensively studied by many authors (see, for example, \cite{Apostolov1999,Hitchin2006} and the references therein). An important class of GK manifolds is given by hyperHermitian manifolds (for which the pair of complex structures $I_\pm$ {\em anti-commute}, that is, $\{I_+,I_-\}:=I_+I_- + I_-I_+ = 0$); these correspond to hyperK\"ahler manifolds when $H = 0$.

In this paper, we extend the construction of GK geometry to commuting pairs of other types of generalized structures. 
To be specific, we consider pairs of endomorphisms $\Aa_\pm$ of $T \oplus T^*$ that commute and satisfy
\begin{align*}
\Aa_\pm^2=\ap\id, \quad \la \Aa_\pm\cdot,\Aa_\pm\cdot\ra=\bt\la\cdot,\cdot\ra,
\end{align*}
where $(\ap,\bt)$ is a pair of signs. The product $\GG=\Aa_+\Aa_-$ is then always of type $(+,+)$; when it is also non-degenerate, $\GG$ is a generalized (indefinite) metric and one again gets a pseudo-Riemannian metric $\eta$ on $M$ and a pair $A_\pm$ of tangent bundle endomorphisms on $M$ (see Sections \ref{sec:gen_metric} and \ref{sec:commuting_pairs} for details). Note that GK geometry corresponds to the case where $\Aa_\pm$ are of type $(-,+)$. We obtain also three additional geometries this way:
 \begin{itemize}
 	\item 
 	{\it generalized para-K\"ahler}\footnote{This is not the geometry described in \cite{vaisman2015generalized} by Vaisman, where $\Aa_\pm$ is a commuting pair of generalized structures of different types. In the present paper we use the name generalized para-K\"ahler, however, for the para-complex version of GK geometry. More on this is explained in Remark \ref{rem:commutingpairs}.} (GpK) geometry if $\Aa_\pm$ are of type $(+,-)$, 
 	\item 
 	{\it generalized chiral} (GCh) geometry if $\Aa_\pm$ are of type $(+,+)$,
 	\item 
 	{\it generalized anti-K\"{a}hler} (GaK) geometry if $\Aa_\pm$ are of type $(-,-)$.
 \end{itemize}
In this paper, we will mostly focus on generalized para-K\"ahler geometry (see Section \ref{sec:GpK}), but nonetheless provide a notion of integrability as well as examples of generalized chiral and generalized anti-K\"ahler structures (see Sections \ref{sec:gen_chiral} and \ref{sec:bismut_integr}).

\bigskip
Let us state some of our main results concerning these geometries.

\smallskip
\noindent
\textbf{Generalized para-K\"ahler geometry.}
We begin by summarizing our results on GpK geometry. We first note that many of the properties of GpK structures are analogous to those of GK structures, with some statements made in the para-holomorphic category instead of the holomorphic one. 

Given a GpK structure $\KK_\pm$, one recovers a pair $(\eta,K_\pm)$ of para-Hermitian structures on $M$:
\begin{align*}
K_\pm \in \End(T), \quad K_\pm^2=\id,\quad \eta(K_\pm\cdot,K_\pm\cdot)=-\eta(\cdot,\cdot),
\end{align*}
 and the correspondence with $\KK_\pm$ is given explicitly by
\begin{align*}
\KK_{\pm}=\frac{1}{2}
\begin{pmatrix}
\id & 0 \\
b & \id
\end{pmatrix}
\begin{pmatrix}
K_+\pm K_- & \omega^{-1}_+\mp \omega^{-1}_- \\
\omega_+\mp \omega_- & -(K_+^*\pm K_-^*)
\end{pmatrix}
\begin{pmatrix}
\id & 0 \\
-b & \id
\end{pmatrix},
\end{align*}
where $\omega_\pm := \eta K_\pm$. Observe that the contractions $\omega_\pm$ are 2-forms on $M$, called the {\em fundamental forms} of the para-Hermitian structures $(\eta,K_\pm)$,
because $K_\pm$ are anti-isometries of $\eta$.
Moreover, as expected from the GK case, we show in Section \ref{sec:GpK} (Theorem \ref{theo:integr_GpK}) that the generalized para-complex structures $\KK_\pm$ are integrable if and only if the following hold:
\begin{enumerate}
\item $K_\pm$ are integrable para-complex structures;
\item $\rd^p_\pm \omega_\pm=\pm H$ for some closed 3-form $H \in \Gamma(\Lambda^3T^*)$, where $\rd^p_\pm=K_\pm^*\circ \rd \circ K_\pm^*$.
\end{enumerate}
The special case of GpK geometry when $K_\pm$ anti-commute gives rise this time to para-hyperHermitian structures (consisting of para-hypercomplex triples
$I := K_+K_-,K_+,K_-$ whose compatibility with $\eta$ is
\begin{align}\label{intro:PHK}
\eta(I\cdot,I\cdot)=-\eta(K_\pm\cdot,K_\pm\cdot)=\eta(\cdot,\cdot);
\end{align}
when $H = 0$, these correspond to para-hyperK\"ahler structures.

Finally, we also show in Section \ref{sec:GpK} (Theorem \ref{theo:bi-paraH_Poisson}) that any GpK structure $\KK_\pm$ with associate {\em bi-para-Hermitian} data $(\eta,K_\pm)$ induces a {\it para-holomorphic Poisson bi-vector} $Q$ on $M$ defined by the following expression
\begin{align*}
Q=\frac{1}{2}[K_+,K_-]\eta^{-1}.
\end{align*}
This is the GpK analog of a result of Hitchin's that relates GK structures to holomorphic Poisson structures \cite{Hitchin2006}.

\bigskip\noindent
\textbf{Generalized chiral geometry.}
In the case of generalized chiral geometry, we again obtain a pair of -- now {\em chiral} -- endomorphisms $J_\pm$ of $T$:
\begin{align*}
J_\pm \in \End(T), \quad J_\pm^2=\id,\quad \HH(J_\pm\cdot,J_\pm\cdot)=\HH(\cdot,\cdot),
\end{align*}
where $\HH$ is a pseudo-Riemannian metric, such that the corresponding generalized product structures are
\begin{align}\label{intro:Jpm}
\JJ_{\pm}=\frac{1}{2}
\begin{pmatrix}
\id & 0 \\
b & \id
\end{pmatrix}
\begin{pmatrix}
J_+\pm J_- & \eta^{-1}_+\mp \eta^{-1}_- \\
\eta_+\mp \eta_- & J_+^*\pm J_-^*
\end{pmatrix}
\begin{pmatrix}
\id & 0 \\
-b & \id
\end{pmatrix},
\end{align}
with $\eta_\pm=\HH J_\pm$. Moreover, the contractions $\eta_\pm$ are now pseudo-Riemannian metrics on $M$ because $J\pm$ are isometries of $\HH$.
Finally, we show in Section \ref{sec:bismut_integr} (Theorem \ref{theo:integrabilityGChGaK}) that the generalized chiral structures $\JJ_\pm$ are integrable if and only if $J_\pm$ are both parallel with respect to the Levi-Civita connection of $\HH$.

Our most important observation concerning generalized chiral geometry is that it contains, as a subcase, Born geometry. Indeed, if a generalized chiral structure $\JJ_\pm$ is given by the pair of {\em anti-commuting} chiral tangent bundle endomorphisms $J_\pm$,  
then their product $I=J_+J_-$ is an almost complex structure on $M$ whose compatibility with the metric $\HH$ is 
\begin{align}\label{intro:Born}
\HH(I\cdot,I\cdot)=\HH(J_\pm\cdot,J_\pm\cdot)=\HH(\cdot,\cdot).
\end{align}
Usually, Born geometry is thought of as a pair of metrics $(\eta,\HH)$ and a compatible two-form $\omega$, such that
\begin{align*}
\eta^{-1}\HH=\HH^{-1}\eta,\quad \omega^{-1}\HH=-\HH^{-1}\omega .
\end{align*}
This picture is equivalent to the above one upon setting $\eta=\HH J_+$ and $\omega=\HH I$. Note that the generalized chiral structure $\JJ_\pm$ correspond to Born structures of ``hyperK\"ahler-type" when they are integrable because $(\eta, I)$ is anti-K\"ahler in this case.

\bigskip\noindent
\textbf{Generalized anti-K\"ahler geometry.}
Lastly, a generalized anti-K\"ahler structure induces a pair anti-Hermitian endomorphisms $J_\pm$ of $T$:
\begin{align*}
J_\pm \in \End(T), \quad J_\pm^2=-\id,\quad \eta(J_\pm\cdot,J_\pm\cdot)=-\eta(\cdot,\cdot),
\end{align*}
where $\eta$ is a pseudo-Riemannian structure. The corresponding generalized anti-complex structures are now

\begin{align}\label{intro:Jpm-anti}
\JJ_{\pm}=\frac{1}{2}
\begin{pmatrix}
\id & 0 \\
b & \id
\end{pmatrix}
\begin{pmatrix}
J_+\pm J_- & -(\HH^{-1}_+\mp \HH^{-1}_-) \\
\HH_+\mp \HH_- & J_+^*\pm J_-^*
\end{pmatrix}
\begin{pmatrix}
\id & 0 \\
-b & \id
\end{pmatrix},
\end{align}
with $\HH_\pm =\eta J_\pm$. 
As for generalized chiral structures, the tensors $\HH_\pm$ are again pseudo-Riemannian metrics on $M$ since $J_\pm$ are anti-isometries of $\HH$,
and the generalized anti-K\"ahler structures $\JJ_\pm$ are integrable if and only if $J_\pm$ are both parallel with respect to the Levi-Civita connection of $\eta$ (see Section \ref{sec:bismut_integr} (Theorem \ref{theo:integrabilityGChGaK})). Finally, if the pair of anti-Hermitian tangent bundle endomorphisms $J_\pm$ corresponding to a GaK structure $\JJ_\pm$ {\em anti-commute}, then their product $I=J_+J_-$ is again an almost complex structure whose compatibility with the metric $\eta$ is 
\begin{align*}
\eta(I\cdot,I\cdot)=-\eta(J_\pm\cdot,J_\pm\cdot)=\eta(\cdot,\cdot).
\end{align*}
In other words, $(\eta, I, J_+,J_-)$ is an anti-hyperHermitian structure, which is anti-hyperK\"ahler when $\JJ_\pm$ is integrable.

\smallskip
To summarize, the anti-commuting cases of GK, GpK, GCh and GaK geometries recover hyperHermitian, para-hyperHermitian, Born and anti-hyperHermitian geometries, respectively, as special cases:
\begin{align*}
(\II_+,\II_-) \text{ GK} &\Leftrightarrow (\eta,I_\pm) \text{ bi-Hermitian} &\xrightarrow{\{I_+,I_-\}=0} &\text{ hyperHermitian},\\
(\KK_+,\KK_-) \text{ GpK} &\Leftrightarrow (\eta,K_\pm) \text{ bi-para-Hermitian} &\xrightarrow{\{K_+,K_-\}=0} &\text{ para-hyperHermitian},\\
(\JJ_+,\JJ_-) \text{ GCh} &\Leftrightarrow (\HH,J_\pm) \text{ bi-chiral} &\xrightarrow{\{J_+,J_-\}=0} &\text{ Born},\\
(\II_+,\II_-) \text{ GaK} &\Leftrightarrow (\eta,I_\pm) \text{ bi-anti-Hermitian} &\xrightarrow{\{I_+,I_-\}=0} &\text{ anti-hyperHermitian}.
\end{align*}

\bigskip\noindent
\textbf{Integrability.}
In this paper, we also study the integrability of GpK, GCh and GaK structures. We first note that, as for GK structures, GpK structures correspond to certain Dirac subbundles of $\TT$. The integrability of such structures can thus also be defined as closure of the Dirac subbundles under the Dorfman bracket. Unfortunately, GCh and GaK do not correspond to Dirac subbundle of $\TT$. One therefore needs another notion of integrability in this case.
Nonetheless, for GK structures, it was shown in \cite{Gualtieri:2007bq} that the integrability of such structures can be characterised in terms of generalized connections.
Indeed, to any generalized metric $\GG$, one can associate a canonical generalized connection $D^\GG$, called the {\it generalized Bismut connection}, and the integrability of a GK structure $(\II_+,\II_-)$ with $\GG=-\II_+\II_-$ is then equivalent to:
\begin{enumerate}
\item $D^\GG\II_\pm=0$;
\item The generalized torsion of $D^\GG$ is of type $(2,1)\!+\!(1,2)$ with respect to $\II_\pm$.
\end{enumerate}
We prove in Section \ref{sec:bismut_integr} (Theorem \ref{theo:integrabilityGpK}) that the same result also holds in the GpK case. We then propose the following notion of integrability for GCh and GaK structures:
\begin{itemize}[label={},leftmargin=0pt]
\item Let $(\JJ_+,\JJ_-)$ be a generalized chiral or anti-K\"ahler structure with $\GG=\JJ_+\JJ_-$. We then define $(\JJ_+,\JJ_-)$ to be {\em integrable} if and only if the following holds:
\begin{enumerate}
\item $D^\GG\JJ_\pm=0$;
\item The generalized torsion of $D^\GG$ is of type $(3,0)\!+\!(0,3)$ with respect to $\JJ_\pm$.
\end{enumerate}
\end{itemize}
Note that the different condition on the generalized torsion of $D^\GG$ (that it be of type $(3,0)\!+\!(0,3)$ instead of $(2,1)\!+\!(1,2)$ with respect to $\JJ_\pm$) was chosen to ensure, as in the GK/GpK case, that the  bi-chiral/bi-anti-Hermitian structures $J_\pm$ associated to $\JJ_\pm$ are integrable whenever $\JJ_\pm$ are. In fact, we prove in Section \ref{sec:bismut_integr} (Theorem (\ref{theo:integrabilityGChGaK}) that a GCh/GaK structure $\JJ_\pm$ is integrable if and only its associated tangent bundle endomorphisms $J_\pm$ are parallel with respect to the Levi-Civita connection of $\eta$. 

Finally, for a commuting pair of generalized structures $\Aa_\pm$ with $\GG=\Aa_+\Aa_-$, we also define {\it weak integrability} to be given by the condition $D^\GG\Aa_\pm=0$. We then prove that a GCh/GaK structure is weakly integrable if and only if its induced tangent bundle structures $(\eta,J_\pm)$ are of type $\mathcal{W}_3$ (see Proposition \ref{prop:weak_int_chiral}).

\bigskip\noindent
\textbf{Applications to physics.}
There are numerous physical applications of the results presented in this paper. The bi-para-Hermitian geometry of a GpK structure is precisely the geometry described in \cite{Hull_twisted_susy} as the target space geometry for $(2,2)$ {\it twisted supersymmetry}\footnote{The name {\it twisted} is in \cite{Hull_twisted_susy} not used in the sense of {\it topological twisting} but rather describing something different or opposite to the usual supersymmetry.}. Therefore, GpK geometry corresponds to twisted supersymmetry exactly in the same way that GK geometry corresponds to usual supersymmetry. In this context, the notion of weak integrability naturally appears as well; the case where the corresponding tangent bundle geometries are not necessarily integrable was explored for example in \cite{SUSY_non-integrable}. Sigma models with twisted supersymmetry and their relationship to GpK geometry, in particular topological twists of such models, will be explored in the forthcoming work \cite{toappear}.

Pairs of para-Hermitian structures also appear in \cite{Svoboda:2018rci,Szabo-paraherm} as a way to incorporate fluxes in the para-Hermitian formalism of Double Field Theory. Here, the non-integrability of the para-Hermitian structures is in fact a desirable feature, giving rise to non-trivial {\it fluxes}. In the related setting of Poisson-Lie sigma models, the non-integrable para-Hermitian structures were also studied in \cite{Hassler:2019wvn}. The formulation in terms of weakly integrable GpK geometry might therefore yield interesting results in these areas of physics as well.

The bi-chiral geometry of a generalized chiral structure has also been studied in the physics literature \cite{Stojevic:2009ub} in the context of $2D$ sigma models. There, the chiral structures $(g,J_\pm)$ give rise to copies of the $(1,1)$ superconformal algebra labelled by $J_\pm$ and by considering the cases when $J_\pm$ are not integrable, the author relates this observation to non-geometric backgrounds of string theory. This non-integrable geometry, $(g,J_\pm)$, then corresponds exactly to a weakly integrable generalized chiral structure.

%

\vspace{15pt}
The paper is organised as follows. In Section 2, we review mostly well established definitions and facts about generalized geometry, particularly, generalized structures of both isotropic and non-isotropic type, generalized metrics and generalized Bismut connections.

Section 3 is the main body of this work. We begin the section by outlining the correspondence between commuting pairs of generalized structures and certain pairs of the tangent bundles structures. We then describe in greater detail the geometry of the new generalized sructures, namely, the generalized para-K\"ahler, chiral and anti-K\"ahler structures.

In Section 4, we describe applications of our results to physics. We conclude the paper with proposed future research problems in both mathematics and physics in Section 5.

The Appendices provide all the results on the geometry of tangent bundle geometry needed in the paper; we should note that they contain several non-trivial results that we did not find in the literature (particularly, results about chiral and anti-Hermitian structures (Appendix \ref{sec:chiral/antiH})).

\section{Generalized Geometry}
Here we will review facts and definitions about generalized geometry, by which we mean the study of the {\it generalized tangent bundle} $\TT$ and geometric structures defined on this bundle. This involves in particular the {\it Dirac geometry}, which studies the bundle $\TT$ itself and its natural Courant algebroid structure, and the {\it generalized structures}, which are endomorphisms on $\TT$ compatible with the underlying Dirac geometry. A special non-degenerate type of generalized structures give rise to {\it generalized metrics}, which we will also discuss along with interesting connections and bracket operations they induce on $\TT$.

\subsection{Review of Dirac Geometry}
The natural Courant algebroid structure \cite{Liu:1995lsa} on $\TT$ is given by the following data. The symmetric pairing,
\begin{align*}
\langle X+\ap,Y+\bt\rangle=\ap(Y)+\bt(Y),
\end{align*}
the \textbf{Dorfman bracket}
\begin{align}\label{eq:dorfman}
[ X+\ap,Y+\bt]=[X,Y]+\Lie_X\bt-\imath_Y\rd \ap, 
\end{align}
and the anchor $\pi:X+\ap\mapsto X$. The three structures are compatible in the following way
\begin{align}\label{eq:compatibility_courant_alg}
\pi(X+\ap)\la Y+\bt,Z+\gamma\ra=\la [X+\alpha,Y+\bt],Z+\gamma\ra+\la Y+\bt,[X+\ap,Z+\gamma]\ra.
\end{align}

In the above, $X+\ap$ denotes a section of $\TT$ with the splitting to tangent and cotangent parts given explicitly. The Dorfman bracket can be thought of as an extension of the Lie bracket from $T$ to $\TT$ and therefore we opt to use the same notation for both brackets; the expression $[X,Y]$ is always the Lie bracket of vector fields whether we think of $\brac$ as the Lie bracket or the Dorfman bracket and no confusion is therefore possible.

\begin{remark}
The Courant algebroid structure can be equivalently given by the {\bf Courant bracket}, which is just a skew-symmetrization of $\brac$:
\begin{align*}
[X+\ap,Y+\bt]_{\text{Cour.}}&=\frac{1}{2}( [X+\ap,Y+\bt]-[Y+\bt,X+\ap])\\
&=[X,Y]+\Lie_X\bt-\Lie_Y\ap -\frac{1}{2}\rd(\imath_X\bt-\imath_Y\ap).
\end{align*}
The inverse relationship is given by
\begin{align*}
[X+\ap,Y+\bt]=[X+\ap,Y+\bt]_{\text{Cour.}}+\rd\la X+\ap,Y+\bt\ra.
\end{align*}
While $\brac_{\text{Cour.}}$ is conveniently skew-symmetric, it does not satisfy the Jacobi identity, which $\brac$ does. Instead, the Jacobi identity of $\brac_{\text{Cour.}}$ is violated by an exact non-vanishing $3$-product, which is why Courant algebroids are Lie 2-algebroids (or, Lie algebroids up to homotopy).
\end{remark}

The Courant algebroid on $\TT$ is exact, meaning that the associated sequence of vector bundles
\begin{align}\label{eq:exact_seq}
0\longrightarrow T^* \xrightarrow{\pi^T} \TT\xrightarrow{\pi} T\longrightarrow 0,
\end{align}
is exact. Here, $\pi^T$ is the transpose of $\pi$ with respect to the pairing $\lara$,
\begin{align*}
\la \pi^T(\ap),Y+\bt\ra=\la \ap,\pi(Y+\bt)\ra=\la \ap,Y\ra
\end{align*}
i.e. $\pi^T: \ap \mapsto \ap+0$. In fact, all possible Courant algebroid structures on $\TT$ are parametrized by a closed three-form $H\in \Omega^3_{cl}$ \cite{Severa:2017oew}, sometimes called $H$-flux\footnote{Flux is a term used mainly in physics, in this context simply meaning the ``tensorial contribution to the bracket''.} or \v Severa class, which enters the definition of the bracket \eqref{eq:dorfman}, changing it to a \textbf{twisted} Dorfman bracket
\begin{align}\label{eq:twisted_dorfman}
[ X+\ap,Y+\bt]_H=[X,Y]+\Lie_X\bt-\imath_Y\rd \ap+\imath_Y\imath_X H.
\end{align}

\begin{remark}
In the following text we tend to omit the word {\it twisted} and it should be assumed we mean ``twisted Dorfman bracket" whenever we say only ``Dorfman bracket" unless specified otherwise.
\end{remark}

\paragraph{b-field transformation.}
Any isotropic splitting of \eqref{eq:exact_seq} $s:T\rightarrow \TT$ is given by a two-form $b$, such that $X\overset{s}{\mapsto}X+b(X)$. This is equivalent to an action of a $b$-field transformation on $\TT$\footnote{Here we are using the term $b$-field transformation more liberally as it is customary to use the term only in the cases when $\rd b=0$ so that $e^b$ is a symmetry of $\brac$.}
\begin{Def}
Let $b$ be an arbitrary two-form. A {\bf b-field transformation} is an endomorphism of $\TT$ given by
\begin{align}\label{eq:b-field}
\begin{aligned}
e^b&=
\begin{pmatrix}
\id & 0 \\
b & \id
\end{pmatrix}
\in \End(\TT)\\
u&=X+\ap \mapsto e^b(u)=X+b(X)+\ap
\end{aligned}
\end{align}
\end{Def}
The map $e^b$ satisfies $\la e^b\cdot,e^b\cdot\ra=\la\cdot,\cdot\ra$ and acts on the (twisted) Dorfman bracket as
\begin{align}\label{eq:b_field-bracket}
[ e^b(X+\ap),e^b(Y+\bt)]_H=e^b([ X+\ap,Y+\bt]_{H+\rd b}),
\end{align}
which implies that when $H$ is trivial in cohomology, then a choice of a $b$-field transformation such that $\rd b=-H$ brings the twisted bracket $\brac_H$ into the standard form \eqref{eq:dorfman}. When $H$ is cohomologically non-trivial this can be done at least locally. This also means that any choice of splitting with a non-trivial $b$-field can be absorbed into the Dorfman bracket in terms of the flux $\rd b$.

We remark here that all the results in this paper remain valid for any exact courant algebroid $E$ (i.e. $E$ fits in the sequence \eqref{eq:exact_seq}), which can be always identified with $\TT$ by the choice of splitting equivalent to a choice of a representative $H\in\Omega^3_{cl}$. This also amounts to setting $b=0$ in all formulas since the $b$-field appears as a difference of two splittings.

\paragraph{Dirac Structures.}
An important object in Dirac geometry are (almost) dirac structures, which are subbundles $L\subset \TT$ with special properties.
\begin{Def}
An {\bf almost Dirac structure} $L$ is a maximally isotropic subbundle of $\TT$, i.e. $\la u,v\ra=0$ for any $u,v \in \se(L)$ and $\text{rank}(L)=\text{rank}(T)$. When $L$ is involutive under the Dorfman bracket, i.e. it satisfies $[L,L]\subset L$, we call $L$ simply a {\bf Dirac structure}.
\end{Def}
An important fact we will repeatedly use is that the Dorfman bracket becomes fully skew when restricted to sections of a Dirac structure $L$ and in particular becomes a Lie algebroid bracket. $L$ then inherits a Lie algebroid structure given by $(\brac\mid_L,\pi_T)$, $\pi_T$ being the projection to the tangent bundle $T$. More details about Dirac structures can be found in \cite{Courant-Weinstein,Dorfman,courant1990dirac}.

We conclude this section with a useful formula for $\brac$ \cite[Prop.~2.7]{Svoboda:2018rci}
\begin{align}\label{eq:dorf-connection}
\begin{aligned}
\la [ X+\ap,Y+\bt],Z+\gamma\ra&=\la \n_X(Y+\bt)-\n_Y(X+\ap),Z+\gamma\ra\\
&+\la \n_Z(X+\ap),Y+\bt\ra,
\end{aligned}
\end{align}
where $\n$ is any torsionless connection.

\subsection{Generalized Structures}
\label{sec:gen_structures}
We continue by introducing generalized structures, i.e. endomorphisms of the generalized tangent bundle $\Aa \in \End(\TT)$ that square to $\pm 1$ and are (anti-)orthogonal with respect to the natural pairing $\lara$ on $\TT$. This involves four different choices:

\begin{Def}\label{def:gen_structures}
An endomorphism $\Aa \in \End(\TT)$ that satisfies $\Aa^2=\id$ or $\Aa^2=-\id$ and in addition $\la\Aa\cdot,\Aa\cdot\ra=\lara$ or $\la\Aa\cdot,\Aa\cdot\ra=-\lara$ is called a {\bf generalized almost structure}. We name the four different types of generalized almost structures:
\begin{itemize}
\item {\bf complex}, when $\Aa^2=-\id$ and $\la \Aa\cdot,\Aa\cdot\ra=\lara$,
\item {\bf para-complex}, when $\Aa^2=\id$ and $\la \Aa\cdot,\Aa\cdot\ra=-\lara$,
\item {\bf product}, when $\Aa^2=\id$ and $\la \Aa\cdot,\Aa\cdot\ra=\lara$, and
\item {\bf anti-complex}, when $\Aa^2=-\id$ and $\la \Aa\cdot,\Aa\cdot\ra=-\lara$.
\end{itemize}
Additionally, $\Aa$ is {\bf isotropic} when it is complex or para-complex and {\bf non-isotropic} when it is product or anti-complex. Whenever the eigenbundles of $\Aa$ are isomorphic to $T$ or $T\otimes \mathbb{C}$ (for $\Aa^2=\id$ and for $\Aa^2=-\id$, respectively) via $\pi$, we call $\Aa$ {\bf non-degenerate}.
\end{Def}
\paragraph{Integrability.}
From Definition \ref{def:gen_structures} it follows that isotropic structures have maximally isotropic eigenbundles with respect to the pairing $\lara$ (i.e. almost Dirac structures), while the eigenbundles of non-isotropic structures are not isotropic. Because the Dorfman bracket restricts on almost Dirac structures to a Lie algebroid bracket, it makes sense to ask for involutivity of such bundles:
\begin{Def}
Let $\Aa$ be an isotropic generalized almost structure. We say $\Aa$ is {\bf integrable} generalized structure or simply generalized structure if its eigenbundles are involutive under the Dorfman bracket. To emphasize integrability with respect to a Dorfman bracket with non-vanishing $H$-flux, we call an integrable generalized structure a {\bf twisted} generalized structure.
\end{Def}
It is customary to omit the word almost whenever the integrability is not relevant in the given context and we will do so when discussing only the linear structure, i.e. the generalized endomorphism itself without considering the Courant algebroid structure on $\TT$. Analogously to the usual tangent bundle geometry, the integrability can be equivalently expressed in terms of a tensorial quantity called the {\bf generalized Nijenhuis tensor}:
\begin{lemma}\label{lem:integrability}
An isotropic generalized almost structure is integrable if an only if the following expression
\begin{align}\label{eq:gen_nijenhuis}
\mathcal{N}_\Aa(u,v)=[\Aa u,\Aa v]+\Aa^2[ u,v]-\Aa([\Aa u,v]+[ u,\Aa v]),
\end{align}
vanishes for all $u,v \in \se(\TT)$.
\end{lemma}
\begin{proof}
A short calculation shows that the $\mathcal{N}_\Aa$ can be expressed as
\begin{align*}
\mathcal{N}_\Aa(u,v)=4(\PS[\bar{\PS}u,\bar{\PS}v]+\bar{\PS}[\PS u,\PS v]),\ \PS=\frac{1}{2}(\id-i\Aa),\ \bar{\PS}=\frac{1}{2}(\id+i\Aa),
\end{align*}
when $\Aa$ is a generalized almost complex structure and
\begin{align*}
\mathcal{N}_\Aa(u,v)=4(\PS[\tl{\PS}u,\tl{\PS}v]+\tl{\PS}[\PS u,\PS v]),\ \PS=\frac{1}{2}(\id+\Aa),\ \tl{\PS}=\frac{1}{2}(\id-\Aa),
\end{align*}
when $\Aa$ is a generalized almost para-complex structure. Vanishing of $\mathcal{N}_\Aa$ is then in both cases seen to be equivalent to requiring that $[u,v]$ belongs to a given eigenbundle whenever both $u$ and $v$ lie in that eigenbundle.
\end{proof}

\begin{remark}\label{rem:integrability_non-iso}
Integrability is not well defined for the generalized product and generalized anti-complex structures because their eigenbundles are not isotropic with respect to the pairing $\lara$ and as a result the involutivity under the Dorfman bracket is not well-defined. This can be seen from the fact that the expression \eqref{eq:gen_nijenhuis} is not tensorial for such structures. We will tackle this issue in Section \ref{sec:bismut_integr}, where we define a notion of integrability that is applicable to non-isotropic structure as well.
\end{remark}

\paragraph{Action of b-field transformation.}
The $b$-field transformation \eqref{eq:b-field} induces an action  on endomorphisms of $\TT$ by:
\begin{align*}
e^b: \End(\TT)&\rightarrow \End(\TT)\\
\Aa &\mapsto e^b(\Aa)=e^b\circ \Aa\circ e^{-b}.
\end{align*}
The properties of $e^b$ then ensure that it preserves the type of a generalized structure:
\begin{proposition}
The $b$-field transformation preserves the type of a generalized almost structure $\Aa$ for any two-form $b$. This means that if $\Aa^2=\pm \id$, then $[e^b(\Aa)]^2=\pm \id$ and if $\la \Aa\cdot,\Aa\cdot\ra=\pm \la \cdot,\cdot\ra$, then also $\la e^b(\Aa),e^b(\Aa) \ra=\pm \la \cdot,\cdot\ra$. Additionally, if $\rd b=0$, $e^b$ also preserves the integrability of an isotropic structure $\Aa$. 
\end{proposition}
\begin{proof}
The fact that $e^b$ preserves type is straightforward to check:
\begin{align*}
e^b(\Aa)e^b(\Aa)&=e^b \Aa e^{-b} e^b \Aa e^{-b}=e^b(\Aa^2)\\
\la e^b(\Aa)\cdot,e^b(\Aa)\cdot\ra&=\la e^b \Aa e^{-b}\cdot,e^b \Aa e^{-b}\cdot\ra=\la \Aa e^{-b}\cdot,\Aa e^{-b}\cdot\ra=\pm \la e^{-b},e^{-b}\ra=\pm \la\cdot,\cdot\ra\\
&=\la \Aa\cdot,\Aa\cdot\ra.
\end{align*}
We now prove the statement about the integrability for $\Aa$ a GpC structure, for GC structures the proof is analogous except the appearing bundles are complexified. Let now $\Aa$ be integrable and $u,v\in \se(\TT)$ be $+1$ eigenvectors of $\Aa$. Then $e^b(u)$ and $e^b(v)$ are $+1$ eigenvectors of $e^b(\Aa)$. Using \eqref{eq:b_field-bracket} and $\rd b=0$:
\begin{align*}
[e^b(u),e^b(v)]_H=e^b[u,v]_H,
\end{align*}
so that the $+1$ eigenbundle of $e^b(\Aa)$ is involutive. Similar argument shows involutivity of the $-1$ eigenbundle of $e^b(\Aa)$.  
\end{proof}

\paragraph{Notation.}
We use $\Aa$ to denote a generic generalized structures, while $\II$, $\KK$ and $\JJ$ will be used for generalized complex, generalized para-complex and generalized product structures, respectively. The reason for this is that via the construction presented in Section \ref{sec:commuting_pairs}, GC structures are related to usual complex structures which we denote by $I$, GpC to para-complex structures that we denote by $K$ and generalized product structures are related to chiral structures which we denote by $J$.

\subsection{Isotropic structures}

\subsubsection{Generalized complex structures}
We start by briefly reviewing the very well known generalized complex geometry. For further details we refer the reader to the seminal work of Gualtieri who introduced this geometry in his thesis \cite{Gualtieri:2003dx}.

\begin{Def}
A \textbf{generalized complex} (GC) structure $\II$ is an endomorphism of $\TT$, such that $\II^2=-\id$ and $\la\II\cdot,\II\cdot\ra=\la\cdot,\cdot\ra$, whose generalized Nijenhuis tensor \eqref{eq:gen_nijenhuis} vanishes.
\end{Def}
As discussed in Section \ref{sec:gen_structures}, we use the name almost whenever we want to emphasize that integrability of $\II$ is not concerned. The most general for of GC structures is the following
\begin{align*}
 \II =
 \begin{pmatrix}
 A & \Pi \\
 \Omega & -A^*
 \end{pmatrix}, \text{ such that }
 \begin{cases}
  A^2 + \Pi \Omega & = -\id\\
  A\Pi - \Pi A^* & = 0\\
  \Omega A + A^*\Omega & = 0
 \end{cases},
\end{align*} 
where $A \in \End(T)$ and $\Omega \in \Omega^2(M)$, $\Pi \in \se(\Lambda^2 T)$ are skew tensors.

An important fact is that not only does any GC structure give rise to a complex Dirac structure as its eigenbundle, but any complex Dirac structure $\Lb$ satisfying $\Lb\oplus \overline{\Lb}=(\TT)\otimes {\mathbb{C}}$ defines a GC structure $\II$. Indeed, take $\II\mid_\Lb=i\id$ and $\II\mid_{\overline{\Lb}}=-i\id$. Clearly, such $\II$ is an almost GC structure because it satisfies $\II^2=-\id$, $\la\II\cdot,\II\cdot\ra=\la\cdot,\cdot\ra$ and its eigenbundles are involutive by assumption.

The most important examples are given by a complex structure $I$ and a symplectic structure $\omega$:
\begin{Ex}\label{ex:gen_cpx}
The diagonal almost GC structure is given by an almost complex structure $I$
\begin{align*}
\II_I=
\begin{pmatrix}
I & 0 \\
0 & -I^*
\end{pmatrix}.
\end{align*}
Its $\pm i$ eigenbundles are $\Lb=T^{(1,0)}\oplus T^{*(0,1)}$ and $\overline{\Lb}=T^{(0,1)}\oplus T^{*(1,0)}$. $\II_I$ is integrable if and only if $I$ is integrable.

The anti-diagonal almost GC structure is given by a non-degenerate two-form $\omega$
\begin{align*}
\II_\omega=
\begin{pmatrix}
0 & \omega^{-1} \\
-\omega & 0
\end{pmatrix},
\end{align*}
with eigenbundles $\Lb/\overline{\Lb}=\text{graph}(\pm i\omega)=\{X\pm i \omega(X)\mid X\in \XX\otimes \mathbb{C}\}$ and is integrable if and only if $\omega$ is symplectic, $\rd \omega=0$.
\end{Ex}

\subsubsection{Generalized para-complex structures}
In \cite{wade2004dirac,Zabzine:2006uz}, the notion of generalized para-complex (GpC) geometry along with basic integrability conditions and examples was introduced. All basic facts, which are fairly analogous to generalized complex (GC) geometry, will be reviewed here.

\begin{Def}
A \textbf{generalized para-complex} (GpC) structure $\KK$ is an endomorphism of $\TT$, such that $\KK^2=\id$ and $\la\KK\cdot,\KK\cdot\ra=-\la\cdot,\cdot\ra$, whose generalized Nijenhuis tensor \ref{eq:gen_nijenhuis} vanishes.
\end{Def}
As discussed in Section \ref{sec:gen_structures}, we use the name almost whenever we want to emphasize that integrability of $\KK$ is not concerned. The most general form of an almost GpC structure is given by

\begin{align}\label{eq:GpC_generalform}
 \KK =
 \begin{pmatrix}
 A & \Pi \\
 \Omega & -A^*
 \end{pmatrix}, \text{ such that }
 \begin{cases}
  A^2 + \Pi \Omega & = \id\\
  A\Pi - \Pi A^* & = 0\\
  \Omega A + A^*\Omega & = 0
 \end{cases}.
\end{align}
where $A \in \End(T)$ and $\Omega \in \Omega^2(M)$, $\Pi \in \se(\Lambda^2 T)$ are skew tensors.

Denote $\Lb$ and $\widetilde{\Lb}$ the $+1$ and $-1$ eigenbundles of $\KK$, respectively. It is clear that both $\Lb$ and $\widetilde{\Lb}$ are almost Dirac structures. Similarly to the complex case, we also have a correspondence between Dirac structures and generalized para-complex structures:

\begin{theorem*}[\cite{wade2004dirac}]\label{thm:pairofdirac}
There is a one-to-one correspondence between generalized para-complex structures on $M$ and pairs of transversal Dirac subbundles of $\TT$.
\end{theorem*}

Combining this result with the well-known result of \cite{Liu:1995lsa} which states that any pair of transversal Dirac structures $(L,\Lt)$ forms a Lie bialgebroid $(L,L^*\simeq \Lt)$, one can immediately infer the following

\begin{lemma*}
Generalized para-complex structures on $\TT$ are in one-to-one correspondence with Lie bialgebroid pairs $(L,L^*)$ such that $L\oplus L^*=\TT$. 
\end{lemma*}

Similarly to GC structures, when a GpC structure is integrable, the bivector $\Pi$ in \eqref{eq:GpC_generalform} needs to be a Poisson structure (for GC structures this is shown for example in \cite{Crainic:2004ic}):
\begin{lemma}\label{lem:GpC_poisson}
Let $\KK$ be a GpC structure given by \eqref{eq:GpC_generalform}. Then $\Pi$ is a Poisson bivector, i.e. its Schouten bracket with itself vanishes, $[\Pi,\Pi]=0$.
\end{lemma}
\begin{proof}
Let us evaluate \eqref{eq:gen_nijenhuis} on a pair of one-forms, i.e. $u=\ap$ and $v=\bt$:
\begin{align*}
\mathcal{N}_\KK(\ap,\bt)=[\Pi(\ap)-A^*(\ap),\Pi(\bt)-A^*(\bt)]-\KK([\Pi(\ap),\bt]+[\ap,\Pi(\bt)]),
\end{align*}
since $\brac$ vanishes on $1$-forms. Taking now the $1$-form part of the above and setting it to zero yields
\begin{align}\label{eq:proof_poisson}
\Pi([\ap,\bt]_\Pi)=[\Pi(\ap),\Pi(\bt)],
\end{align}
where $\brac_\Pi$ is the Poisson Lie algebroid bracket \cite{Courant_1994}
\begin{align*}
[\ap,\bt]_\Pi=\Lie_{\Pi(\ap)}\bt-\Lie_{\Pi(\bt)}\ap-\rd\Pi(\ap,\bt),
\end{align*}
and \eqref{eq:proof_poisson} is equivalent to $\Pi$ being a Poisson bivector.
\end{proof}

We also define the notion of a type for GpC structures

\begin{Def}
Let $\KK$ be a GpC structure on $M$. The {\it type} of $\KK$ at $x \in M$ is a pair $(l,\tl{l}) \in \mathbb{Z}^2$, with $0 \leq l, \tilde l \leq \dim M$, where $l$ (resp. $\tilde l$) is the type of the Dirac structure $\Lb$ (resp. $\widetilde{\Lb}$) at $x$. We say that $\KK$ is {\it non-degenerate} at $x$ if $l = \tilde l = \text{dim} (M)$.
\end{Def}

We now present main examples. More can be found in \cite{wade2004dirac}.

\begin{Ex}[The trivial structure and its deformations]\label{ex:GpC_trivial}
Any manifold supports the following GpC structure
\begin{align*}
\KK_0=
\begin{pmatrix}
\id & 0 \\
0 & -\id
\end{pmatrix},
\end{align*}
that has eigenbundles $T$ and $T^*$ and is always integrable. The following two GpC structures can be seen as deformations of $\KK_\omega$ by either a two-form $b$ or a bi-vector $\beta$:
\begin{align*}
\KK_b=
\begin{pmatrix}
\id & 0 \\
2b & -\id
\end{pmatrix},\quad
\KK_\beta=
\begin{pmatrix}
\id & 2\beta \\
0 & -\id
\end{pmatrix}.
\end{align*}
$\KK_b$ is integrable iff $\rd b =0$, i.e. $b$ is presymplectic and its eigenbundles are $\Lb_b=\text{graph}(b)=\{X+b(X)\mid X \in \XX\}$ and $\widetilde{\Lb}=T^*$. Similarly, $\KK_\beta$ is integrable iff $\beta$ is Poisson (by Lemma \ref{lem:GpC_poisson}) and its eigenbundles are $\Lb=T$ and $\widetilde{\Lb}=\text{graph}(-\beta)=\{\ap-\beta(\ap)\}\mid \ap \in \Omega\}$. The types of $\KK_0$ and $\KK_b$ are $(n,n)$ ($n$ being the dimension of the base manifold), while the type of $\KK_\beta$ can range anywhere between $(n,0)$ and $(n,n)$, depending on the rank of $\beta$ and its degeneracy, it can even change from point to point throughout the manifold.
\end{Ex}

\begin{Ex}[Product structures]
A product structure $J\in \Endo(T)$, defines the diagonal generalized para-complex structure:
\begin{align*}
\KK_J=
\begin{pmatrix}
J & 0 \\
0 & -J^*
\end{pmatrix}.
\end{align*}
The corresponding Dirac structures are given by $\Lb=T^{(1,0)}\oplus T^{*(0,1)}$ and $\widetilde{\Lb}=T^{(0,1)}\oplus T^{*(1,0)}$, where the bigrading is with respect to $J$. The integrability of $\KK_J$ is equivalent to Frobenius integrability of $J$ , i.e. vanishing of the Nijenhuis tensor of $J$. The type $(k,l)$ of $\KK_J$ is always such that $k+l=2n$; in particular, if $J$ is a para-complex structure, $\KK_J$ is of type $(n,n)$.
%

\end{Ex}

\begin{Ex}[Symplectic structures]\label{ex:GpC_sympl}
A symplectic form $\omega$ defines the anti-diagonal GpC structure
\begin{align*}
\KK_\omega=
\begin{pmatrix}
0 & \omega^{-1} \\
\omega & 0
\end{pmatrix}.
\end{align*}
The $\pm 1$ eigenbundles are given by $\text{graph}(\pm\omega)=\{X\pm\omega(X)\mid X\in \XX\}$, and the integrability of $\KK_\omega$ is equivalent to $\rd\omega=0$. Its type is $(n,n)$. This is an  example of a nondegenerate GpC structure, since both its eigenbundles are isomorphic to $T$ (as well as $T^*$).
\end{Ex}
\paragraph{Comparision with GC structures.}
Examples \ref{ex:GpC_sympl} and \ref{ex:gen_cpx} show that a symplectic manifold is both a GC and GpC manifold. However, while almost GC structures exist only on almost complex manifolds \cite{Gualtieri:2003dx}, Example \ref{ex:GpC_trivial} demonstrates that GpC structures exist on any Poisson manifold and in particular on any smooth manifold (with trivial Poisson structure). Another feature of GpC geometry that is not present in GC geometry is that the GpC structures can be half-integrable (similarly to the usual tangent bundle case, as explained in Appendix \ref{sec:paracpx/product}): see the cases of $\KK_b$ and $\KK_\beta$ from Example \ref{ex:GpC_trivial} which are always at least half integrable and are fully integrable iff $b$ is closed and $\beta$ is Poisson, respectively. On the other hand, $\KK_\omega$ in Example \ref{ex:GpC_sympl} does not have this property.


\subsection{Non-isotropic structures}

\subsubsection{Generalized product structures}
\begin{Def}
A {\bf generalized product} structure (GP) is an endomorphism $\JJ \in \End(\TT)$, such that $\JJ^2=\id$ and $\langle \JJ,\JJ\rangle=\langle\cdot,\cdot\rangle$.
\end{Def}
As discussed in Section \ref{sec:gen_structures}, for non-isotropic structures there is no straightforward notion of integrability and for this reason we typically do not use the labels almost/integrable in the case of non-isotropic generalized structures.

A general form of GP structures is the following
\begin{align}\label{eq:GP_genform}
 \JJ =
 \begin{pmatrix}
A & g \\
\sigma & A^*
\end{pmatrix}, \text{ such that }
 \begin{cases}
 A^2+g\sigma&=\id,\\
Ag+gA^*&=0,\\
\sigma A+A^*\sigma&=0,
 \end{cases}
\end{align} 
where $A \in \End(T)$ and $g \in \se(T \otimes T)$, $\sigma \in \se(T^*\otimes T^*)$ are symmetric tensors. The main examples are the following:
\begin{Ex}[Product structures]
Any (almost) product structure $J \in \End(T)$ defines a GP structure $\JJ_J\in \End(\TT)$ in the following way
\begin{align*}
\JJ_J=
\begin{pmatrix}
J & 0 \\
0 & J^*
\end{pmatrix},
\end{align*}
and the $\pm 1$ eigenbundles are $\Lb=T^{(1,0)}\oplus T^{*(1,0)}$ and $\widetilde{\Lb}=T^{(0,1)}\oplus T^{*(0,1)}$, where the bigrading is with respect to $J$.
\end{Ex}

\begin{Ex}(Pseudo-Riemannian structures)\label{ex:gen_metric}
Any (pseudo-)Riemannian structure $\eta$ defines a GP structure $\JJ_\eta$ in the following way
\begin{align*}
\JJ_\eta=
\begin{pmatrix}
0 & \eta^{-1} \\
\eta & 0
\end{pmatrix},
\end{align*}
and the $\pm 1$ eigenbundles are $\text{graph}(\pm\eta)\subset \TT$; $\JJ_\eta$ is therefore non-degenerate. Non-degenerate GP structures are called generalized metrics and will be discussed in detail in Section \ref{sec:gen_metric}
\end{Ex}

\subsubsection{Generalized anti-complex structures}
For completeness, we also review basic facts about generalized anti-complex structures, although they are very similar to generalized product structures
\begin{Def}
A {\bf generalized anti-complex} structure is an endomorphism $\Aa \in \End(\TT)$, such that $\Aa^2=-\id$ and $\langle \Aa,\Aa\rangle=-\langle\cdot,\cdot\rangle$.
\end{Def}
The general form is the same as \eqref{eq:GP_genform}, except for a sign change in the equations the blocks satisfy
\begin{align*}
 \Aa =
 \begin{pmatrix}
A & g \\
\sigma & A^*
\end{pmatrix}, \text{ such that }
 \begin{cases}
 A^2+g\sigma&=-\id,\\
Ag+gA^*&=0,\\
\sigma A+A^*\sigma&=0
 \end{cases}.
\end{align*} 
The examples are similar as for the GP structures, given by complex and pseudo-Riemannian structures:
\begin{Ex}
Any complex structure $I$ defines a generalized anti-complex structure by
\begin{align*}
\Aa_I=
\begin{pmatrix}
I & 0 \\
0 & I^*
\end{pmatrix},
\end{align*}
and any (pseudo-)Riemannian metric defines a generalized anti-complex structure by
\begin{align*}
\Aa_\eta=
\begin{pmatrix}
0 & -\eta^{-1} \\
\eta & 0
\end{pmatrix}.
\end{align*}
This is the non-degenerate type of a generalized anti-complex structure.
\end{Ex}

The nondegenerate generalized anti-complex structures have been previously studied for example by Vaisman in \cite{vaisman2015generalized}.

\subsection{Generalized metrics and related structures}\label{sec:gen_metric}
In this section we discuss generalized metrics, which are the non-degenerate generalized product structures; such structures are generically given by the $b$-field transformations of the structure $\JJ_\eta$ in Example \ref{ex:gen_metric}. We then also recall a definition and some properties of the generalized Bismut connection, which is a generalized connection on $\TT$ that one can naturally associate to any generalized metric.

\begin{Def}
A {\bf generalized (indefinite) metric} is a non-degenerate generalized product structure.
\end{Def}
\begin{remark}
In the following text, we denote generalized metric structures by $\GG$, to emphasize that we wish to think of them as metric tensors on $\TT$. Indeed, $\GG$ defines a metric (non-degenerate symmetric tensor) on $\TT$ by
\begin{align*}
h(u,v)=\la \GG u,v\ra.
\end{align*}
The name generalized metric is typically used when $h$ is positive-definite, but here we will use the term for indefinite metrics as well, emphasizing this fact by the name ``indefinite generalized metric" whenever necessary. We also remark here that the discussion below has been first presented for the positive definite case in \cite{Gualtieri:2003dx}.
\end{remark}

In Definition \ref{def:gen_structures} we defined a non-degenerate structure as a structure whose eigenbundles are isomorphic to the (complexified) tangent bundle. Let us now describe what non-degeneracy implies for the general form of generalized product structures \eqref{eq:GP_genform}. It is easy to show that for the GP structure to be non-degenerate, its upper right corner has to be an invertible map. Whenever this is the case, the system of equations in \eqref{eq:GP_genform} can be solved explicitly in terms of a pseudo-Riemannian metric $\eta\coloneqq g^{-1}$ and a two-form $b\coloneqq -\eta A$. The structure $\JJ$ is then simply a b-transform of $\JJ_\eta$ from Example \ref{ex:gen_metric}:
\begin{align} \label{eq:indefinite_metric}
\JJ\coloneqq \JJ(\eta,b)=e^b(\JJ_\eta)=
\begin{pmatrix}
\id & 0 \\
b & \id
\end{pmatrix}
\begin{pmatrix}
0 & \eta^{-1} \\
\eta & 0
\end{pmatrix}
\begin{pmatrix}
\id & 0 \\
-b & \id
\end{pmatrix}.
\end{align}
The eigenbundles of $\JJ$ are $C_\pm=\text{graph}(g\pm\eta)$ and therefore isomorphic to $T$. We denote the isomorphisms by $\pi_\pm$:
\begin{align}\label{eq:pi_iso}
\begin{aligned}
\pi_\pm: C_\pm&\xleftrightarrow{\simeq} T\\
X+\ap & \overset{\pi_\pm}{\longmapsto} X,\quad X+\ap \in \se(C_\pm)\\
X+ (b\pm \eta) X & \overset{\pi_\pm^{-1}}{\longmapsfrom} X, \quad X\in \se(T).
\end{aligned}
\end{align}

We also recall the following useful formula that recovers the metric $\eta$ from $\GG=\GG(\eta,b)$:
\begin{align}\label{eq:genmetric_metric}
\eta(X,Y)=\frac{1}{2}\la \GG\pi_\pm^{-1}X,\pi_\pm^{-1}Y\ra=\pm\frac{1}{2}\la \pi_\pm^{-1}X,\pi_\pm^{-1}Y\ra.
\end{align}

\paragraph{Generalized Bismut Connection.}
To any (indefinite) generalized metric, one can associate a generalized connection called generalized Bismut connection. This connection will play a central role in Section \ref{sec:bismut_integr}, where it will be used to define integrability of non-isotropic generalized structures.

We start the discussion by recalling the definition of a generalized Bismut connection from \cite{Gualtieri:2007bq}, extending it to indefinite metrics as well:
\begin{Def}
Let $\GG=\GG(\eta,b) \in \End(\TT)$ be a generalized (indefinite) metric and denote $C_\pm$ its eigenbundles. We split the sections $u \in \se(\TT)$ accordingly, $u=u_++u_-$. Then the following expression defines a \textbf{generalized connection parallelizing $\GG$}:
\begin{align}
D^H_uv=[u_-,v_+]_{H+}+[u_+,v_-]_{H-}+[Cu_-,v_-]_{H-}+[Cu_+,v_+]_{H+}.
\end{align}
Here $u,v \in \se(\TT)$, $[\ ,\ ]_H$ is the twisted Dorfman bracket and $C$ is the generalized almost para-Complex structure
\begin{align*}
C=
\begin{pmatrix}
\id & 0 \\
2b & -\id
\end{pmatrix}=
e^b \left(
\begin{pmatrix}
\id & 0 \\
0 & -\id
\end{pmatrix}
\right) \in \End(\TT),
\end{align*}
which maps $C_\pm\overset{C}{\mapsto} C_\mp$. 
\end{Def}
The generalized Bismut connection of $\GG$ is related to two ``usual" connections $\n^\pm$ via the isomorphisms $\pi_\pm$:
\begin{align}\label{eq:n_pm}
\begin{aligned}
D_uv&=\pi_+^{-1}\n_{\pi (u)}^+\pi_+ v_++\pi_-^{-1}\n_{\pi (u)}^-\pi_-v_-,\\
\n^\pm&=\lc\pm \frac{1}{2}\eta^{-1}H_b,
\end{aligned}
\end{align}
where $\lc$ is the Levi-Civita connection of $\eta$ in $\GG(\eta,b)$ and $H_b$ is the $H$-flux of the Courant algebroid with $b$ absorbed, $H_b=H+\rd b$.

The connections $\n^\pm$ appear in physics as the natural connections in the context of supersymmetry, particularly $(2,2)$ supersymmetry. The reason for this is that they parallelize the metric $\eta$ and have fully skew torsion equal to $T^{\n^\pm}=\pm H_b$. Another interesting property of the connections $\n^\pm$ is that they can be expressed as the restriction of the Dorfman bracket to the eigenbundles $C_\pm$ in the following way:

\begin{proposition}\label{prop:n_pm_and_bracket}
Let $\GG(\eta,b)$ be a generalized (indefinite) metric on a courant algebroid with flux $H$ and denote $H_b=H+\rd b$. The mixed projections of the $H_b$-twisted Dorfman bracket to $C_\pm$ then yield the connections $\n^\pm$
\begin{align}
\mp\frac{1}{2}\la [\pi^{-1}_\pm(X),\pi^{-1}_\mp(Y)],\pi^{-1}_\mp (Z)\ra=\eta (\n^\pm_XY,Z).
\end{align}
\end{proposition}
\begin{proof}
Using the formula \eqref{eq:dorf-connection} with $\n=\lc$, we expand
\begin{align*}
\la [\pi^{-1}_\pm(X),\pi^{-1}_\mp(Y)],\pi^{-1}_\mp (Z)\ra &= \la \lc_X\pi^{-1}_\mp(Y)-\lc_Y \pi^{-1}_\pm(X), \pi^{-1}_\mp(Z)\ra\\
&-\la \lc_Z\pi^{-1}_\pm(X),\pi^{-1}_\mp(Y)\ra+H(X,Y,Z) 
\end{align*}
Now, the individual terms can be expanded and further simplified
\begin{align*}
\la \lc_X\pi^{-1}_\mp(Y), \pi^{-1}_\mp(Z)\ra&=\la \lc_X(Y+b(Y)\mp\eta(Y)),Z+b(Z)\mp \eta(Z)\ra\\
&=\mp 2 \eta(\lc_XY,Z)+(\lc_Xb)(Y,Z),
\end{align*}
and carrying out similar calculation with the remaining two terms and using the formula
\begin{align}\label{eq:proof_n_db}
\sum_{Cycl.\ X,Y,Z}\left(\lc_Xb\right)(Y,Z)=(\rd b)(X,Y,Z),
\end{align}
then yields the result.
\end{proof}
Using the properties of the Dorfman bracket, it can be shown that the all other mixed components of the Dorfman bracket are also related to the connections $\n^\pm$. Interestingly, the pure components yield a bracket operation on the algebra of vector fields very closely related to the D-bracket of Double Field Theory as described in \cite{vaisman2012geometry,vaisman2013towards,Svoboda:2018rci,Szabo-paraherm}:

\begin{proposition}\label{prop:almost_D-brac}
Let $\GG(\eta,b)$ and $H_b$ be as in Proposition \ref{prop:n_pm_and_bracket}. The pure components of the $H_b$-twisted Dorfman bracket yield a bracket operation on the tangent bundle called the {\bf almost D-bracket} with a flux $\pm H_b$;
\begin{align}\label{eq:almost_D-brac}
\pm \frac{1}{2}\la [\pi^{-1}_\pm X,\pi^{-1}_\pm Y],\pi^{-1}_\pm Z\ra=\eta(\bl X,Y\br^{\lc},Z)\pm \frac{1}{2}H_b(X,Y,Z),
\end{align}
where $\lc$ denotes the Levi-Civita connection of $\eta$ and $\bracd^{\lc}$ is defined by
\begin{align*}
\eta(\bl X,Y\br^{\lc},Z)=\eta(\lc_XY-\lc_YX,Z)+\eta(\lc_ZX,Y).
\end{align*}
\end{proposition}
\begin{proof}
We use the same strategy as in the proof of the Proposition \ref{prop:n_pm_and_bracket}. Using \eqref{eq:dorf-connection}, we get
\begin{align*}
\la [\pi^{-1}_\pm X,\pi^{-1}_\pm Y],\pi^{-1}_\pm Z\ra&=\la \lc_X\pi^{-1}_\pm(Y)-\lc_Y\pi^{-1}_\pm (X),\pi^{-1}_\pm (Z)\ra\\
&+\la \lc_Z\pi^{-1}_\pm (X),\pi^{-1}_\pm (Y)\ra+H(X,Y,Z),
\end{align*}
which after expanding, simplifying and again making use of\eqref{eq:proof_n_db} gives
\begin{align*}
\pm \frac{1}{2}\la [\pi^{-1}_\pm X,\pi^{-1}_\pm Y],\pi^{-1}_\pm Z\ra&=\eta(\lc_XY-\lc_YX,Z)+\eta(\lc_ZX,Y)\\
&\pm \frac{1}{2}(H+ \rd b)(X,Y,Z).
\end{align*}
\end{proof}
\begin{remark}
The idea of obtaining the D-bracket via the projections $\pi_\pm$ has been used in \cite{Chatzistavrakidis:2018ztm}, where the authors use the generalized metric $\GG(\eta,b=0)$ and $H=0$, observing that this way they recover the well-known local expressions for the D-bracket without fluxes. Here we call the bracket \eqref{eq:almost_D-brac} an {\bf almost} D-bracket because here it is not associated to any para-Hermitian structure, as described in \cite{vaisman2012geometry,vaisman2013towards,Svoboda:2018rci,Szabo-paraherm}. However, in Theorem \ref{theo:D-bracket} we show that in the setting of generalized para-K\"ahler geometry, the almost D-bracket with flux $\pm H_b$ in fact is the D-bracket for the para-Hermitian structures $K_\pm$ associated to the generalized para-K\"ahler structure. 
\end{remark}

There is a tensorial quantity associated to any generalized connection $D$ called the {\bf generalized torsion}:
\begin{align}\label{eq:gen_torsion}
T^D(u,v,w)=\la D_uv-D_vu-[u,v]_H,w\ra +\la D_wu,v\ra.
\end{align}
For the Bismut connection, the generalized torsion is given by \cite{Gualtieri:2007bq} $T^D=2\pi_+^*H_b+2\pi_-^*H_b$:

\begin{proposition}\label{prop:gen_torsion}
Let $D$ be the Bismut connection associated to a generalized metric $\GG(\eta,b)$. Then the pure components in $\Lambda^3 C_\pm$ of the generalized torsion $T^D$ satisfy
\begin{align*}
T^D(\pi_\pm^{-1}X,\pi_\pm^{-1}Y,\pi_\pm^{-1}Z)=2H_b(X,Y,Z),
\end{align*}
while the mixed components vanish.
\end{proposition}
\begin{proof}
The fact that the mixed components vanish is a direct consequence of the formula \eqref{eq:n_pm} for $D$. The pure components can be calculated directly:
\begin{align*}
T^D(\pi_\pm^{-1}X,\pi_\pm^{-1}Y,\pi_\pm^{-1}Z)&=\la \pi_\pm^{-1}\n^\pm_XY-\n^\pm_YX,\pi_\pm^{-1} Z\ra+\la \pi_\pm^{-1}\n^\pm_ZX,\pi_\pm^{-1} Y\ra\\
&-\la [\pi^{-1}_\pm X,\pi^{-1}_\pm Y],\pi^{-1}_\pm Z\ra.
\end{align*}
Now, combining \eqref{eq:genmetric_metric}, \eqref{eq:n_pm} and \eqref{eq:almost_D-brac}, we get
\begin{align*}
T^D(\pi_\pm^{-1}X,\pi_\pm^{-1}Y,\pi_\pm^{-1}Z)&=\pm 2\eta (\lc_XY-\lc_YX,Z)\pm 2\eta(\lc_ZX,Y)+ 3H_b(X,Y,Z)\\
&\mp 2\eta(\bl X,Y\br^{\lc},Z)- H_b(X,Y,Z)\\
&=2H_b(X,Y,Z).
\end{align*}
\end{proof}

\section{Commuting Pairs of Generalized Structures}\label{sec:commuting_pairs}
After discussing individual types of generalized structures, we now take the natural next step and explore the construction of {\it commuting pairs}, which -- as the name suggests -- involves pairs of generalized structures that commute. When we add an additional requirement that the product of such pair of generalized structures is non-degenerate, we always get a corresponding pair of (now not necessarily commuting) tangent bundle endomorphisms. The goal here is to explore the interplay between the properties of the generalized structures and their tangent bundle counterparts, particularly their integrability.

We begin by describing Generalized K\"ahler (GK) geometry, which is a well studied example of a commuting pair, and use it to illustrate the general features of the commuting pairs construction. In particular, we show how it is related to the equivalent description in terms of the tangent bundle data of bi-Hermitian geometry of Gates-Hull-Rocek \cite{Gates-hull-rocek-biherm}.

Let $\II_+$ be an almost GC structure and $\GG$ a commuting generalized metric. Then $\II_-=\GG\II_+$ is another almost $GC$ structure and any two of the triple $(\II_+,\II_-,\GG)$ commute. This then implies that both GC structures preserve the eigenbundles $C_\pm$ of $\GG$ and therefore yield endomorphisms of $C_\pm$ that square to $-\id$. It is then easy to see that $\II_+\mid_{C_\pm}=\pm \II_-\mid_{C_\pm}$ and the following
\begin{align}\label{eq:I_plusminus}
I_+=\pi_+ \II_\pm \pi^{-1}_+\quad I_-=\pm\pi_- \II_\pm \pi^{-1}_-
\end{align}
where $\pi_\pm$ are the isomorphisms associated to $\GG$ by \eqref{eq:pi_iso}, defines a pair of almost complex structures. The metric $g$ defining $\GG=\GG(g,b)$ is then Hermitian with respect to both of them:
\begin{align*}
g(I_\pm X,I_\pm Y)&=\frac{1}{2}\la \pi_+^{-1}I_\pm X,\pi_+^{-1}I_\pm Y\ra=\frac{1}{2}\la \II_\pm\pi_+^{-1} X,\II_\pm\pi_+^{-1} Y\ra=\frac{1}{2}\la \pi_+^{-1} X,\pi_+^{-1} Y\ra\\
&=g(X,Y),
\end{align*}
where we used \eqref{eq:genmetric_metric} and \eqref{eq:I_plusminus}. The data $(\II_+,\II_-,\GG)$ therefore defines an almost {\bf bi-Hermitian structure} with a B-field $b$ \footnote{The B-field becomes relevant when one formulates the integrability conditions on $\II_\pm$ in terms of the bi-Hermitian data} $(g,b,I_+,I_-)$ on the tangent bundle. We can observe that the correspondence
\begin{align*}
(\II_+,\II_-,\GG) \longleftrightarrow (g,b,I_+,I_-)
\end{align*}
preserves the {\it type} of the data; the bi-Hermitian data on $T$ corresponds to bi-Hermitian data on $\TT$, where the Hermitian metric of $\II_\pm$ is given by
\begin{align*}
h(u,v)\coloneqq \la \GG u,v\ra,\quad u,v\in \se(\TT).
\end{align*}
The signature of $g$ then also corresponds to the signature of $h$. A similar construction can be carried out starting from any pair $(\Aa,\GG)$ of commuting generalized almost structures, where $\GG$ is non-degenerate (see Definition \eqref{def:gen_structures}). This is what we will do in the following subsections; when $\Aa=\KK$ is a generalized para-complex structure, we get {\bf generalized para-K\"ahler geometry} and when $\Aa=\JJ$ is generalized product structure, we get {\bf generalized chiral geometry} and when $\Aa=\JJ$ is a generalized anti-complex structure, we get generalized anti-K\"ahler geometry.

To summarize, we label any generalized structure $\Aa$ by a pair of signs $(\ap,\bt)$, according to:
\begin{align*}
\Aa= \ap\id, \quad \la \Aa \cdot,\Aa \cdot \ra=\bt \lara,
\end{align*}
i.e. $(-,+)$ represents a GC structure, $(+,-)$ a GpC structure and $(+,+)$ a generalized Product structure:

\begin{table}[h!]
\centering
\begin{tabular}{ c | c | c | c }
 $\Aa_+$ Type & $\Aa_-$ Type & $\GG=\JJ_+\JJ_-$ Type & \\
 \hline 
 $(-,+)$ & $(-,+)$ & $(+,+)$ & Generalized K\"ahler \\  
 $(+,-)$ & $(+,-)$ & $(+,+)$ & Generalized para-K\"ahler \\
 $(+,+)$ & $(+,+)$ & $(+,+)$ & Generalized chiral\\
 $(-,-)$ & $(-,-)$ & $(+,+)$ & Generalized anti-K\"ahler
\end{tabular}
\caption{Commuting pairs of generalized structures.}
\label{tab:commutingpairs}
\end{table}
\begin{remark}\label{rem:commutingpairs}
The list of all commuting pairs (where $\Aa_\pm$ is not necessarily of the same type) is exhausted by a geometry developed by Vaisman in \cite{vaisman2015generalized} coincidentally called generalized para-K\"ahler geometry, which consists of $\Aa_\pm$ of types $(\pm,\mp)$, i.e. one GC and one GpC structure, and the metric structure $\GG$ is of type $(-,-)$, i.e. non-degenerate generalized anti-complex structure. We would like to argue, however, that this name is more appropriate for the obvious paracomplex analog of generalized K\"ahler geometry, represented in line $2$ of Table \ref{tab:commutingpairs}. For the commuting pair of Vaisman, we coin the term {\bf generalized semi-K\"ahler geometry}, because the corresponding tangent bundle geometry involves both K\"ahler and para-K\"ahler geometries.
\end{remark}

\subsection{Generalized Para-K\"ahler Structures}\label{sec:GpK}
We now discuss the Generalized para-K\"ahler structures in more detail. Because a lot of constructions are entirely analogous to their complex counterpart in generalized K\"ahler geometry, we will frequently not give excessive detail. To consult classical literature on GK geometry, see \cite{Gualtieri:2003dx,Gualtieri:2010fd}.

\begin{Def}
An (almost) \textbf{Generalized para-K\"ahler structure} (GpK) is a commuting pair $(\GG,\KK_+)$ of a split signature generalized metric $\GG=\GG(\eta,b)$ and a GpC structure $\KK_+$. If additionally both $\KK_+$ and $\KK_-\coloneqq \GG\KK_+$ are integrable w.r.t. the (twisted) Dorfman bracket, we call $(\GG,\KK_+)$ a (twisted) GpK structure.
\end{Def}
Since any two structures in the triple $(\GG,\KK_+,\KK_-)$ determine the third, we may refer to the GpK structure $(\GG,\KK_+)$ by the pair $(\KK_+,\KK_-)$, in particular when integrability -- which is tied with $\KK_{\pm}$ -- is discussed.
\begin{Ex}
Let $(\eta,K)$ be an almost para-Hermitian structure, with $\omega=\eta K$ the fundamental form. Then
\begin{align*}
\KK_+=
\begin{pmatrix}
K & 0 \\
0 & -K^*
\end{pmatrix},\ 
\KK_-=
\begin{pmatrix}
0 & \omega^{-1} \\
\omega & 0
\end{pmatrix},\ 
\GG=
\begin{pmatrix}
0 & \eta^{-1} \\
\eta & 0
\end{pmatrix},
\end{align*}
gives an almost generalized para-K\"ahler structure which is integrable iff $(\eta,K)$ is para-K\"ahler.
\end{Ex}

Let $C_\pm$ be the eigenbundles of $\GG$. As discussed in Section \ref{sec:commuting_pairs} above, $\KK_+\mid_{C_\pm}=\pm \KK_-\mid_{C_\pm}$ and we can therefore construct two para-complex structures $K_\pm$ as follows:
\begin{align}\label{eq:induced_Ks}
K_+=\pi_+ \KK_\pm \pi^{-1}_+\quad K_-=\pm\pi_- \KK_\pm \pi^{-1}_-
\end{align}
Using \eqref{eq:genmetric_metric}, it can be easily checked that $\eta(K_\pm X,K_\pm Y)=-\eta(X,Y)$ and $\eta K_\pm\coloneqq \omega_\pm$ defines two almost symplectic forms, therefore $(\eta,K_\pm)$ are two almost para-Hermitian structures. We therefore see that any (almost) generalized para-K\"ahler structure defines an (almost) bi-para-Hermitian structure $(\eta,K_\pm)$ with an extra data given by the two-form $b$. The converse is also true; given $(K_\pm,\eta, b)$ we reconstruct the isomorphisms $\pi_\pm$ and use them to define a pair of commuting structures $\KK_\pm$ using $K_\pm$:
\begin{align}\label{eq:GpK_from_biparaherm}
\KK_\pm=\pi_+^{-1}K_+\pi_+P_{C_+}\pm \pi_-^{-1}K_-\pi_-P_{C_-},
\end{align}
where $P_{C_\pm}$ are the projections onto $C_\pm$ given by $P_{C_\pm}=\frac{1}{2}(\id\pm \GG)$. In matrix form, this yields an expression similar to one well-known from GK geometry
\begin{align}\label{eq:GpK_genform}
\KK_{\pm}=\frac{1}{2}
\begin{pmatrix}
\id & 0 \\
b & \id
\end{pmatrix}
\begin{pmatrix}
K_+\pm K_- & \omega^{-1}_+\mp \omega^{-1}_- \\
\omega_+\mp \omega_- & -(K_+^*\pm K_-^*)
\end{pmatrix}
\begin{pmatrix}
\id & 0 \\
-b & \id
\end{pmatrix},
\end{align}

Because the existence of a generalized almost para-K\"ahler structure implies an existence of two almost para-Hermitian structures, we see the base manifold needs to be even-dimensional, even though individual generalized para-complex structures exist on any manifold.

\begin{proposition}
Any generalized almost para-K\"ahler manifold is of even dimension.
\end{proposition}

We now denote the $+1$ and $-1$ eigenbundles of $\KK_\pm$ by $\Lb_\pm$ and $\Lbt_\pm$. Because $\KK_\pm$ commute, $\KK_-$ will further split the eigenbundles of $\KK_+$ and vice versa. We will therefore denote
\begin{align*}
\begin{aligned}
\ell_+ &\coloneqq \Lb_+ \cap \Lb_-,\quad \ell_- \coloneqq \Lb_+ \cap \Lbt_-\\
\tl{\ell}_+ & \coloneqq \Lbt_+ \cap \Lbt_-,\quad \tl{\ell}_- \coloneqq \Lbt_+ \cap \Lb_-
\end{aligned}
\end{align*}
so that
\begin{align*}
\begin{aligned}
\Lb_+ &= \ell_+ \oplus \ell_-,\quad \Lbt_+ = \ellt_+ \oplus \ellt_-\\
\Lb_-&= \ell_+\oplus \ellt_-,\quad  \Lbt_-=\ellt_+\oplus \ell_-,
\end{aligned}
\end{align*}
as well as
\begin{align*}
\begin{aligned}
&C_\pm = \ell_\pm \oplus \tilde{\ell}_\pm.
\end{aligned}
\end{align*}
We get the decomposition of $\TT$ to four eigenbundles of $\KK_\pm$
\begin{align}\label{eq:decomposition}
\TT=\ell_+\oplus\ell_-\oplus\ellt_+\oplus\ellt_-.
\end{align}

\subsubsection{Bi-para-Hermitian Geometry and Integrability}
We will now discuss the relationship between the integrability of $\KK_{\pm}$ and properties of the induced tangent bundle data. We immediately see that the integrability of the GpC structures $\KK_\pm$ implies the involutivity of each of the four eigenbundles of \eqref{eq:decomposition}, since they are the intersections of involutive subbundles $\Lb_{\pm}$ and $\Lbt_{\pm}$. In fact, this is also a sufficient condition.

\begin{proposition}\label{prop:integr1}
The generalized almost para-K\"ahler structure $(\KK_1,\KK_2)$ is integrable iff all eigenbundles in the decomposition \eqref{eq:decomposition} are Courant involutive.
\end{proposition}
\begin{proof}
We will show that the integrability of both $\ell_\pm$ implies integrability of $\Lb_+=\ell_+\oplus\ell_-$ and the involutivity of $\Lb_-$ and $\Lbt_{\pm}$ needed for integrability of $\KK_{\pm}$ follows analogously. Let $x_+,z_+ \in \se(\ell_+)$ and $y_-,z_- \in \se(\ell_-)$. Because we assume $\ell_\pm$ are involutive, we only need to show that $[\ell_+,\ell_-],[\ell_-,\ell_+] \subset \Lb_+$. Using the property \eqref{eq:compatibility_courant_alg}, we have
\begin{align*}
&\langle [x_+,y_-],z_+\rangle=\pi(x_+)\langle y_-,z_+ \rangle -\langle y_-,[x_+,z_+]\rangle=0\\
&\langle [x_+,y_-],z_-\rangle=-\pi(y_-)\langle x_+,z_- \rangle +\langle x_+,[y_-,z_-]\rangle=0,
\end{align*}
because $\ell_\pm$ are mutually orthogonal and $\ell_\pm$ are involutive. This shows that $[\ell_+,\ell_-] \perp \Lb_+$, therefore $[\ell_+,\ell_-] \subset \Lb_+$ because $\Lb_+$ is maximally isotropic, proving that $\Lb_+$ is involutive.
\end{proof}

We will now aim to express the involutivity of $\ell_\pm$ and $\ellt_\pm$ in terms of the induced bi-para-Hermitian data $(\omega_{\pm},\eta,b)$. For this, we first notice that the isomorphisms $\pi_\pm$ map the four bundles $\ell_\pm$ and $\ellt_\pm$ exactly to the four eigenbundles of $K_{\pm}$ in $T$, which we will denote by $T^{(1,0)_\pm}$ and $T^{(0,1)_\pm}$ Explicitly, we have
\begin{align*}
\pi_+\ell_+&= T^{(1,0)_+},\ \pi_+\ellt_+= T^{(0,1)_+}\\
\pi_-\ell_-&= T^{(1,0)_-},\ \pi_+\ellt_-= T^{(0,1)_-},
\end{align*}
where $(p,q)_\pm$ denote the $(p,q)$ decompositions \eqref{eq_plusminus_decomp} induced by $K_{\pm}$. The above can be checked, for example by using \eqref{eq:induced_Ks}:
\begin{align*}
K_+(\pi_+\ell_+)=\pi_+\KK_{\pm}\ell_+=\pi_+\ell_+,
\end{align*}
and so $\pi_+\ell_+$ must be the $+1$ eigenbundle of $K_+$, $T^{(1,0)_+}$. Next, because of this, each of the bundles $\ell_\pm,\ellt_\pm$ can be expressed as graphs of $b\pm \eta$. For example, for $x_+\in \se(\ell_+)$ (recalling $\eta=\omega_{\pm}K_{\pm}$), we get
\begin{align}\label{eq:graph_pi}
x_+=\pi_+^{-1}X=X+(b+ \eta)X=X+(b+ \omega_+K_+)X=X+(b+ \omega_+)X,
\end{align}
for some $X\in \se(T^{(1,0)_+})$. All the bundles in \eqref{eq:decomposition} can therefore be expressed as graphs of two-forms $b\pm\omega_{\pm}$ mapping from the eigenbundles of $K_{\pm}$.
We can now formulate the integrability of $\KK_{\pm}$ using the data $(\omega_{\pm},\eta,b)$.

\begin{theorem}\label{theo:integr_GpK}
A generalized almost para-K\"ahler structure $(\KK_+,\KK_-)$, given alternatively by the induced biparahermitian data $(K_{+},K_-,\eta,b)$, is integrable if and only if the following conditions are simultaneously satisfied
\begin{enumerate}
\item $K_{\pm}$ are integrable para-Hermitian structures, i.e. their Nijenhuis tensors vanish
\item $\rd^p_+\omega_+=-\rd^p_-\omega_-=-(H+\rd b)$,
\end{enumerate}
where $\rd^p_{\pm}=(\p^{(1,0)}_\pm-\p^{(0,1)}_\pm)$ are the $\rd^p$ operators \eqref{eq:dp-operator} of $K_\pm$.
\end{theorem}
\begin{proof}
We have seen previously that the integrability of $(\KK_+,\KK_-)$ is equivalent to the bundles $\ell_\pm$ and $\ellt_\pm$ being involutive under the Dorfman bracket. We have further found that all $\ell_\pm,\ellt_\pm$ can be written as $X+(b\pm \omega_\pm)$ for $X$ a vector in eigenbundles of $K_\pm$. We will now use to find the conditions on the involutivity of these bundles.

We start with $\ell_+$ which is given by $X+(b+\omega_+)X$, $X\in \se(T^{(1,0)_+})$. The Dorfman bracket of two such sections is
\begin{align*}
[X+(b+\omega_+)X,Y+(b+\omega_+)Y]_H&=[X,Y]+(b+\omega_+)([X,Y])\\
&+\imath_Y\imath_{X}(\rd(b+\omega_+)+H).
\end{align*}
The only tangent component is $[X,Y]$ and so it has to belong to $T^{(1,0)_+}$, meaning the bundle $\se(T^{(1,0)_+})$ is Frobenius integrable. When this is satisfies, we see that $[X,Y]+(b+\omega_+)([X,Y])$ in turn belongs to $\ell_+$, which also implies that
\begin{align}\label{eq:proof3}
\imath_Y\imath_{X}(\rd(b+\omega_+)+H)=0, \quad X,Y\in \se(T^{(1,0)_+})
\end{align}
must be satisfied. The $(1,0)_+$ component of this equation yields
\begin{align*}
(\rd b+H)^{(3,0)_+}=0,
\end{align*}
since $(\rd \omega_+)^{(3,0)}=0$ whenever $T^{(1,0)_+}$ is integrable. The $(0,1)_+$ component of \eqref{eq:proof3} then translates to
\begin{align*}
(\rd b+H)^{(2,1)_+}=-\rd\omega_+^{(2,1)_+}.
\end{align*}
Carrying out the same argument for the bundle $\ellt_+$ then tells us $T^{(0,1)_+}$ is integrable and
\begin{align*}
(\rd b+H)^{(0,3)_+}&=0,\\
(\rd b+H)^{(1,2)_+}&=+\rd\omega_+^{(1,2)_+}.
\end{align*}
Summing up, the involutivity of $\ell_+$ and $\ellt_+$ is equivalent to $K_+$ being integrable and further, since $\omega_+$ is of type $(1,1)_+$,
\begin{align*}
(\rd b+H)=-\p_+^{(1,0)}\omega_++\p^{(0,1)}_-\omega_+=-\rd^p_+\omega_+.
\end{align*}
Analogous calculation for $\ell_-$ and $\ellt_-$ then completes the proof.
\end{proof}

\subsubsection{Examples of GpK structures}\label{sec:examples}
We have already seen that the simplest GpK structure is given by a para-K\"ahler structure. Here we present few more examples.

\begin{Ex}[Para-Hyperk\"ahler geometry]\label{ex:parahyperkhlr}
	This is the para-complex version of the correspondence between the hyperk\"ahler and Generalized K\"ahler geometries \cite[Example 6.3]{Gualtieri:2003dx}. Let $(\eta,I,J,K)$ be a para-hyper-K\"ahler structure (see Appendix \ref{app:parahyperhermitian}). This means that $(I,J,K)$ is a para-hypercomplex triple, $-I^2=J^2=K^2=\id$ and $(J,\eta)$, $(K,\eta)$ are para-K\"ahler, while $(I,\eta)$ is pseudo-K\"ahler, i.e. the associated fundamental forms $\omega_{I/J/K}=\eta I/\eta J/\eta K$ are symplectic. In particular, $(J,K,\eta)$ is a bi-Hermitian structure and therefore defines a generalized para-K\"ahler structure $(\KK_\pm)$ by
	\begin{align*}
	\KK_{\pm}=\frac{1}{2}
	\begin{pmatrix}
	J\pm K & \omega^{-1}_J\mp \omega^{-1}_K \\
	\omega_J\mp \omega_K & -(J^*\pm K^*)
	\end{pmatrix},
	\end{align*}
	which, just like in the GK case, can be rewritten as $b$-field transformations by $\pm\omega_I$ of two non-degenerate GpC structures (customarily called symplectic type structures):
	\begin{align}\label{eq:GpK_parahyperkahler}
	\KK_{\pm}=
	\begin{pmatrix}
	\id & 0 \\
	\pm \omega_I & \id
	\end{pmatrix}
	\begin{pmatrix}
	0 & \frac{1}{2}(\omega^{-1}_J\mp \omega^{-1}_K) \\
	\omega_J\mp \omega_K & 0
	\end{pmatrix}
	\begin{pmatrix}
	\id & 0 \\
	\mp \omega_I & \id
	\end{pmatrix}.
	\end{align}	
\end{Ex}

\begin{Ex}[B-transformation of a para-K\"ahler structure]
	In \cite{Svoboda:2018rci} a B-transformation of a para-Hermitian structure was introduced. This is an operation on an almost para-Hermitian manifold $(\PS,\eta,K)$, which shears the $+1$ eigenbundle $T_+$ of $K$ into the direction of $T_-$ (the $-1$ eigenbundle of $K$) and amounts to adding a $(2,0)$ form to $\omega=\eta K$,
	\begin{align*}
	\omega \mapsto \omega+2b,
	\end{align*}
	producing a new almost para-Hermitian structure $(\eta,K_B=K+2B)$, where we denote $B=\eta^{-1}b$. $K_B$ can then be thought of as a finite deformation of $K$.
	
The metric $\eta$ facilitates an isomorphism $T_\pm\simeq T^*_\mp$ (we discuss this in more detail in Section \ref{sec:small_large_CA}) and the tangent bundle can consequently be seen as $T_+\oplus T_-=T_\pm\oplus T_\pm^*$. From the point of view of the bundle $T_+\oplus T_+^*$, the B-transformation of $K$ is the usual $b$-field transformation, changing the splitting $T_+\oplus T_+^*\mapsto e^b(T_+)\oplus T_+^*$. Similarly, we can see this operation as a $\beta$-field transformation of the bundle $T_-\oplus T_-^*$.
	
	One can then associate a GpK structure to the bi-para-Hermitian data $(\eta,K_B,K)$:
	\begin{align*}
	\KK_+=
	\begin{pmatrix}
	K+B & \beta \\
	b & -(K+B)^*
	\end{pmatrix},\quad
	\KK_-=
	\begin{pmatrix}
	B & \omega^{-1}+\bt \\
	\omega+b & -B^*
	\end{pmatrix},
	\end{align*}
	where $\beta=\eta^{-1}b\eta^{-1}$.
\end{Ex}

\begin{Ex}
	We end this section with examples of GpK structures with non-zero flux.
	Consider the para-quaternions 
	\[ \mathbb{H}' = \{ q = x_1 + x_2i + x_3j + x_4k : -i^2 = j^2 = k^2 = 1, k=ij, ij=-ji \}. \]
	Multiplication by $i,j,k$ induces six distinct natural structures on $\mathbb{H}'$, depending on whether one uses multiplication on the right or on the left. We denote by $I_+,J_+,K_+$ multiplication on the left by $i,j,k$, respectively, and by $I_-,J_-,K_-$ multiplication on the right by $i,j,k$, respectively. Note that both triples $I_+, J_+, K_+$ and $I_-,J_-,K_-$ are para-hypercomplex structures on $\mathbb{H}'$. Also, the para-complex structures $J_\pm,K_\pm$ are such that $J_+J_- = J_-J_+$, $J_+K_- = K_-J_+$, $K_+K_- = K_-K_+$ and $K_+J_- = J_-K_+$.
	
	Consider the quotient $Y = (\mathbb{H}' \backslash \{ x_1^2 +x_2^2 = x_3^2 + x_4^2 \}) / \sim$ where $q \sim 2q$ for all $q \in \mathbb{H}' \backslash \{ x_1^2 +x_2^2 = x_3^2 + x_4^2 \}$. Note that the structures $I_\pm,J_\pm,K_\pm$ described above descend to the quotient $Y$. If we set
	\[ |q|^2 = x_1^2 +x_2^2 - x_3^2 - x_4^2,
	\]
	then 
	\[ \eta = \frac{1}{|q|^2}( dx_1 \otimes dx_1 + dx_2 \otimes dx_2 - dx_3 \otimes dx_3 - dx_4 \otimes dx_4) \] 
	is a pseudo-Hermitian metric on $Y$ of signature $(n,n)$ such that 
	\[ \eta(I_\pm \cdot, I_\pm \cdot) = -\eta(J_\pm \cdot, J_\pm \cdot) = -\eta(K_\pm \cdot, K_\pm \cdot) = \eta(\cdot,\cdot). \] 
	In other words, $(\eta,I_\pm,J_\pm,K_\pm)$ are para-hyperHermitian structures on $Y$. Moreover, a direct computation gives
	\[ \rd^p_{J_\pm}\omega_{J_\pm} = \rd^p_{K_\pm}\omega_{K_\pm} = \pm H,\]
	where
	\[  H = \frac{2}{|q|^4}
	( x_1 dx_2 \wedge dx_3 \wedge dx_4 -  x_2 dx_1 \wedge dx_3 \wedge dx_4 +  x_3 dx_1 \wedge dx_2 \wedge dx_4
	-  x_4 dx_1 \wedge dx_2 \wedge dx_3)
	\]
	and $dH = 0$, implying that $(J_\pm,\eta)$, $(K_\pm,\eta)$, and $(J_\pm,K_\mp,\eta)$ are GpK structures on $Y$ with non-zero flux $H$. These are also examples of GpK that do not come from a para-hyperK\"ahler or even para-hyperHermtian structure (since the para-complex structures do not commute).
\end{Ex}

\subsubsection{Para-Holomorphic Poisson structures}
We now show that every GpK manifold has a Poisson bivector, which in addition is para-holomorphic with respect to both para-Hermitian structures $(K_+,K_-)$. This then gives another pair of GpC structures -- which are related by a certain B-field transformation -- constructed purely from the GpK data.

\begin{theorem}\label{theo:bi-paraH_Poisson}
Let $(\KK_+,\GG)$ be a GpK structure and $(\eta, K_+,K_-,b)$ the corresponding bi-para-hermitian data. Then
\begin{align}\label{eq:parahol_poisson_Q}
Q=\frac{1}{2}[K_+,K_-]\eta^{-1} = \dfrac{1}{2}\eta^{-1}[K_+^*, K_-^*]
\end{align}
is a Poisson bivector of type $(2,0)_\pm+(0,2)_\pm$, which is para-holomorphic with respect to both $K_\pm$.
\end{theorem}
\begin{proof}
$Q$ being of type $(2,0)+(0,2)$ with respect to both $K_\pm$ is equivalent to $Q(K_\pm^*,K_\pm^*)=Q(\cdot,\cdot)$, or $K_\pm Q K_\pm^*=Q$. This can be simply checked by using $[K_+,K_-]=(K_++K_-)(K_+-K_-)$ and para-Hermitian compatibility conditions, $K_\pm\eta^{-1}=-\eta^{-1} K^*_\pm$.

For $Q$ to be para-Holomorphic, the local coefficient functions $Q^{ij}$ of the $(2,0)$ components have to be locally independent of the $\xt$ coordinates, i.e. $\pt^i Q^{jk}=0$ and similarly the $(0,2)$ components have to satisfy $\p_iQ_{jk}=0$ (see Example \ref{ex:parahol_wedgepowers} and Equation \eqref{eq:parahol_section}). We now check the condition on the $(2,0)$ component for $K_+$, i.e. $\pt^{i}Q^{jk}=0$, where we are simplifying the notation by omitting the $+$ subscript labelling $K_+$ and in further text we also denote $K_+=K$.

Let $(x^i,\xt_i)$ be the local adapted coordinates of $K$. The coordinate functions of the $(2,0)$ component of $Q$ are given by
\begin{align}\label{eq:proof_holpoisson}
\begin{aligned}
Q^{jk}=Q(dx^j,dx^k)&=\frac{1}{2}\la (KK_--K_-K) \eta^{-1}(dx^j),dx^k\ra\\
&=\frac{1}{2}(\la K_-\eta^{-1}(dx^j),K^*dx^k\ra+\la K_-\eta^{-1}(dx^j),dx^k\ra)\\
&=\omega_-^{-1}(dx^j,dx^k),
\end{aligned}
\end{align}
because $\eta^{-1}(dx^j)$ is in $\Lt$.

We now use the following formula for the exterior derivative of a one-form:
\begin{align}\label{eq:proof4}
(\rd \ap) (X,Y)=X\ap(Y)-Y\ap(X)-\ap([X,Y]).
\end{align}
Choosing in \eqref{eq:proof4} $\ap=dx^j$, $X=\pt^i$ and $Y=\omega_-^{-1}(dx^k)$, we get (because $\rd dx^j=0$ and $dx^j(\pt^i)=0$)
\begin{align*}
\pt^i\omega_-^{-1}(dx^j,dx^k)&=\la[\pt^i,\omega_-^{-1}(dx^k)],dx^j\ra=\la \lc_{\pt^i}\omega_-^{-1}(dx^k)-\lc_{\omega_-^{-1}(dx^k)}\pt^i,dx^j\ra\\
&=\la (\lc_{\pt^i}\omega_-^{-1})dx^k+\omega_-^{-1}(\lc_{\pt^i}dx^k)-\lc_{\omega_-^{-1}(dx^k)}\pt^i,dx^j\ra.
\end{align*}
The formula \eqref{eq:proof4} in the following form
\begin{align*}
(\rd \ap) (X,Y)=(\lc_X\ap)(Y)-(\lc_Y\ap)(X),
\end{align*}
with $\ap=dx^k$, $X=\pt^j$ and $Y=\omega_-^{-1}(dx^j)$, can be used to derive
\begin{align*}
\la \lc_{\pt^i}dx^k,\omega_-^{-1}(dx^j)\ra=\la \lc_{\omega_-^{-1}(dx^j)}dx^k,\pt^i\ra=-\la  \lc_{\omega_-^{-1}(dx^j)}\pt^i,dx^k\ra,
\end{align*}
so that we have
\begin{align}\label{eq:proof5}
\begin{aligned}
\pt^i\omega_-^{-1}(dx^j,dx^k)=&\la (\lc_{\pt^i}\omega_-^{-1})dx^k,dx^j\ra\\
&+\la  \lc_{\omega_-^{-1}(dx^j)}\pt^i,dx^k\ra-\la\lc_{\omega_-^{-1}(dx^k)}\pt^i,dx^j\ra.
\end{aligned}
\end{align}
We further use $\n^\pm K_\pm=0$, where $\n^\pm=\lc\pm\frac{1}{2}\eta^{-1}H$, $H=h+\rd b$, which implies
\begin{align*}
(\lc_X\omega_\pm)(Y,Z)=\mp\frac{1}{2}(H(X,K_\pm Y,Z)+H(X,Y,K_\pm Z)).
\end{align*}
Consequently, this yields
\begin{align*}
\la (\lc_{\pt^i}\omega_-^{-1})dx^k,dx^j\ra=\frac{1}{2}(H(\pt^i,\omega_-^{-1}dx^k,\eta^{-1}dx^j)+H(\pt^i,\eta^{-1}dx^k,\omega_-^{-1}dx^j)),
\end{align*}
as well as
\begin{align*}
\la  \lc_{\omega_-^{-1}(dx^j)}\pt^i,dx^k\ra &
=-\frac{1}{2} \left(\lc_{\omega_-^{-1}(dx^j)}\omega\right)(\pt^i,\eta^{-1}dx^k)\\
&=-\frac{1}{2}H(\omega_-^{-1}dx^j,\pt^i,\eta^{-1}dx^k),
\end{align*}
and
\begin{align*}
-\la\lc_{\omega_-^{-1}(dx^k)}\pt^i,dx^j\ra=\frac{1}{2}H(\omega_-^{-1}dx^k,\pt^i,\eta^{-1}dx^j).
\end{align*}
Summing these terms as in \eqref{eq:proof5}, we conclude that $\pt^i\omega_-^{-1}(dx^j,dx^k)=0$.

Let us now show that $Q$ is also a Poisson structure, i.e. $[Q,Q]=0$, where $\brac$ denotes the natural extension of a Lie brackets to polyvector fields, the Schouten bracket. To prove this, we will use an argument presented in \cite{Hitchin:2005cv} in the proof of the analogous statement in GK geometry. First, we observe that $[Q_+,Q_-]=[Q_-,Q_+]=0$, $Q_\pm$ denoting the $(2,0)$ and $(0,2)$ components, because $Q$ is para-holomorphic. Therefore, $[Q,Q]=[Q_+,Q_+]+[Q_-,Q_-]$ and is of type $(3,0)+(0,3)$. Next, combining Lemma \ref{lem:GpC_poisson} with Equation \eqref{eq:GpK_genform} tells us that $[\omega_+^{-1}\pm\omega_-^{-1},\omega_+^{-1}\pm\omega_-^{-1}]=0$, which in particular means that
\begin{align}\label{eq:proof_holpoisson2}
 [\omega_+^{-1},\omega_+^{-1}]+[\omega_-^{-1},\omega_-^{-1}]=0.
\end{align}
With respect to $K_+$, $\omega_+^{-1}$ is type $(1,1)$ and therefore  $[\omega_+^{-1},\omega_+^{-1}]$ has no $(3,0)$ component. On the other hand, calculation similar to \eqref{eq:proof_holpoisson} shows that $Q=(\omega_-^{-1})^{(2,0)}-(\omega_+^{-1})^{(0,2)}=Q_++Q_-$. Therefore, the only $(3,0)$ component of \eqref{eq:proof_holpoisson2} is given by $[Q_+,Q_+]$, which means $[Q_+,Q_+]=0$ and similarly for $Q_-$. $[Q_+,Q_-]$ also vanishes due to $Q$ being para-holomorphic, i.e. $\p_i Q_-=\pt^i Q_+=0$.
\end{proof}

Because $Q$ is para-holomorphic with respect to both $K_\pm$, we obtain another pair of GpC structures given by
\begin{align}\label{eq:parahol_poisson_GpC}
\mathcal{Q}_\pm=
\begin{pmatrix}
K_\pm & Q \\
0 & -K_\pm^*
\end{pmatrix}.
\end{align}
Moreover, we find that when $(K_++K_-)$ is invertible, these GpC structures are related by a 
\begin{proposition}
Let $(\KK_+,\GG)$ be a GpK structure, $(\eta, K_+,K_-,b)$ the corresponding bi-para-hermitian data and $Q$ the para-holomorphic poisson structure given by \eqref{eq:parahol_poisson_Q}. Then \eqref{eq:parahol_poisson_GpC} defines pair of GpC structures and when $(K_++K_-)$ is in addition invertible, $\mathcal{Q}_\pm$ are related by a $b$-field transformation: $e^F(\mathcal{Q}_+)=\mathcal{Q}_-$, where $F\coloneqq 2\eta(K_++K_-)^{-1}$.
\end{proposition}
\begin{proof}
The fact that $\mathcal{Q}_\pm$ are GpC structures follows from the compatibility between $Q$ and $K_\pm$ and the fact that $K_\pm$ are integrable and $Q$ para-holomorphic. The fact that $e^F(\mathcal{Q}_+)=\mathcal{Q}_-$ is equivalent to the equations
\begin{align*}
K_+^*-K_-^*&=FQ\\
FK_++K_-^*F&=0,
\end{align*}
which are easy to verify. 
\end{proof}

\subsubsection{The Large and Small Courant Algebroids and the D-bracket}\label{sec:small_large_CA}
Recall that when $(\PS,\eta,K)$ is an $2n$-dimensional integrable para-Hermitian manifold, the eigenbundles $T^{(1,0)}\coloneqq T_+$ and $T^{(0,1)}\coloneqq T_-$ of $K$ corresponding to eigenvalues $\pm 1$ integrate to $n$-dimensional foliations $\FF_{(\pm)}$. We now consider the generalized tangent bundles $(\TT)\FF_\pm$, for which we will use the name coined in physics \cite{Chatzistavrakidis:2018ztm} {\bf small Courant algebroids} of $\PS$. The Courant algebroid on $(\TT)\PS$ will be called a {\bf large Courant algebroid} if the distinction needs to be emphasised.

There are vector bundle isomorphisms \cite{freidel2017generalised,Svoboda:2018rci}
\begin{align*}
\begin{aligned}
\rho_\pbmb:\ T\PS &= T_+\oplus T_-\rightarrow T_\pm\oplus T_\pm^*\\
X &= x_++x_- \mapsto x_\pm+\eta(x_\mp),\quad x_\pm \in \se(T_\pm).
\end{aligned}
\end{align*}
By definition, $T\FF_{(\pm)}=T_\pm$ and so $T^*\FF_{(\pm)}=T_\pm^*\overset{\rho_\pbmb}{\simeq}T_\mp$, which means that $\rho_\pbmb$ map the small courant algebroids $(\TT)\FF_{(\pm)}$ to $T\PS=T_+\oplus T_-$. This gives $T\PS$ a Courant algebroid structures coming from $(\TT)\FF_{(\pm)}$:
\begin{align*}
(\TT)\FF_\pbmb&\xrightarrow{\rho_\pbmb}T\PS \\
(\lara,\brac_{\pbmb},\pi_T)&\longmapsto (\eta,\bracd_\pbmb,P_\pm).
\end{align*}
In the above, $\brac_{\pbmb}$ denotes the standard Dorfman brackets on $(\TT)\FF_\pbmb$, $\lara$ the standard pairing and $\pi_T$ the projections onto the tangent factor $T\FF_\pbmb$. The brackets $\bracd_\pbmb$ are defined simply by
\begin{align*}
\rho(\bl X,Y\br_\pbmb)=[\rho X,\rho Y]_{\pbmb}.
\end{align*}

\begin{remark}
When $(\eta,K_\pm)$ is the bi-para-Hermitian data corresponding to a GpK structure, the corresponding small courant algebroids associated to $K_\pm$ can be seen as the real analog of the holomorphic Courant algebroids that appear as a holomorphic reduction of generalized K\"ahler geometry \cite{Gualtieri:2010fd}.
\end{remark}

\paragraph{The D-bracket From GpK Geometry.} In \cite{Svoboda:2018rci}, it is shown that the sum of the brackets $\bracd_\pbmb$ is a bracket called the D-bracket $\bracd$ which appears in physics literature, i.e.
\[\bl X,Y \br := \bl X, Y \br_{(+)} + \bl X,Y \br_{(-)}\]
It is also shown that this bracket can be expressed in terms of $\lc$, the Levi-Civita connection of $\eta$, and $\omega$, the fundamental form of the para-Hermitian structure:
\begin{align}\label{eq:D-bracket}
\begin{aligned}
&\eta(\bl X,Y\br,Z)=\eta(\lc_XY-\lc_YX,Z)+\eta(\lc_ZX,Y)\\
&-\frac{1}{2}[\rd\omega^{(3,0)}(X,Y,Z)+\rd\omega^{(2,1)}(X,Y,Z)
-\rd\omega^{(1,2)}(X,Y,Z)-\rd\omega^{(0,3)}(X,Y,Z)].
\end{aligned}
\end{align}

Let now $(\GG,\KK)$ be a GpK structure on $\PS$, i.e. $(\GG,\KK)$ are endomorphisms of the large Courant algebroid $(\TT)\PS$ with a flux $H$, and let $(\eta,K_\pm)$ be the corresponding bi-Hermitian data. According to the above discussion, there is a D-bracket on $T\PS$ associated to each $K_\pm$. We now show that in this case the D-brackets match the almost D-brackets (see Proposition \ref{prop:almost_D-brac}) associated to $\GG$.
\begin{theorem}\label{theo:D-bracket}
Let $(\GG,\KK)$ be a GpK structure on $\PS$ with a flux $H$ and $(\eta,K_\pm)$ the corresponding bi-Hermitian data. Then the D-brackets $\bracd_\pm$\footnote{Here the labels $\pm$ correspond to the structures $K_\pm$, not their $\pm 1$ eigenbundles.} associated to the para-Hermitian structures $(\eta,K_\pm)$ are given by
\begin{align}\label{eq:D-bracket_pi}
\eta(\bl X,Y\br_\pm,Z)=\pm \frac{1}{2}\la [\pi^{-1}_\pm X,\pi^{-1}_\pm Y],\pi^{-1}_\pm Z\ra,
 \end{align}
where $\brac$ is the Dorfman bracket on $(\TT)\PS$ and $\pi_\pm$ the projections \eqref{eq:pi_iso} associated to $\GG$.
\end{theorem}
\begin{proof}
From Proposition \ref{prop:almost_D-brac} it follows that 
\begin{align*}
\pm \frac{1}{2}\la [\pi^{-1}_\pm X,\pi^{-1}_\pm Y],\pi^{-1}_\pm Z\ra&=\eta(\lc_XY-\lc_YX,Z)+\eta(\lc_ZX,Y)\\
&\pm \frac{1}{2}H_b(X,Y,Z).
\end{align*}
It remains to relate this to the expressions for the $D$-bracket \eqref{eq:D-bracket} associated to $K_\pm$. Theorem \ref{theo:integr_GpK} tells us that $K_\pm$ are necessarily integrable and therefore the $(3,0)$ and $(0,3)$ components of $\rd \omega_\pm$ in \eqref{eq:D-bracket} vanish. Because $\rd^p\omega=\rd\omega^{(2,1)}-\rd\omega^{(1,2)}$, the $(2,1)$ and $(1,2)$ components then get matched (recalling again Theorem \ref{theo:integr_GpK}) by equation $\rd^p_\pm \omega_\pm=\mp (H+\rd b)$. This completes the proof.
\end{proof}

In \cite{Svoboda:2018rci,Szabo-paraherm} the DFT fluxes are understood as a {\it relative phenomenon} between two para-Hermitian structures $K$ and $K'$. More precisely, let $\bracd$ and $\bracd'$ denote the D-brackets associated to $K$ and $K'$, respectively. Then the DFT flux ${\cal F}$, which is a $3$-form on the para-Hermitian manifold, is given by:
\begin{align*}
{\cal F}(X,Y,Z)=\eta(\bl X,Y\br-\bl X,Y\br',Z).
\end{align*}
In the context GpK geometry, we also acquire two para-Hermitian structures $K_\pm$. It then follows from the above calculations that the relative DFT flux for this pair is then, remarkably, the $H_b$-flux of the underlying large courant algebroid:
\begin{align*}
{\cal F}(X,Y,Z)=H+\rd b.
\end{align*}

We conclude with the following observation about generalized structures on the small Courant algebroids.
\begin{remark}
Let $(\PS,\eta,K)$ be a para-Hermitian manifold and again denote its small Courant algebroids $(\TT)\FF_\pbmb$. Because we have the isomorphisms $\rho_\pm$ mapping between the tangent bundle $T\PS$ and $(\TT)\FF_\pbmb$, we can understood any tangent bundle endomorphism $A\in \End(T\PS)$ also as an endomorphism of $(\TT)\FF_\pbmb$ and in particular when $A^2=\pm \id$ and $\eta(A\cdot,A\cdot)=\pm\eta$, $A$ induces generalized almost structures on the small Courant algebroids $(\TT)\FF_\pbmb$ defined by
\begin{align*}
\Aa_\pbmb \rho_\pbmb=\rho_\pbmb A.
\end{align*}
Since $T\FF_\pb\simeq T_+\simeq T^*\FF_\mb$\footnote{Here the first equivalence is for vector bundles over $\FF_\pb$ and the second for vector bundles over $\FF_\mb$.} and $T^*\FF_\pb\simeq T_-\simeq T\FF_\mb$, the para-Hermitian structure $K$ itself induces GpC structures $K_\pbmb$, which are the trivial ones (Example \ref{ex:GpC_trivial}):
\begin{align}\label{eq:small_GpC}
K_\pb=
\begin{pmatrix}
\id & 0 \\
0 & -\id
\end{pmatrix},\quad
K_\mb=
\begin{pmatrix}
-\id & 0 \\
0 & \id
\end{pmatrix}.
\end{align}
Analogously, any (almost) para-Hermitian, pseudo-Hermitian, chiral and anti-Hermitian structure on a para-Hermitian manifold will therefore induce two generalized almost para-complex, complex, product and anti-complex structures, respectively. In particular, because any commuting pair $(\Aa_\pm)$ of generalized almost structures on the large Courant algebroid over $\PS$ yields a pair of tangent bundle endomorphisms $A_\pm$, it will induce pairs of generalized almost structures on the small Courant algebroids given the corresponding generalized metric $\GG=\Aa_+\Aa_-$ is of the form $\GG(\eta,b)$.
\end{remark}

\subsection{Generalized Chiral Structures}\label{sec:gen_chiral}
In this section, we explore more in-depth the commuting pair $(\GG,\JJ)$ giving the generalized chiral structure.

\begin{Def}
A \textbf{generalized chiral structure} (GCh) is a commuting pair $(\GG,\JJ_+)$ of a generalized metric $\GG=\GG(\eta,b)$ and a GP structure $\JJ_+$.
\end{Def}
Note that for the commuting pair $(\GG,\JJ_+)$, $\JJ_-\coloneqq \GG\JJ_+$ is another GP structure. All the generalized almost structures defining a generalized chiral structure are thus non-isotropic, and so there is no notion of integrability for such structures in terms of the Courant bracket as in GK/GpK geometry. We nonetheless introduce a related notion of {\it integrability} for these structures in Section \ref{sec:bismut_integr}. 

The canonical example of a generalized chiral structure is given by usual chiral geometry (see Section \ref{sec:chiral/antiH}):
\begin{Ex}[Chiral geometry]\label{ex:GCh_canon}
Let $(J,\eta)$ be an almost chiral structure. The contraction $\HH:=\eta J$ is then a pseudo-Riemannian metric and $(J,\HH)$ is also an almost chiral structure. Define
\begin{align*}
\GG(\eta)=
\begin{pmatrix}
0 & \eta^{-1} \\
\eta & 0
\end{pmatrix},\quad
\GG(\HH)=
\begin{pmatrix}
0 & \HH^{-1} \\
\HH & 0	
\end{pmatrix},\quad
\JJ_+=
\begin{pmatrix}
J & 0 \\
0 & J^*
\end{pmatrix}.
\end{align*}
Then, both $(\GG(\eta),\JJ_+)$ and $(\GG(\HH),\JJ_+)$ define GCh structures such that $\JJ_-=\GG(\eta)\JJ_+=\GG(\HH)$.
\end{Ex}

Let now $(\GG=\GG(\HH,b),\JJ_+)$ be a generalized chiral structure and denote the eigenbundles of $\GG$ by $C_\pm$. All the facts about the tangent bundle structures corresponding to commuting pairs outlined at the begining of Section \ref{sec:commuting_pairs} hold true. Denoting the isomorphisms \eqref{eq:pi_iso} associated to $\GG$ by $\pi_\pm$, we obtain a pair of product structures $J_\pm$ on the tangent bundle given by
\begin{align*}
J_+=\pi_+\JJ_\pm\pi_+^{-1},\quad J_-=\pm \pi_-\JJ_\pm\pi^{-1}_-,
\end{align*}
such that $(J_\pm,\HH)$ is a pair of chiral structures on the tangent bundle. 

Conversely, the formula that recovers the generalized chiral data from $(J_\pm,\HH,b)$ is given by
\begin{align}\label{eq:GCh_from_bichiral}
\JJ_\pm=\pi_+^{-1}J_+\pi_+P_{C_+}\pm \pi_-^{-1}J_-\pi_-P_{C_-},
\end{align}
where  $P_{C_\pm}=\frac{1}{2}(\id\pm \GG)$ are the projections onto $C_\pm$. The usual expressions in the matrix form are
\begin{align}\label{eq:GCh_genform}
\JJ_{\pm}=\frac{1}{2}
\begin{pmatrix}
\id & 0 \\
b & \id
\end{pmatrix}
\begin{pmatrix}
J_+\pm J_- & \eta^{-1}_+\mp \eta^{-1}_- \\
\eta_+\mp \eta_- & J_+^*\pm J_-^*
\end{pmatrix}
\begin{pmatrix}
\id & 0 \\
-b & \id
\end{pmatrix},
\end{align}
where $\eta_\pm:=\HH J_\pm$ denote the two metrics associated to $(J_\pm,\HH)$.

\subsubsection{Born Geometry as a Generalized Chiral Structure}
We now explain how Born geometry fits in the picture of commuting pairs as a generalized chiral structure with anti-commuting tangent bundle data. 
\begin{proposition}
Let $(\GG(\HH,b),\JJ)$ be a generalized chiral structure and let $(J_\pm,\HH)$ be the corresponding tangent bundle data. Then $\{J_+,J_-\}=0$ is equivalent to $(J_\pm,\HH)$ being an (almost) Born structure.
\end{proposition}
\begin{proof}
Proposition \ref{prop:born=bichiral=bihermitian} tells us that the data $(\eta,I,J,K)$ of an (almost) Born structure induces a pair of chiral structures $(J,\HH)$ and $(K,\HH)$ with $\{J,K\}=0$, where $\HH=\eta J$. This pair is enough to construct the generalized chiral structure $(\GG(\HH,b),\JJ)$ with arbitrary $b$. The converse is obvious from the statement of Proposition \ref{prop:born=bichiral=bihermitian}.
\end{proof}

Because $J_\pm$ anti-commute, this situation is analogous to Example \ref{ex:parahyperkhlr}, where an (almost) para-hyperHermitian structure gives rise to an (almost) GpK structure. Indeed, let $(\eta, I,J,K)$ be an almost Born geometry and denote $\omega=\eta K$ and $\eta'=\eta I$. The commuting pair of generalized structures $\JJ_\pm$ \eqref{eq:GCh_genform} with $b$ set to zero then take the form (compare to \eqref{eq:GpK_parahyperkahler})
\begin{align*}
\JJ_\pm=
\begin{pmatrix}
J\pm K & \eta^{-1}\mp \eta'^{-1} \\
\eta\mp \eta' & J^*\pm K^*
\end{pmatrix}=
\begin{pmatrix}
\id & 0 \\
\mp\omega & \id
\end{pmatrix}
\begin{pmatrix}
0 & \frac{1}{2}(\eta^{-1}\mp\eta'^{-1}) \\
\eta\pm\eta' & 0
\end{pmatrix}
\begin{pmatrix}
\id & 0 \\
\pm \omega & \id
\end{pmatrix},
\end{align*}
which means the structures $\JJ_\pm$ are $b$-field transformations of anti-diagonal non-degenerate generalized product structures by $\omega$.

\subsection{Generalized Anti-K\"ahler Structures}
Generalized anti-K\"ahler geometry is the complex counterpart of generalized chiral geometry. 
Given the similarities, we only flesh out the basics.

\begin{Def}
A \textbf{generalized anti-K\"ahler structure} (GaK) is a commuting pair $(\GG,\JJ_+)$ of a generalized metric $\GG=\GG(\eta,b)$ and a generalized anti-complex structure $\JJ_+$.
\end{Def}

Generalized anti-K\"ahler structures $(\GG,\JJ_+)$ correspond to pairs of (almost) anti-Hermitian structures $(\eta,J_\pm)$ such that
\begin{align*}
\JJ_{\pm}=\frac{1}{2}
\begin{pmatrix}
\id & 0 \\
b & \id
\end{pmatrix}
\begin{pmatrix}
J_+\pm J_- & -(\HH^{-1}_+\mp \HH^{-1}_-) \\
\HH_+\mp \HH_- & J_+^*\pm J_-^*
\end{pmatrix}
\begin{pmatrix}
\id & 0 \\
-b & \id
\end{pmatrix},
\end{align*}
where $\JJ_- := \GG \JJ_+$ and $\HH_\pm:=\eta J_\pm$. 

As in the chiral case, the canonical example of a generalized anti-K\"ahler structure is given by usual anti-Hermitian geometry (see Section \ref{sec:chiral/antiH}):
\begin{Ex}[Anti-Hermitian geometry]\label{ex:GCh_canon}
	Let $(J,\eta)$ be an almost anti-Hermitian structure. The contraction $\HH:=\eta J$ is then a pseudo-Riemannian metric and $(J,\HH)$ is also an almost anti-Hermitian structure. Define
	\begin{align*}
	\GG(\eta)=
	\begin{pmatrix}
	0 & \eta^{-1} \\
	\eta & 0
	\end{pmatrix},\quad
	\GG(\HH)=
	\begin{pmatrix}
	0 & \HH^{-1} \\
	\HH & 0	
	\end{pmatrix},\quad
	\JJ_+=
	\begin{pmatrix}
	J & 0 \\
	0 & J^*
	\end{pmatrix}.
	\end{align*}
	Then, both $(\GG(\eta),\JJ_+)$ and $(\GG(\HH),\JJ_+)$ define GaK structures such that $\JJ_-=\GG(\eta)\JJ_+=\GG(\HH)$.
\end{Ex}

Moreover, almost anti-hyperHermitian structures correspond to almost generalized anti-K\"ahler structures whose almost complex structures anti-commute:
\begin{Ex}[Anti-hyperHermitian structures]
Let $(\eta, I, J, K)$ be an almost anti-hyperHermitian structure (see Appendix \ref{sec:anti-HH}). The pairs $(\eta, J)$ and $(\eta,K)$ are then both almost anti-Hermitian structures with $\{ J,K\} = 0$, and thus induce an almost generalized anti-K\"ahler structure whose associated almost complex structures $J,K$ anti-commute. Conversely, suppose $(\GG,\JJ_+)$ is a GaK structure whose tangent bundle data $(\eta, J_\pm)$ has anti-commuting almost complex structures $J_\pm$. Then, referring to Appendix \ref{sec:anti-HH}, $(\eta, I, J_+, J_-)$ is an anti-hyperHermitian structure.
\end{Ex}

Finally, as for GCh structures, one cannot define the integrability of GaK structures in terms of the Courant bracket because the generalized almost structures defining them are non-isotropic. We are however able to define their {\it integrability} in Section \ref{sec:bismut_integr} is the same way as for GCh structures.

\subsection{Generalized Bismut Connections and Integrability}\label{sec:bismut_integr}

A generalized Bismut connection $D$ associated to a generalized metric $\GG$ was introduced in \cite{Gualtieri:2007bq} as a Courant algebroid connection that parallelizes $\GG$ and has useful properties. In particular, it is proved in \cite{Gualtieri:2007bq} that a $GK$ structure $(\GG,\II)$ is integrable if and only if $D\II=0$ and the torsion of the connection is of an appropriate type. The idea of this section is to extend this observation to any commuting pair $(\GG,\Aa)$ and define integrability of generalized chiral and anti-K\"ahler structures in analogous way. As a middle step, we define a notion of {\it weak} integrability of commuting pairs by requiring only $D\Aa=0$. Further restrictions on the type of the generalized torsion of $D$ then defines {\it full} integrability; in the case of G(p)K geometry we require that the type is $(2,1)+(1,2)$ with respect to both generalized (para-)complex structures, while in the case of generalized chiral/anti-K\"ahler structures we require that the type is $(3,0)+(0,3)$. In this way, we can talk about integrability of generalized structures even if their eigenbundles are not isotropic (see discussion in Remark \ref{rem:integrability_non-iso}).

An additional advantage of this approach is that it provides a natural way to weaken the integrability. As we will see, weak integrability relaxes the Frobenius integrability of the corresponding tangent bundle structures, which can sometimes be desirable from the point of view of physics. For example, the para-Hermitian geometry of Double Field Theory (DFT) \cite{freidel2017generalised,Svoboda:2018rci,Szabo-paraherm,Hassler:2019wvn} may not always be fully integrable and various DFT fluxes enter as an obstruction to integrability. Moreover, in applications to non-linear supersymmetric sigma models, where the geometry of commuting generalized pairs $(\GG,\Aa)$ enters in the form of the pair of tangent bundle endomorphisms $A_\pm$, it has been observed that sometimes only the requirement that $A_\pm$ are parallelized by the connections \eqref{eq:n_pm}, $\n^\pm A_\pm$, might be sufficient \cite{SUSY_non-integrable,Stojevic:2009ub}. As we will show in Proposition \ref{prop:bismut_n_pm}, this is exactly the condition of weak integrability.

We start by presenting the integrability statement in terms of the generalized Bismut for GK structures:
\begin{theorem}[\cite{Gualtieri:2007bq}]\label{theo:gual_branes}
Let $(\GG,\II_\pm)$ be a commuting pair of generalized complex structures with metric $\GG=-\II_+\II_-$ and $D$ its generalized Bismut connection. Then $(\GG,\II_+)$ defines a GK structure and in particular both $\II_\pm$ are Courant integrable iff $D\II=0$ and $T^D$ is of type $(2,1)+(1,2)$ with respect to both $\II_\pm$.
\end{theorem}

The relationship between the integrability of GK structures and the generalized Bismut connection can be summed up as follows:
\begin{align*}
(\GG,\II)\ \text{Generalized K\"ahler} \Longleftrightarrow D\II=0,\ T^D\ \text{type}\ (2,1)+(1,2).
\end{align*}
This can be understood as an integrability condition on the GC structure $\II$ induced by a data given by its partner metric $\GG$ in the commuting pair $(\GG,\II)$ and so we introduce the notion of weak integrability for an arbitrary commuting pair $(\GG,\Aa)$:
\begin{Def}
Let $(\GG,\Aa)$ be a commuting pair consisting of an indefinite generalized metric $\GG$ and arbitrary generalized structure $\Aa$ and let $D$ be the generalized Bismut connection of $\GG$. We say $\Aa$ is {\bf weakly integrable} when $D\Aa=0$.
\end{Def}
Note that it follows that in the above definition when $\Aa$ is weakly integrable then also $\Aa'=\GG\Aa$ is weakly integrable. We will now analyse what the condition $D\Aa=0$ means in terms of the tangent bundle data corresponding to $(\GG,\Aa)$. As we have seen previously, we get a pair of tangent bundle endomorphisms for any commuting pair $(\GG,\Aa)$ whenever $\GG$ is an (indefinite) generalized metric via the formula
\begin{align*}
A_\pm=\pm\pi_\pm\Aa\pi_\pm^{-1}.
\end{align*}
This can be inverted into a formula for $\Aa$ in terms of $A_\pm$:
\begin{align}\label{eq:JJ_from_J_pm}
\Aa=\pi_+^{-1}A_+\pi_+ P_++\pi_-^{-1}A_-\pi_- P_-,
\end{align}
where $P_\pm=\frac{1}{2}(\id\pm\GG)$ projects from $\TT$ to $C_\pm$. Using \eqref{eq:JJ_from_J_pm} and \eqref{eq:n_pm} We can now rephrase the equation $D\Aa=0$ in terms of $\n^\pm$ and $A_\pm$:

\begin{proposition}\label{prop:bismut_n_pm}
Let $(\GG,\Aa)$ be a commuting pair with $\GG$ a (indefinite) generalized metric and $D$ the generalized Bismut connection of $\GG$ given by \eqref{eq:n_pm}. Then $D\Aa=0$ if and only if $\n^\pm A_\pm=0$, $A_\pm$ being the tangent bundle endomorphisms corresponding to $\Aa$.
\end{proposition}
\begin{proof}
From $(D_u\Aa)v=D_u(\Aa v)-\Aa(D_u v)$ we get
\begin{align*}
(D_u\Aa)v&=\pi_+^{-1}\n_{\pi (u)}^+A_+(\pi_+v_+)+\pi_-^{-1}\n_{\pi (u)}^-A_-(\pi_-v_-)\\
&-\pi_+^{-1}A_+(\n_{\pi (u)}^+\pi_+v_+)-\pi_-^{-1}A_-(\n_{\pi (u)}^-\pi_-v_-),
\end{align*}
and combining the terms that take values in $C_\pm$ yields the result.
\end{proof}
As we mentioned above, the weak integrability condition is in the case of the $GK$ and $GpK$ commuting pairs simply a weakening of the usual integrability conditions (for $GpK$ formulated in Theorem \ref{theo:integr_GpK}, for $GK$ in \cite[Prop. 6.17]{Gualtieri:2003dx}):
\begin{proposition}\label{prop:weak_int_iso}
An almost $GK$ structure $(\GG,\II)$ is weakly integrable if and only if the fundamental forms $\omega_\pm$ of the induced bi-Hermitian data $(I_\pm,g)$ and the corresponding Nijenhuis tensors $N_{I_\pm}$ are related to the $H_b$-flux by
\begin{align}
\begin{aligned}
i(\rd\omega_\pm^{(3,0)_\pm}-\rd\omega_\pm^{(0,3)_\pm})&=\pm 3H_b^{(3,0)_\pm+(0,3)_\pm}=\pm\frac{3}{4}N_{I_\pm}\\
\rd\omega_\pm^{(2,1)_\pm+(1,2)_\pm}&=\mp i(H_b^{(2,1)_\pm}-H_b^{(1,2)_\pm}).
\end{aligned}
\end{align}

Similarly, an almost GpK structure $(\GG,\KK)$ is weakly integrable if and only if the fundamental forms $\omega_\pm$ of the induced bi-para-Hermitian data $(K_\pm,g)$ and the corresponding Nijenhuis tensors $N_{K_\pm}$ are related to the $H_b$-flux by
\begin{align}\label{eq:weak_int_GpK}
\begin{aligned}
\rd\omega_\pm^{(3,0)_\pm}&=\mp 3H_b^{(3,0)_\pm}=\pm\frac{3}{4} N^{(3,0)_\pm}_{K_\pm}\\
\rd\omega_\pm^{(0,3)_\pm}&=\pm 3H_b^{(0,3)_\pm}=\mp \frac{3}{4}N^{(0,3)_\pm}_{K_\pm}\\
\rd\omega_\pm^{(2,1)_\pm}&=\mp H_b^{(2,1)_\pm}\\
\rd\omega_\pm^{(1,2)_\pm}&=\pm H_b^{(1,2)_\pm}.
\end{aligned}
\end{align}
\end{proposition}
\begin{proof}
For the para-Hermitian case, the proof follows directly from Proposition \ref{prop:n_(para)herm} which is applied to the bi-para-Hermitian data $(\eta,K_\pm)$ by taking $K=K_\pm$ and correspondingly $h=\pm H_\pm$. For the Hermitian side, this follows in the same way from \cite[Prop. 6.24]{Gualtieri:2003dx}.
\end{proof}

\begin{lemma}
Let $\Aa_\pm$ be a commuting pair with $\GG=\Aa_+\Aa_-$ and $A_\pm$ the corresponding tangent bundle endomorphisms and let $T^D=2\pi_+^*H_b+2\pi^*_-H_b$ be the generalized torsion of the generalized Bismut connection of $\GG$. Then the $(k,l)+(l,k)$ ($k+l=3$) components with respect to both $\Aa_\pm$ vanish if and only if the $(k,l)+(l,k)$ components of $H_b$ with respect to $A_\pm$ vanish.
\end{lemma}
\begin{proof}
Using Proposition \ref{prop:gen_torsion} and denoting the splitting of sections of $\TT$ corresponding to eigenbundles of $\GG$ by $u=u_++u_-$, we have
\begin{align*}
\frac{1}{2}T^D(u,v,w)=H_b(\pi_+u_+,\pi_+v_+,\pi_+w_+)+H_b(\pi_-u_-,\pi_-v_-,\pi_-w_-).
\end{align*}
The relationships between $\Aa_\pm$ and $A_\pm$ are
\begin{align*}
\pi_+\Aa_\pm=A_+\pi_+,\quad \pi_-\Aa_\pm=\pm A_- \pi_-,
\end{align*}
which implies the following equations
\begin{align*}
\frac{1}{2}(T^D)^{(k,l)_+}=\pi_+^*H_b^{(k,l)_+}+\pi_-^*H_b^{(k,l)_-}\\
\frac{1}{2}(T^D)^{(k,l)_-}=\pi_+^*H_b^{(k,l)_+}+\pi_-^*H_b^{(l,k)_-},
\end{align*}
where the bigrading $(k,l)_\pm$ of $T^D$ is with respect to $\Aa_\pm$, while the bigrading of $H_b$ is with respect to $A_\pm$. The above set of equations relating $(T^D)^{(k,l)_\pm}$ and $H_b^{(k,l)_\pm}$ then implies the statement of the Lemma.
\end{proof}

We see from Proposition \ref{prop:weak_int_iso} that, indeed, when not requiring that the generalized torsion is of type $(2,1)+(1,2)$, the involved (para-)Hermitian structures need not be integrable and their integrability is controlled by the (now non-vanishing) $(3,0)+(0,3)$ components of $H_b$. To summarize, the situation is for GpK entirely analogous to the GK case:
\begin{theorem}\label{theo:integrabilityGpK}
Let $\KK_\pm$ be a commuting pair of generalized complex structures with metric $\GG=\KK_+\KK_-$ and $D$ its generalized Bismut connection. Then $(\GG,\KK_+)$ defines a GpK structure and in particular both $\KK_\pm$ are Courant integrable iff $D\KK_\pm=0$ and $T^D$ is of type $(2,1)+(1,2)$ with respect to both $\KK_\pm$.
\end{theorem}

\begin{remark}
In \cite{vaisman2015generalized}, Vaisman considers a commuting pair consisting of a generalized complex and a generalized para-complex structure $(\II,\KK)$ whose product is a nondegenerate generalized anti-complex structure (we call this generalized semi-K\"ahler geometry, see Remark \ref{rem:commutingpairs}). There, the integrability conditions on the tangent bundle data stemming from the integrability of the generalized (para-)complex structures is analogous to the presented case: the tangent bundle  (para-)Hermitian structures are forced to be integrable and parallel with respect to a connection with a fully skew torsion. This forces their fundamental form to be of type $(2,1)+(1,2)$ and related to the $H_b$-flux. In terms of the generalized Bismut connection, this would mean that $D\II=D\KK=0$ and $T^D$ is of type $(2,1)+(1,2)$.
\end{remark}

We now turn to non-isotropic generalized structures, particularly generalized chiral structures. In this case, we know that the results cannot be fully analogous because the corresponding tangent bundle geometry is very different; for example, the fundamental tensor of the tangent bundle chiral structure, $F(X,Y,Z)=\eta((\lc_XJ)Y,Z)$, is not fully skew and is of type $(2,1)+(1,2)$ (with respect to $J$) and so is the Nijenhuis tensor $N_J(X,Y,Z)=\eta(N_J(X,Y),Z)$. However, they can still be related to the flux $H_b$:

\begin{proposition}\label{prop:weak_int_chiral}
An {almost generalized chiral (anti-K\"ahler)} structure $(\GG,\JJ)$ is weakly integrable if and only if the fundamental tensors $F_\pm$ of the corresponding tangent bundle structures $(\eta,J_\pm)$ are related to the $H_b$-flux by
\begin{align}
F_\pm(X,Y,Z)=\mp\frac{1}{2}\left(H_b(X,J_\pm Y,Z) - H_b(X,Y,J_\pm Z)\right),
\end{align}
Equivalently, both $(g,J_\pm)$ are of type $\mathcal{W}_3$ almost product pseudo-Riemannian structures whose Nijenhuis tensors $N_\pm$ are related to $H$ by
	\begin{align*}
	N_\pm(X,Y,Z)=\pm 2\left(H_b^{(2,1)_\pm+(1,2)_\pm}(JX,Y,JZ) + H_b^{(2,1)_\pm+(1,2)_\pm}(X,J_\pm Y,J_\pm Z)\right)
	\end{align*}
\end{proposition}
\begin{proof}
The proof follows directly from the statement of Propositions \ref{prop:bismut_n_pm} and \ref{prop:DJ}, by replacing $J$ by $J_\pm$ and, simultaneously, $h$ by $\pm H_b$. 
\end{proof}

The properties of the fundamental tensor $F$ \eqref{eq:F_properties} imply that $H_b$ determines all non-zero components of $F$. Furthermore, in contrast to the G(p)K geometry where the weak integrability relates all components of $H_b$ to components of the fundamental forms $\omega_\pm$ and integrability of the tangent bundle structures is controlled by the $(3,0)+(0,3)$ parts, in the generalized chiral (anti-K\"ahler) case the weak integrability condition $D \JJ_\pm=0$ only fixes the $(2,1)+(1,2)$ components of $H_b$, which are also the components tied to integrability of the tangent bundle structures.
%
We therefore introduce the following definition of the (full) integrability for the non-isotropic case:
\begin{Def}
Let $(\GG,\JJ_+)$ be an almost generalized chiral (anti-K\"ahler) structure with $\JJ_-=\GG\JJ_+$. We say $(\GG,\JJ_+)$ is \textbf{integrable} when it is weakly integrable and the generalized torsion of the Bismut connection of $\GG$ is of type $(3,0)+(0,3)$ with respect to both $\JJ_\pm$.
\end{Def}

We then have the following statement, which follows from Corollary \ref{cor:W_0}: 
\begin{theorem}\label{theo:integrabilityGChGaK}
An almost generalized chiral (respectively, anti-K\"ahler) structure $(\GG,\JJ_+)$ is integrable if and only if the corresponding tangent bundle data $(\eta,J_\pm)$ are type ${\cal W}_0$ chiral 
(respectively, anti-Hermitian) structures.
\end{theorem}

\begin{Ex}
For examples of integrable generalized chiral structures, consider the Born structures of hyperK\"ahler-type given in example \ref{ex:BornHK}. Similarly, anti-hyperK\"ahler manifolds (see Appendix \ref{sec:anti-HH})
are examples of integrable anti-K\"ahler manifolds.
\end{Ex}

\section{Physical Interpretation and Relationship to Supersymmetry}
In this section we explain how the geometry introduced in this paper appears in physics in the context of 2D supersymmetric non-linear sigma models. The discussion is directed mostly at physicists and therefore we use physical terminology and introduce several objects customarily used in physics without explaining their exact meanings and definitions. For basics of supersymmetry (SUSY) and other details the reader can consult for example \cite{deligne1999quantum}.

\subsection{Twisted (2,2) SUSY and GpK Geometry}
In this subsection we explain how GpK geometry naturally appears in $2D$ $(2,2)$ supersymmetric sigma models, more concretely in {\it twisted} supersymmetric models introduced by Hull and Abou-Zeid in \cite{Hull_twisted_susy}. We compare this with the well-known story about how GK geometry appears in the usual $(2,2)$ supersymmetry. For a classic reference for the usual $(2,2)$ SUSY sigma models see \cite{Gates-hull-rocek-biherm} or a thesis \cite{susy_topological_thesis} containing many useful calculations.

We start by considering the general $(1,1)$ SUSY sigma model given by the action
\begin{align}\label{eq:(1,1)action}
S_{(1,1)}(\phi)=\int_{\hat{\Sigma}}[g(\phi)+b(\phi)]_{ij}D_+\phi^iD_-\phi^j,
\end{align}
where $\phi=(\phi_i)_{i=1\cdots n}$ are {\it fields}, i.e. maps $\phi: \hat{\Sigma} \rightarrow (M,g)$, where $\hat{\Sigma}$ is a super-Riemann surface with two formal odd directions, $(\theta_1,\theta_2)$, $(M,g)$ is (for now) arbitrary pseudo-Riemannian manifold and $b$ denotes a local two-form. $(1,1)$ supersymmetry means that this action is invariant under transformations generated by the two {\it supercharges} $Q_\pm$ obeying the supercommutation relations
\begin{align}\label{eq:susy}
\{Q_\pm,Q_\pm\}=P_\pm,
\end{align}
where $P_\pm$ are generators of translations.

The idea now is to study under which conditions the action $S_{(1,1)}$ admits additional supersymmetries. It turns out that this puts severe restrictions on the geometry of $M$. In particular, if we are to extend the supersymmetry to $(2,2)$, $(M,g)$ is necessarily a GK manifold. If the $(1,1)$ supersymmetry is to be extended to a {\it twisted} $(2,2)$ supersymmetry, the target manifold needs to be a GpK manifold, which we now explain in more detail.

In \cite{Hull_twisted_susy}, the authors derive that the $(2,2)$ twisted SUSY is equivalent to a bi-para-Hermitian geometry $(\eta,K_\pm)$ on the target $M$, which is then by the results of Section \ref{sec:GpK} equivelent to a GpK geometry. Here we briefly recall the arguments presented in \cite{Hull_twisted_susy} to show how the structures $K_\pm$ appear.

The extension of the $(1,1)$ SUSY to $(2,2)$ is  necessarily generated by the following transformations of the fields
\begin{align}\label{eq:extended_susy}
\delta\phi^i=\epsilon_+ (K_+)^i_jD_+\phi^j+\epsilon_- (K_-)^i_jD_-\phi^j,
\end{align}
for some (for now unspecified) tensors $K_\pm$. The requirement that the action \eqref{eq:(1,1)action} is invariant under this transformation forces the compatibility between $(g+b)$ and $K_\pm$:
\begin{align*}
g(K_\pm\cdot,\cdot)+g(\cdot,K_\pm\cdot)&=0\\
b(K_\pm\cdot,\cdot)+b(\cdot,K_\pm\cdot)&=0,
\end{align*}
along with the condition
\begin{align*}
\n^\pm K_\pm=0,
\end{align*}
where $\n^\pm$ are the connection defined in \eqref{eq:n_pm},
\begin{align*}
\n^\pm=\lc\pm \frac{1}{2}H.
\end{align*}
Here $\lc$ is the Levi-Civita connection of $g$ and $H$ is a closed global three-form, such that $b$ is locally its potential, $\rd b=H$\footnote{The expression \eqref{eq:(1,1)action} is local, which is the reason why the local two-form $b$ appears}. The additional requirement that the transformations \eqref{eq:extended_susy} indeed extend \eqref{eq:susy} to a $(2,2)$ twisted supersymmetry is equivalent to the conditions
\begin{align*}
K_\pm^2=\id\quad N_{K_\pm}=0,
\end{align*}
rendering $(g,K_\pm,b)$ a bi-para-Hermitian geometry, or equivalently, $M$ to be a GpK manifold.

When we require that the theory is {\it parity-symmetric}, we find that the $b$-field term in \eqref{eq:(1,1)action} has to vanish and additionally $K_+=K_-=T$, which gives the para-K\"ahler limit of the geometry. Additionally, one might require additional supersymmetry, which requires additional para-complex structure that anti-commutes with $T$, which is therefore desribed by the para-hyper-K\"ahler limit of GpK geometry. Various other heterotic $(p,1)$ supersymmetries can be realized as well, all as special cases of the GpK geometry.

We conclude this section by remarking that the integrability of $K_\pm$ can be relaxed \cite{SUSY_non-integrable}, giving the GpK (or GK in the case of usual SUSY) geometries which are only integrable in the weaker sense introduced in Section \ref{sec:bismut_integr}.

\subsection{(1,1) Superconformal Algebra and Generalized Chiral Geometry}
In \cite{Stojevic:2009ub}, it has been shown that the chiral geometry also plays an important role in introducing additional symmetries to $(1,1)$ sigma models \eqref{eq:(1,1)action}. While pairs of Hermitian and para-Hermitian structures naturally arise when considering an extended supersymmetry, pairs of chiral structures have different physical interpretation in terms of sigma models -- they correspond to introduction of additional copies of the $(1,1)$ superconformal algebra. Here we briefly review the results of \cite{Stojevic:2009ub}.

Consider a sigma model on a target $(M,g)$ given by the action \eqref{eq:(1,1)action}. For every such sigma model, there are so-called superconformal symmetries stemming from the fact that $M$ carries the metric $g$. The symmetries close to form an algebra, called a superconformal algebra. Now, it is shown in \cite{Stojevic:2009ub} that when $M$ admits two (almost-)project structures $J_\pm$ orthogonal with respect to $g$, that are also covariantly constant with respect to $\n^\pm$ \eqref{eq:n_pm},
\begin{align}\label{eq:bichiral}
g(J_\pm\cdot,J_\pm\cdot)=g,\quad \n^\pm J_\pm=0,
\end{align}
one can introduce additional symmetries $\delta_{P_\pm}$ and $\delta_{Q_\pm}$ associated to\footnote{We will not explain here how the symmetries are associated to the projectors $P_\pm$ and $Q_\pm$; we merely remark that the projectors are the only additional geometrical data entering the definitions of $\delta_{P_\pm}$ and $\delta_{Q_\pm}$.} the $+1$ and $-1$ projectors $P_\pm$ and $Q_\pm$, respectively
\begin{align*}
P_\pm=\frac{1}{2}(\id+ J_\pm),\quad Q_\pm=\frac{1}{2}(\id- J_\pm).
\end{align*}
The symmetries $\delta_{P_\pm}$ and $\delta_{Q_\pm}$ then form copies of the $(1,1)$ superconformal algebra. The conditions \eqref{eq:bichiral} are the only conditions on the tensors $J_\pm$, in particular there are no further requirements on integrability of $J_\pm$. By results of Section \ref{sec:gen_chiral} and Proposition \ref{prop:bismut_n_pm}, this means that $(J_\pm,g,b)$ defines a generalized chiral structure that is weakly integrable.

Because the additional symmetries $\delta_{P_\pm}$ and $\delta_{Q_\pm}$ form a superconformal algebra even when $J_\pm$ are not integrable, they lack a spacetime description in terms of a corresponding Riemannian manifold, contrary to the original algebra associated to $(M,g)$. The author of \cite{Stojevic:2009ub} then relates this fact to the existence of non-geometric string backgrounds.

\section{Conclusions and Outlook}
In this paper, we explored the properties of various commuting pairs of generalized structures, which was a natural next step in the study of generalized geometry. We showed that GpK geometry is -- as expected -- the para-complex analog of GK geometry and proved the basic GpK results that are well known in the GK case. Note that the similarity between GK and GpK geometries is also manifested in physics in the study of supersymmetric sigma models.

We also showed that Born geometry naturally fits into the framework of commuting pairs as the anti-commuting subcase of generalized chiral geometry. Additionally, we were able to circumvent the well-known issue of missing integrability conditions for non-isotropic generalized structures and define the integrability of non-isotropic commuting pairs in terms of the generalized Bismut connection, recovering results analogous to the isotropic case.

There are several natural directions for continuing this work. The first is to further explore GpK geometry and to further reproduce well-established results from GK geometry to the GpK setting. This includes the study of the para-holomorphic reduction of GpK geometry, which could yield an elegant way of imposing the ``section condition" -- one of the important problems in Double Field Theory -- for the corresponding para-Hermitian structures. Furthermore, GpK geometry can be used to further the results in \cite{Hull_twisted_susy} and study the twisted supersymmetric sigma models. In forthcoming work \cite{toappear}, the topological twists of such sigma models are explored. Nonetheless, there are still many questions to answer beyond this, particularly the notion of mirror symmetry for these sigma models, the relationship with T-duality on the underlying para-Hermitian manifolds as well as with the usual mirror symmetry, both its SYZ and homological incarnations.

Secondly, the realization of Born geometry in terms of the generalized chiral structures opens up a new point of view on this geometry. From the physics point of view, Born geometry is typically seen as an underlying (almost) para-Hermitian structure along with a choice of a metric structure on one of its eigenbundles. In \cite{Freidel:2018tkj}, it is shown that this fully determines a Born geometry and any Born structure is of this form. The para-Hermitian structure then gives rise to the {\it T-duality frame}, which is a local splitting of the base manifold, that represents the {\it extended space-time}), into the usual space-time with a metric (along the $+1$-eigendirections of the para-Hermitian structure) and its T-dual counterpart (along the $-1$-eigendirections). However, from the generalized chiral structures viewpoint, it is natural to see Born geometry as a pair of anti-commuting chiral structures and, in particular, the integrability conditions are formulated this way. The immediate question is therefore: what is the physical interpretation of this bi-chiral geometry and what do the integrability conditions imply for physics?

Lastly, after describing the commuting pairs, the natural next step is to discuss pairs of commuting pairs of generalized structures. This has been already done for pairs of GK structures that anti-commute, that is, two sets of commuting pairs of generalized complex structures $\II_\pm$, $\II'_\pm$ such that
\begin{align*}
\{\II_\pm,\II'_\pm\}=0.
\end{align*}
This yields generalized hyper-K\"ahler geometry \cite{Bredthauer:2006sz}, or equivalently bi-hyper-Hermitian geometry on the tangent bundle, describing targets of $(4,4)$ supersymmetric sigma models. One can therefore expect that the para-complex story will be similar: by considering a pair of GpK structures that mutually anticommute, that is, a generalized para-Hyper-K\"ahler geometry, one recovers a bi-para-hyper-Hermitian geometry on the tangent bundle, describing the targets of $(4,4)$ twisted supersymmetric sigma models. One can also consider a GpK structure $\KK_\pm$ and a generalized chiral structure $\JJ_\pm$ such that $\{\KK_\pm,\JJ_\pm\}=0$. This on the tangent bundle yields a pair of Born geometries. Clearly, there is a considerable amount of various combinations of commuting pairs that will potentially yield interesting geometries and results.

\appendix
\section{Product and para-complex structures}

	\subsection{Definitions}
	\label{sec:paracpx/product}	
	In this section, we briefly review important notions in para-complex geometry. For more details, see for example \cite{Cortes:2010ykx,Cortes:2003zd}, or the survey on paracomplex geometry \cite{Cruceanu} and the references therein.
	\begin{Def}\label{def:paracpx}
		An \textbf{almost product structure} on a smooth manifold $\PS$ is a smooth endomorphism $K\in \Endo(T\PS)$ that squares to the identity: $K^2=\id_{T\PS}$. An \textbf{almost para-complex structure} is a product structure $K$ whose $+1$- and $-1$-eigenbundles, denoted $T^{(1,0)}$ and $T^{(1,0)}$, respectively, have the same rank. Finally, an \textbf{almost product/para-complex manifold} is a manifold $\mathcal{P}$ endowed with a product/para-complex structure $K$, which we denote $(\mathcal{P},K)$.
	\end{Def}
	A direct consequence of above definition is that any para-complex manifold is even-dimensional. The use of the word \textbf{almost} as usual refers to integrability of the endomorphism, that is, whether its eigenbundles are involutive under the Lie bracket and therefore define a foliation of the underlying manifold. Similarly to the complex case, the integrability of an almost product/para-complex structure $K$ is governed by the \textbf{Nijenhuis tensor $N_K$}, which is given by
	\begin{align}\label{eq:nijenhuis}
	\begin{aligned}
	N_K(X,Y) := & \; [X,Y]+[KX,KY]-K([KX,Y]+[X,KY])\\
	=& \; (\n_{KX}K)Y+(\n_XK)KY-(\n_{KY}K)X-(\n_YK)KX\\
	=& \; 4(P_\pm[P_\mp X,P_\mp Y]+P_\mp[P_\pm X,P_\pm Y]),
	\end{aligned}
	\end{align}
	for all $X,Y \in \Gamma(T\mathcal{P})$, where $\n$ is any torsionless connection and $P_\pm\coloneqq\frac{1}{2}(\id\pm K)$ is projection onto the $\pm 1$-eigenbundle. Moreover, we say that $K$ is \textbf{integrable} and call it a \textbf{product/para-complex structure} if $N_K=0$, in which case, $(\mathcal{P},K)$ is a \textbf{product/para-complex manifold}. From \eqref{eq:nijenhuis}, it is apparent that $K$ is integrable if and only if {\it both} its eigenbundles are simultaneously Frobenius integrable (that is, involutive distributions in $T\PS$); the integrability of one of the eigenbundles is, however, not tied to the integrability of the other. This is one of the main differences between complex geometry and para-complex geomery: while in the complex case the eigenbundles are complex bundles related by complex conjugation, here the eigenbundles are real and therefore one can be integrable while the other is not. We call this phenomenon {\bf half-integrability}. More on this can be found for example in \cite{freidel2017generalised,Svoboda:2018rci} or \cite{Hassler:2019wvn}, where examples of half-integrable para-complex structures motivated by physics are given.

\subsection{Adapted coordinates and the Dolbeault complex}
	
	Let now $(\mathcal{P},K)$ be an almost para-complex manifold. If $K$ is integrable, we get a set of $2n$ coordinates $(x^i,\tilde{x}_i)$ called {\bf adapted coordinates}, $\mathcal{P}$ locally splits as $M\times \tilde{M}$, and $K$ acts as identity on $TM=T^{(1,0)}$ and negative identity on $T\tilde{M}=T^{(0,1)}$ (see for example \cite{Cortes:2003zd}). The splitting of the tangent bundle gives rise to a decomposition of tensors analogous to the $(p,q)$-decomposition in complex geometry. Denote $\Lambda^{(k,0)}(T^*\mathcal{P})\coloneqq \Lambda^k(L^*)$ and $\Lambda^{(0,k)}(T^*\mathcal{P})\coloneqq\Lambda^k(\Lt^*)$. The splitting is then
	\begin{align}\label{eq_plusminus_decomp}
	\Lambda^k (T^*\mathcal{P})=\bigoplus_{k=m+n}\Lambda^{(m,n)}(T^*\mathcal{P}),
	\end{align}
	with corresponding sections denoted as $\Omega^{(m,n)}(\mathcal{P})$. The bigrading \eqref{eq_plusminus_decomp} yields the natural projections
	\begin{align*}
	\Pi^{(p,q)}:\Lambda^k(T^*\mathcal{P})\rightarrow \Lambda^{(p,q)}(T^*\mathcal{P}),
	\end{align*}
	so that the de-Rham differential splits as $\rd=\p^{(1,0)}+\p^{(0,1)}$, where
	\begin{align*}
	&\p^{(1,0)} \coloneqq \Pi^{(p+1,q)}\circ \rd\\
	&\p^{(0,1)} \coloneqq \Pi^{(p,q+1)}\circ \rd,
	\end{align*}
	are the \textbf{para-complex Dolbeault operators}, acting on forms as
	\begin{align}\label{eq:partials_plusminus}
	\begin{aligned}
	\p^{(1,0)}&:\Omega^{(p,q)}(\mathcal{P})\rightarrow \Omega^{(p+1,q)}(\mathcal{P})\\
	\p^{(0,1)}&:\Omega^{(p,q)}(\mathcal{P})\rightarrow \Omega^{(p,q+1)}(\mathcal{P}),
	\end{aligned}
	\end{align}
	such that when $K$ is integrable, we have
	\begin{align*}
	(\p^{(0,1)})^2= 0, (\p^{(1,0)})^2=0, \text{ and } \p^{(1,0)}\p^{(0,1)}+\p^{(0,1)}\p^{(1,0)}=0.
	\end{align*}
	
	We also introduce the {\it twisted differential $\rd^p\coloneqq(\Lambda^{k+1}K)\circ\rd\circ (\Lambda^kK)$}:
	\begin{lemma}
		Let $(\PS,K)$ be a paracomplex manifold. Then $\rd^p\coloneqq(\Lambda^{k+1}K)\circ\rd\circ (\Lambda^kK)$ can be expressed as
		\begin{align}\label{eq:dp-operator}
		\rd^p=\p^{(1,0)}-\p^{(0,1)}.
		\end{align}
	\end{lemma}
	\begin{proof}
		Let $\ap \in \Omega^{m,n}(\PS)$. Then we have
		\begin{align*}
		\rd^p\ap=(-1)^n(\Lambda^kK)\rd\ap=(-1)^{2n}\p^{(1,0)} \ap +(-1)^{2n+1}\p^{(0,1)}\ap=(\p^{(1,0)}-\p^{(0,1)})\ap,
		\end{align*}
	\end{proof}
\subsection{Para-holomorphic functions and bundles}
We will now explore the para-holomorphic structure of para-complex manifolds, and give important examples of para-holomorphic vector bundles.

We start with the natural definition of a para-Holomorphic map between para-complex vector spaces.
\begin{Def}
	Let $(M,K_M)$ and $(N,K_N)$ be para-complex manifolds. A map $f:M\rightarrow N$ is called para-holomorphic if
	\begin{align*}
	K_N\circ f_*=f_*\circ K_M
	\end{align*}
\end{Def}

Locally, the definition means the following. Let $V$ and $W$ be $2n$- and $2m$-dimensional vector spaces, respectively. Choose the respective adapted bases for $V$ and $W$ as $\{v_i,\tl{v}^j\}_{i,j=1\cdots n}$, $\{v_k,\tl{v}^l\}_{k,l=1\cdots m}$, so that $K_V$ and $K_W$ take the diagonal forms
$\begin{pmatrix}
\id & 0 \\
0 & -\id
\end{pmatrix}$. 
It is easy to check that a para-holomorphic map $f:V\rightarrow W$ then takes the form
\begin{align*}
f=(w_1(v_i),\cdots,w_m(v_i),\tl{w}^1(\tl{v}^i),\cdots,\tl{w}^m(\tl{v}^i)),
\end{align*}
i.e. the first $m$ components of $f$ are independent of the $\tl{v}$ variables, while the remaining last $m$ components are independent of the $v$ variables, meaning $(w_i,\tl{w}^j)$ satisfy the para-complex Cauchy-Riemann equations:
\begin{align}\label{eq:CR}
\frac{\p}{\p v_i} \tl{w}^j=\frac{\p}{\p \tl{v}^i} w_j=0,\quad \forall i=1,\cdots, n,\ j=1,\cdots, m.
\end{align}
Therefore, a para-holomorphic map $f$ of para-complex manifolds $(M,K_M)$ and $(N,K_N)$ takes on each pair of patches $U\subset M$ and $V\subset N$ the local form
\begin{align*}
f:\RR^{2n}\rightarrow \RR^{2m}:\ f=(y(x),\yt(\xt)),
\end{align*}
where $(x,\xt):U\rightarrow \RR^{2n}$ and $(y,\yt): V\rightarrow \RR^{2m}$ are the local adapted coordinates on $U$ and $V$.

Now, because on a $2n$-dimensional para-complex manifold, the adapted coordinates $(x,\xt)$ patch into two separate foliations which can be seen as two $n$-dimensional manifolds, the coordinates along these manifolds transform among themselves, which means the transition functions $\phi_{UV}:U\mid_{U\cap V}\rightarrow V\mid_{U\cap V}$ between two patches $U$ and $V$ have to have the form $(x'(x),\xt'(\xt))$, $(x,\xt)$ and $(x',\xt')$ being the adapted coordinates on $U$ and $V$, respectively. The transition functions on a para-complex manifold are therefore easily seen to be para-holomorphic functions, and we call this the para-holomorphic structure of the manifold.

Let us now explore para-holomorphic vector bundles, starting from the following definition
\begin{Def}
	A para-holomorphic vector bundle $E\xrightarrow{\pi} M$ over a para-complex manifold $M$ with is a para-complex vector bundle (i.e. the fibers are para-complex vector spaces $V$), such that its transition functions $g_{UV}:U\cap V\rightarrow GL(V, K_V)$ are para-holomorphic maps. A para-holomorphic section of $E$ is a section of the projection $\pi$ that is a para-holomorphic map.
\end{Def}

Here, $(V,K_V)$ is an even-dimensional vector space with the diagonal para-complex structure $K_V=\begin{pmatrix}
\id & 0 \\
0 & -\id
\end{pmatrix}$
and $GL(V,K_V)$ is the structure group preserving $K_V$. Denoting by $V_\pm$ the eigenbundles of $K_V$, it is easy to see that $GL(V,K_V)$ is simply given by two two copies of $GL$ for each eigenbundle,  $GL(V, K_V) \cong GL(V_+) \times GL(V_-)$ and additionally that $GL(V,K_V)$ itself has a para-complex structure, that acts by identity on the first factor and by minus identity on the second factor.

\begin{Ex}[Tangent bundle of a para-complex manifold]\label{ex:parahol_TM}
	Let $(M,K)$ be a $2n$-dimensional para-complex manifold. Its tangent bundle $TM$ is a para-holomorphic vector bundle. We noted above that due to the para-holomorphic structure of $M$, the gluing functions between two patches $\phi_{UV}:U\mid_{U\cap V}\rightarrow V\mid_{U\cap V}$ take the form $(x,\xt)\mapsto (y(x),\tl{y}(x))$. The transition function for $TM$ is then given by the push-forward of   $\phi_{UV}$:
	\begin{align*}
	(\phi_{UV})_*: U\mid_{U\cap V}\times \RR^{2n}&\rightarrow V\mid_{U\cap V}\times \RR^{2n}\\
	\left(x^i,\xt_i,\frac{\p}{\p x^i}=\p_i,\frac{\p}{\p \xt_i}=\pt^i\right)&\mapsto \left(y^j(x^i),\tl{y}_j(\xt_i),\frac{\p y^j}{\p x^i}\frac{\p}{\p y^j},\frac{\p \tl{y}_j}{\p \xt_i}\frac{\p}{\p \tl{y}_j}\right),
	\end{align*}
	which can also be seen as a map $g_{UV}:U\cap V\rightarrow GL_n^+ \times GL_n^- \cong GL(V, K_V)$
	\begin{align*}
	g_{UV}:(x,\xt)\mapsto 
	\begin{pmatrix}
	\left(\frac{\p y(x)}{\p x}\right) & 0 \\
	0 & \left(\frac{\p \tl{y}(\xt)}{\p \xt}\right)
	\end{pmatrix},
	\end{align*}
	which is para-holomorphic with respect to the para-complex structure of $GL_n^+ \times GL_n^-$ that acts diagonally by $\pm \id$ on the copies $GL_n^\pm$.
	
	The reason $TM=T^{(1,0)}\oplus T^{(0,1)}$ is para-holomorphic is because we can see the eigenbundles $T^{(1,0)}$ and $T^{(0,1)}$ as tangent bundles of two foliations $\FF_\pm$: $T^{(1,0)}=T\FF_+$, $T^{(0,1)}=T\FF_-$, which can be understood as individual $n$-dimensional manifolds. It is therefore clear that each factor in the sum $T^{(1,0)}\oplus T^{(0,1)}$ transforms with transition functions  only depending on the coordinates of the corresponding foliation manifold.
	
	\begin{Ex}[Wedge powers of $T^{(1,0)}$ and $T^{(0,1)}$]\label{ex:parahol_wedgepowers}
		Consider now the vector bundle $E=\Lambda^{k}(T^{(1,0)})\oplus \Lambda^k(T^{(0,1)})$ with sections the poly-vector fields $\XX^{(k,0)+(0,k)}$ for some $1<k<n$ over a $2n$-dimensional para-complex manifold $(\PS,K)$. From the discussion in Example \ref{ex:parahol_TM} we can see that the transition functions of such bundle is going to be given by $g_{UV}=\Lambda^k\left(\frac{\p y(x)}{\p x}\right)\oplus \Lambda^k \left(\frac{\p \tl{y}(\xt)}{\p \xt}\right)$, acting on the basis vectors $(e_I,e^I)$,
		\begin{align*}
		e_I=\p_{i_1}&\w\cdots\w\p_{i_k},\  \tl{e}^I=\pt^{i_1}\w\cdots\w\pt^{i_k},\\
		0&<i_1<\cdots<i_k<n,\ I=1,\cdots,r= {n \choose k},
		\end{align*}
		where the fibre para-complex structure $K_E$ is given by
		\begin{align*}
		K_E(e_I)=e_I,\quad K_E(\tl{e}^I)=-\tl{e}^I.
		\end{align*}
		Again, such vector bundle is para-holomorphic for the same reasons $TM=L\oplus \Lt$ itself is para-holomorphic. Let $\sigma:M\rightarrow E$ be a section of the bundle $E$. Expanding $\sigma$ in a local coordinates, we get
		\begin{align*}
		\sigma=a(x,\xt)^{i_1\cdots i_r}\p_{i_1}\w\cdots\w\p_{i_r}+\tl{a}(x,\xt)_{j_1\cdots j_r}\pt^{j_1}\w\cdots\w\pt^{j_r}.
		\end{align*}
		It is easy to see that in order for $\sigma$ to be a holomorphic section, the coefficient functions $(a,\tl{a})$ have to satisfy
		\begin{align}\label{eq:parahol_section}
		\pt^i a(x,\xt)=\p_i \tl{a}(x,\xt)=0,
		\end{align}
		meaning the coefficient functions satisfy the para-complex Cauchy-Riemann equations \eqref{eq:CR}.
	\end{Ex}
\end{Ex}

\section{Structures with a compatible metric}
We have reviewed para-complex and product structures and now we will add into the discussion an appropriately compatible metric structure. For the compatibility we now have two options: either the case of {\bf para-Hermitian geometry}, where we have a para-complex structure anti-orthogonal with respect to the metric, $\eta(K\cdot,K\cdot)=-\eta$, in which case $\eta$ is necessarily of split signature $(n,n)$, or we can consider the case of {\bf chiral geometry}\footnote{Such structures are in mathematics literature typically called {\it pseudo-Riemannian almost product structures}, here we invoke a terminology from physics reflecting the fact that in string theory this type of structure determines the chiral right and left moving sectors.}, which consists of an (almost) product structure $J$ orthogonal with respect to a metric $g$, $g(J\cdot,J\cdot)=g$. In the present discussion, we will mostly encounter cases, where $J$ is para-complex but in general the eigenbundles of $J$ need not have same rank.

When we start with a complex structure $I$, on the other hand, we acquire either Hermitian geometry, for which the underlying metric structure $g$ is Riemannian and $I$ is orthogonal, $g(I\cdot,I\cdot)=g$, or {\bf anti-Hermitian geometry}, where $I$ is anti-orthogonal, $\eta(I\cdot,I\cdot)=-\eta$ and the metric structure $\eta$ need to be of split signature.

Being most studied, we start by briefly recalling main definitions of Hermitian geometry. For more details, we refer the reader for example to \cite{moroianu2007lectures}.

\begin{Def}
	Let $(\mathcal{P},g)$ be a pseudo-Riemannian manifold and $I$ an (almost) complex structure orthogonal with respect to $g$:
	\begin{align*}
	g(IX,IY)=g(X,Y)
	\end{align*}
	for all $X,Y\in \Gamma(T\mathcal{P})$. We call $(\mathcal{P},I,g)$ an \textbf{(almost) pseudo-Hermitian manifold}. If $g$ is positive-definite, we simply say that $(\mathcal{P},I,g)$ is an \textbf{(almost) Hermitian manifold}.
\end{Def}

For any almost pseudo-Hermitian structure $(g,I)$ on $\mathcal{P}$, the contraction $\omega:=gI$ is a non-degenerate two-form called the \textbf{fundamental form} of $(g,I)$.
\begin{Def}
	An (almost) pseudo-Hermitian $(\mathcal{P},I,g)$ manifold is called \textbf{(almost) pseudo-K\"ahler manifold} if its fundamental form $\omega$ is closed: $d\omega = 0$.
\end{Def}

We also recall the following properties of connections with fully skew torsion that are compatible with a Hermitian structure $(g,I)$, which were proven in \cite{Gualtieri:2003dx}:

\begin{proposition}\label{prop:n_herm}
	Let $(g,I)$ be an almost Hermitian structure with fundamental form $\omega=g I$. Also let $\lc$ be the Levi-Civita connection of $g$ and $h$ be a 3-form. 
	Set 
	\[\n^h=\lc+\frac{1}{2}g^{-1}h.\] This is a metric connection with torsion $g^{-1}h$ and the following are equivalent:  
	\begin{enumerate}
		\item 
		$\n^h I=0$.
		\item 
		$N_I = -4g^{-1}h^{(3,0)+(0,3)}$ and
		$\rd\omega^{(2,1)+(1,2)} = -ih^{(2,1)}+ih^{(1,2)}$.
	\end{enumerate}
\end{proposition}

If we also require $I$ to be integrable, we obtain:
\begin{corollary}
	Let $h$ be any 3-form on an almost Hermitian manifold $(\mathcal{P},g,I)$ with fundamental form $\omega=gI$. Moreover, let $\n^h=\lc+\frac{1}{2}g^{-1}h$, where $\lc$ be the Levi-Civita connection of $g$. Then, the following are equivalent:
	\begin{enumerate}
		\item 
		$\n^h I=0$ and $h$ is of type (2,1)+(1,2).
		\item 
		$I$ is integrable and $\rd^c\omega =-h$.
	\end{enumerate}
\end{corollary}

\subsection{Para-Hermitian structures}
Para-Hermitian geometry should be thought of as the para-complex version of Hermitian geometry, i.e. an additional metric structure compatible with the para-complex structure is introduced. This metric then induces a non-degenerate two-form, which can be closed, giving rise to a para-K\"ahler geometry.
\begin{Def}
	Let $(\PS,K)$ be a para-complex manifold and let $\eta$ be a pseudo-Riemannian metric that satisfies $\eta(K\cdot,K\cdot)=-\eta$. Then we call $(\PS,K,\eta)$ a \textbf{para-Hermitian manifold}\footnote{If $K$ is not integrable, i.e. $(\PS,K)$ is almost para-complex, we would call $(\PS,K,\eta)$ an almost para-Hermitian manifold.}.
\end{Def}
The above definition implies that the tensor $\omega\coloneqq \eta K$ is skew
\begin{align*}
\omega(X,Y)=\eta(KX,Y)=-\eta(X,KY)=-\omega(Y,X),
\end{align*}
and nondegenrate (because $\eta$ is nondegenerate), therefore $\omega$ is an almost symplectic form, sometimes called the \textbf{fundamental form}. From $K^2=\id$ we also have $K=\eta^{-1}\omega=\omega^{-1}\eta$. Another observation is that since the eigenbundles of $K$ have the same rank, $\eta$ has split signature $(n,n)$. Furthermore, the eigenbundles of $K$ are isotropic with respect to both $\eta$ and $\omega$. This means that the almost symplectic form $\omega$ is of the type $(1,1)$, $\omega \in \Omega^{(1,1)}$.
\begin{remark*}
	As shown above, the data $(\PS,K,\eta)$, $(\PS,\eta,\omega)$ and $(\PS,K,\omega)$ are on a para-Hermitian manifold equivalent and so we may use the different triples interchangeably to refer to a para-Hermitian manifold.
\end{remark*}
\begin{Def}
	Let $(\PS,\eta,\omega)$ be a para-Hermitian manifold with $\rd\omega=0$. We call $(\PS,\eta,\omega)$ a \textbf{para-K\"ahler manifold}.
\end{Def}

\begin{Ex}[Local structures]\label{local_paraH}
	Almost para-Hermitian structures all look the same locally.
	Indeed, let $(\mathcal{P},K,\eta)$ be a $2n$-dimensional almost para-Hermitian manifold. 
	Bejan then shows \cite{bejan1993existence} that there exist local frames of $T\mathcal{P}$ with respect to which 
	\begin{align*}
	K=
	\begin{pmatrix}
	0 & \id_n \\
	\id_n & 0
	\end{pmatrix},\quad
	\eta=
	\begin{pmatrix}
	\id_n & 0 \\
	0 & -\id_n
	\end{pmatrix},
	\end{align*}
	or
	\begin{align*}
	K=
	\begin{pmatrix}
	\id_n & 0 \\
	0 & -\id_n
	\end{pmatrix},\quad
	\eta=
	\begin{pmatrix}
	0 & \id_n \\
	\id_n & 0
	\end{pmatrix},
	\end{align*}
	where $\id_n$ is the $n \times n$ identity matrix.
\end{Ex}

We will need the following property concerning connections with a fully skew torsion compatible with a para-hermitian structure $(\eta,K)$:

\begin{proposition}\label{prop:n_(para)herm}
	Let $(\eta,K)$ be an almost para-Hermitian structure with fundamental form $\omega=\eta K$. Also let $\lc$ be the Levi-Civita connection of $\eta$ and $h$ be a 3-form. 
	Set $\n^h=\lc+\frac{1}{2}\eta^{-1}h$. This is a metric connection with torsion $\eta^{-1}h$ and the following are equivalent:  
	\begin{enumerate}
		\item 
		$\n^h K=0$.
		\item 
		$(\eta,K)$ and $h$ satisfy the equations:
		\begin{align}\label{eq:h-components}
		\begin{aligned}
		N &= -4h^{(3,0)+(0,3)} \\
		\rd\omega^{(3,0)+(0,3)} &= -3h^{(3,0)}+3h^{(0,3)} \\
		\rd\omega^{(2,1)+(1,2)} &= -h^{(2,1)}+h^{(2,1)}.
		\end{aligned}
		\end{align}
	\end{enumerate}
\end{proposition}
\begin{proof}
	Let us first note that, by definition of $\n^h$, 
	\[ \eta((\n^h_XK)Y,Z) = \eta((\lc_XK)Y,Z) -\frac{1}{2}\left(h(X,KY,Z) - h(X,Y,KZ)\right)\]
	for any $X,Y,Z \in \Gamma(T)$. Consequently, $\n^h K=0$ if and only if 
	\begin{align}\label{LC}
	\eta((\lc_XK)Y,Z)=-\frac{1}{2}\left(h(X,KY,Z) + h(X,Y,KZ)\right).
	\end{align}
	Moreover, recall that $N$ and $\rd\omega$ can be expressed in terms of $\lc$ as follows:
	\begin{align}\label{N:LC}
	N(X,Y,Z) = \eta((\lc_{KX}K)Y - (\lc_{KY}K)X + (\lc_XK)KY - (\lc_YK)KX,Z)
	\end{align}
	and 
	\begin{align}\label{d:LC}
	\rd \omega(X,Y,Z)=\sum_{Cycl. X,Y,Z} \eta((\lc_XK)Y,Z)
	\end{align}
	for all $X,Y,Z \in \Gamma(T)$.
	
	Let us first assume that $\n^h K=0$. Then, substituting \eqref{LC} in the expressions of $N$ and $\rd\omega$, we obtain:
	\begin{align*}
	\begin{aligned}
	N(X,Y,Z) & = -h(KX,Y,KZ) - h(X,KY,KZ) - h(KX,KY,Z) - h(X,Y,Z)\\
	& = -4h^{(3,0)+(0,3)}(X,Y,Z)
	\end{aligned}
	\end{align*}
	and
	\begin{align*}
	\rd\omega = -h(KX,Y,Z) - h(X,KY,Z) - h(X,Y,KZ) 
	\end{align*}
	for all $X,Y,Z \in \Gamma(T)$. This implies that $N = -4\eta^{-1}h^{(3,0)+(0,3)}$, $\rd\omega^{(3,0)+(0,3)} = -3h^{(3,0)} + 3h^{(0,3)}$ and
	$\rd\omega^{(2,1)+(1,2)} = -h^{(2,1)}+h^{(1,2)}$. 
	
	For the converse, we note that formulas \eqref{N:LC} and \eqref{d:LC} imply that
	\begin{align}
	\eta((\lc_XK)Y,Z) = \frac{1}{2}\left(\rd\omega(X,Y,Z) + \rd\omega(X,KY,KZ) - N(Y,KZ,X) \right) 
	\end{align}
	for all $X,Y,Z \in \Gamma(T)$.
	Consequently, if $N = -4\eta^{-1}h^{(3,0)+(0,3)}$, $\rd\omega^{(3,0)+(0,3)} = -3h^{(3,0)} + 3h^{(0,3)}$ and
	$\rd\omega^{(2,1)+(1,2)} = -h^{(2,1)}+h^{(1,2)}$, 
	\[ 	\eta((\lc_XK)Y,Z) = -\frac{1}{2}\left(h(X,KY,Z) + h(X,Y,KZ)\right).\]
	In other words, \eqref{LC} holds, proving that $\n^hK = 0$.
\end{proof}

If we also require $K$ to be integrable, we obtain:
\begin{corollary}
	Let $(\eta,K)$ be an almost para-Hermitian structure with fundamental form $\omega=\eta K$. Moreover, let $\lc$ be the Levi-Civita connection of $\eta$ and $h$ be any 3-form. 
	Set $\n^h=\lc+\frac{1}{2}\eta^{-1}h$. Then, the following are equivalent:
	\begin{enumerate}
		\item 
		$\n^h K=0$ and $h$ is of type (2,1)+(1,2).
		\item 
		$K$ is integrable and $\rd^p\omega =-h$.
	\end{enumerate}
\end{corollary}

\subsection{Chiral and anti-Hermitian structures}\label{sec:chiral/antiH}

In this section, we consider pairs $(J,\eta)$ consisting of an (almost) complex or product structure $J$ together with a metric $\eta$ with respect to which $J$ is an anti-isometry or isometry, respectively, giving rise to anti-Hermitian or chiral structures, respectively.

We first consider (almost) anti-Hermitian structures.
\begin{Def}
	An \textbf{(almost) anti-Hermitian structure}\footnote{(Almost) anti-Hermitian manifolds are sometimes called (almost) K\"ahler-Norden manifolds \cite{norden2}, (almost) complex manifolds with Norden metric \cite{norden1} or (almost) generalized B-manifolds \cite{B-manifolds} in the literature.} is given by a pair $(J,\eta)$ where $J$ is an (almost) complex structure and $\eta$ is a pseudo-Riemannian metric such that $\eta(J\cdot,J\cdot)=-\eta(\cdot,\cdot)$. An \textbf{(almost) anti-Hermitian manifold} is manifold endowed with an (almost) anti-Hermitian structure.
\end{Def}

Note that the compatibility condition $\eta(J\cdot,J\cdot)=-\eta(\cdot,\cdot)$ along with the fact that $J$ is almost complex implies that $\eta$ is necessarily of split signature $(n,n)$. Moreover, contrary to Hermitian geometry, $J$ is an {\it anti-isometry} of $\eta$ as opposed to an isometry so that the tensor $\eta J$ is again a pseudo-Riemannian structure as opposed to a two-form.

\begin{Ex}[Local structures]\label{local_AntiH}
	One proves as in the para-Hermitian case \ref{local_paraH} that anti-Hermitian structures all have similar local descriptions.
	To be precise, any $2n$-dimensional almost anti-Hermitian manifold $(\mathcal{P},J,\eta)$ admits local frames of $T\mathcal{P}$ with respect to which 
	\begin{align*}
	J=
	\begin{pmatrix}
	0 & -\id_n \\
	\id_n & 0
	\end{pmatrix},\quad
	\eta=
	\begin{pmatrix}
	\id_n & A \\
	A & -\id_n
	\end{pmatrix},
	\end{align*}
	where $\id_n$ is the identity $n \times n$ matrix and $A$ is a symmetric matrix such that $(\det A)^2 \neq (-1)^n$ at every point.
\end{Ex}

We now consider the case of (almost) chiral structures.
\begin{Def}
	An \textbf{(almost) chiral structure}\footnote{The name chiral comes from physics, in mathematical literature such structures are called (almost) product pseudo-Riemannian.} is given by a pair $(J,\eta)$ where $J$ is an (almost) product structure and $\eta$ is a pseudo-Riemannian metric such that $\eta(J\cdot,J\cdot)=\eta(\cdot,\cdot)$.
\end{Def}

Contrary to para-Hermitian geometry, the structure $J$ need {\it not} be para-complex and $J$ is an {\it isometry} of $\eta$ as opposed to an anti-isometry. A consequence of this is that the tensor $\eta J$ is now again a pseudo-Riemannian structure as opposed to a two-form. Moreover, the $+1$-eigenbundle of $J$ is orthogonal to its $-1$-eigenbundle with respect to $\eta$. 

As for almost para-Hermitian and anti-Hermitian structures (see examples \ref{local_paraH} and \ref{local_AntiH}), one has a canonical local description of almost chiral structures:


\begin{Ex}[Local structures]
	Let $(\mathcal{P},J,\eta)$ be an almost chiral manifold of dimension $m$ with $T^{(1,0)}$ of rank $s$ and $T^{(0,1)}$ of rank $r$ (so that $s+t = m$). Suppose that the restriction of $\eta$ to $T^{(1,0)}$ (respectively, $T^{(0,1)}$) has signature $(l,s-l)$ (respectively, $(k,t-k)$).  
	Then, 
	\[
	J=
	\begin{pmatrix}
	\id_s & 0 \\
	0 & -\id_t
	\end{pmatrix},\quad
	\eta=
	\begin{pmatrix}
	\begin{pmatrix}
	\id_l & 0 \\
	0 & -\id_{s-l}
	\end{pmatrix} & 0 \\
	0 & \begin{pmatrix}
	\id_k & 0 \\
	0 & -\id_{t-k}
	\end{pmatrix}
	\end{pmatrix},
	\]
	where $\id_r$ is the identity $r \times r$ matrix, with repect to some local frame of $T\mathcal{P}$.
	For instance, if $\mathcal{P} = \mathbb{R}^{2n}$, then
	\[
	J=
	\begin{pmatrix}
	\id_n & 0 \\
	0 & -\id_n
	\end{pmatrix},\quad
	\eta=
	\begin{pmatrix}
	\id_n & 0 \\
	0 & \id_n
	\end{pmatrix}
	\]
	is a chiral structure on $\mathcal{P}$.
\end{Ex}

Almost anti-Hermitian and almost chiral structures have similar geometries. We thus present some of their properties simultaneously. Let $(J,\eta)$ be an almost anti-Hermitian structure \emph{or} an almost chiral structure. As we have just seen, the tensor $\eta J$ is a again pseudo-Riemannian metric in both cases, which is why anti-Hermitian and chiral structures have similar geometries.
A full classification of such structures is known in terms of 36 classes \cite{naviera_classification} characterized among else by the {\bf fundamental tensor} 
\[ F(X,Y,Z)\coloneqq \eta((\lc_X J)Y,Z), \] 
$X,Y,Z \in \Gamma(T)$, where $\lc$ denotes the Levi-Civita connection of $\eta$. Two of the 36 classes are the classes $\mathcal{W}_0$ and $\mathcal{W}_3$ \cite{chiral1}, which are defined as follows:

\begin{Def}
	Let $(J,\eta)$ be an almost anti-Hermitian or an almost chiral structure.
	Then, $(J,\eta)$ is said to be:
	\begin{enumerate}
		\item 
		\emph{of class $\mathcal{W}_0$} if $F=0$;
		\item 
		\emph{of class $\mathcal{W}_3$} if
		\begin{align}\label{eq:classW3}
		\sum_{Cycl. X,Y,Z} F(X,Y,Z)=0\,
		\end{align}
		for all $X,Y,Z \in \Gamma(T)$.
	\end{enumerate}
\end{Def}

\begin{remark*}
	By definition of $F$ and non-degenracy of $\eta$, we have $F = 0$ if and only if $\lc J = 0$. In other words, $(J,\eta)$ is of class $\mathcal{W}_0$ if and only if $J$ is parallel with respect to the Levi-Civita connection of $\eta$. Furthermore, if we set $N(X,Y,Z) := \eta(N(X,Y),Z)$, then
	\[ N(X,Y,Z) = F(JX,Y,Z) - F(JY,X,Z) + F(X,JY,Z) - F(Y,JX,Z) \]
	for all $X,Y,Z \in \Gamma(T)$, so that $J$ is integrable if $F = 0$. All structures $(J,\eta)$ of class $\mathcal{W}_0$ are thus integrable. Note that if $F = 0$, then $(J,\eta)$ is also of class $\mathcal{W}_3$ so that the class $\mathcal{W}_0$ is contained in the class $\mathcal{W}_3$. We in fact prove below that $\mathcal{W}_0$ consists of the integrable structures $(J,\eta)$ of class $\mathcal{W}_3$ (see Lemma \ref{F:properties}).
\end{remark*}
\begin{remark*}
	Recall that for an almost (para-)Hermitian structure $(K,g)$ with fundamental form $\omega = g K$, the exterior derivative $\rd \omega$ can be written in terms of the Levi-Civita connection $\n^{LC}$ of $g$ as follows:
	\begin{align*}
	\rd \omega(X,Y,Z)=\sum_{Cycl.\ X,Y,Z} g((\n^{LC}_X K)Y,Z).
	\end{align*}
	The characteristic property \eqref{eq:classW3} of almost anti-Hermitian or chiral structures of class $\mathcal{W}_3$ is then analogous to the requirement $\rd \omega =0$ in (para-)Hermitian geometry for the almost (para-)Hermitian structure $(K,g)$ to be K\"ahler.
	Almost anti-Hermitian or chiral structures of class $\mathcal{W}_3$ thus correspond to almost K\"ahler or para-K\"ahler structures, respectively.
\end{remark*}

Here are some useful properties of the fundamental tensor $F$:
\begin{lemma}\label{F:properties}
	Let $(J,\eta)$ be an almost anti-Hermitian or an almost chiral structure with fundamental tensor $F$. Then:
	\begin{enumerate}
		\item 
		For all $X,Y,Z \in \Gamma(T)$,
		\begin{align}\label{eq:F_properties}
		F(X,Y,Z)=F(X,Z,Y)=\pm F(X,JY,JZ)
		\end{align}
		if $J^2 = \mp \id$.
		\item 
		If $(J,\eta)$ is of class $\mathcal{W}_3$, then
		\begin{align}\label{eq:N_W3} 
		N(X,Y,Z) := \eta(N(X,Y),Z) = 2\left( F(X,JY,Z) + F(JX,Y,Z)\right) 
		\end{align}
		for  all $X,Y,Z \in \Gamma(T)$. 
	\end{enumerate}
\end{lemma}
\begin{proof}
	The first property follows directly from the definition of $F$ and the properties of the Levi-Civita connection (also see e.g. \cite{chiral1} for the almost chiral case).
	Moreover, note that 
	\[ N(X,Y,Z) = F(JX,Y,Z) - F(JY,X,Z) + F(X,JY,Z) - F(Y,JX,Z) \]
	for all $X,Y,Z \in \Gamma(T)$. Consequently, if $(J,\eta)$ is of class $\mathcal{W}_3$, 
	\[ F(JY,X,Z) = -F(X,Z,JY) - F(Z,JY,X) = -F(X,JY,Z) - F(Z,JY,X)\]
	and
	\[ F(Y,JX,Z) = -F(JX,Z,Y) - F(Z,Y,JX) = -F(JX,Y,Z) + F(Z,JY,X),\]
	so that $N(X,Y,Z) = 2\left( F(X,JY,Z) + F(JX,Y,Z)\right)$, proving the second property.
\end{proof}

As a direct consequence of the above lemma, we have:
\begin{proposition}
	An almost anti-Hermitian or almost chiral structure $(J,\eta)$ is integrable of class $\mathcal{W}_3$ if and only if it is of class $\mathcal{W}_0$ if and only if $\lc J = 0$.
\end{proposition}
\begin{proof}
	Suppose that $(J,\eta)$ is of class $\mathcal{W}_3$. Then, given the expression \eqref{eq:N_W3} of $N$ in terms of $F$, we see that $J$ is integrable if and only if $F = 0$. Finally, by the definition of $F$, we see that $F = 0$ if and only if $\lc J = 0$.
\end{proof}

\begin{remark*}
	This proposition tells us that anti-Hermitian/chiral structures of class $\mathcal{W}_0$ correspond to K\"ahler/para-K\"ahler structures in Hermitian/para-Hermitian geometry.
\end{remark*}

We now state an analog of Proposition \ref{prop:n_(para)herm} in the anti-Hermitian/chiral setting that highlights a special set of almost anti-Hermitian/chiral structures of class $\mathcal{W}_3$:

\begin{proposition}\label{prop:DJ}
	Let $(\eta,J)$ be an anti-Hermitian or almost chiral structure and $\lc$ be the Levi-Civita connection of $\eta$. Also let $h$ be a 3-form and set 
	\[\n^h :=\lc+\frac{1}{2}g^{-1}h.\] 
	Then,
	$\nabla^h$ is a metric connection with torsion $\eta^{-1}h$ and the following are equivalent:
	\begin{enumerate}
		\item 
		$\n^h J=0$.
		\item 
		$F$ satisfies
		\begin{align}\label{eq:F_h_reln}
		F(X,Y,Z)=-\frac{1}{2}\left(h(X,JY,Z) - h(X,Y,JZ)\right)
		\end{align}
		for any $X,Y,Z \in \Gamma(T)$. 
		\item 
		$(\eta,J)$ is of class $\mathcal{W}_3$ and its Nijenhuis tensor is related to $h$ in the following way:
		\begin{align}\label{eq:N_h_reln}
		N(X,Y,Z)=2\left(h^{(2,1)+(1,2)}(JX,Y,JZ) + h^{(2,1)+(1,2)}(X,JY,JZ)\right)
		\end{align}
		for all $X,Y,Z \in \Gamma(T)$.
	\end{enumerate}
\end{proposition}
\begin{proof}
	Let us first note that, by definition of $\n^h$, 
	\[ \eta((\n^h_XJ)Y,Z) = F(X,Y,Z) -\frac{1}{2}\left(h(X,JY,Z) - h(X,Y,JZ)\right)\]
	for any $X,Y,Z \in \Gamma(T)$. Consequently, $\n^h J=0$ if and only if 
	\[F(X,Y,Z)=-\frac{1}{2}\left(h(X,JY,Z) - h(X,Y,JZ)\right),\] 
	proving that (1) and (2) are equivalent.
	
	Let us now suppose that $\n^h J=0$ so that $F$ satisfies \eqref{eq:F_h_reln}. As direct computation then gives us
	\[ F(X,Y,Z) + F(Y,Z,X) + F(Z,X,Y) = 0,\]
	implying that $(J,\eta)$ is of class $\mathcal{W}_3$. Moreover, if $J^2 = \mp \id$,
	\[ 
	\begin{aligned}
	N(X,Y,Z)&=F(JX,Y,Z) - F(JY,X,Z) + F(X,JY,Z) - F(Y,JX,Z)\\
	&=h(JX,Y,JZ)+h(X,JY,JZ)-h(JX,JY,Z)\pm h(X,Y,Z)\\
	& = 2\left(h^{(2,1)+(1,2)}(JX,Y,JZ) + h^{(2,1)+(1,2)}(X,JY,JZ)\right).
	\end{aligned}\]
	
	Conversely, suppose that $(\eta,J)$ is of class $\mathcal{W}_3$ and $N$ satisfies \eqref{eq:N_h_reln}. Since $(J,\eta)$ is of class $\mathcal{W}_3$, then
	\[ N(X,Y,Z) = 2\left( F(X,JY,Z) + F(JX,Y,Z)\right) \]
	by \eqref{eq:N_W3}.
	Hence, since $N$ satisfies \eqref{eq:N_h_reln}, we have that 
	\[ F(X,JY,Z) + F(JX,Y,Z) = h^{(2,1)+(1,2)}(JX,Y,JZ) + h^{(2,1)+(1,2)}(X,JY,JZ) \]
	for all $X,Y,Z \in \Gamma(T)$, implying that $F$ satisfies equation \eqref{eq:F_h_reln}.
	This proves that (2) is equivalent to (3).
\end{proof}

\begin{remark*}
	We see from the proposition that the $(3,0)+(0,3)$-component of $h$ has no bearing on whether or not $\n^h J = 0$ since only the $(2,1)+(1,2)$-component of $h$ comes into play in \eqref{eq:N_h_reln}. In fact, if $\n^hJ = 0$, referring to \eqref{eq:N_h_reln}, $J$ is integrable if and only if
	\[ h^{(2,1)+(1,2)}(JX,Y,JZ) = - h^{(2,1)+(1,2)}(X,JY,JZ)\]
	for all $X,Y,Z \in \Gamma(T)$, which is equivalent to $h^{(2,1)+(1,2)} = 0$. In addition, if $h^{(2,1)+(1,2)} = 0$, equation \eqref{eq:F_h_reln} tells us that $F = 0$, which is equivalent to $\lc J = 0$. Putting it all together, we obtain:
\end{remark*}

\begin{corollary}\label{cor:W_0}
	Let $(\eta,J)$ be an anti-Hermitian or almost chiral structure and $\lc$ be the Levi-Civita connection of $\eta$. Also let $h$ be a 3-form and set 
	\[\n^h :=\lc+\frac{1}{2}\eta^{-1}h.\] 
	The following are equivalent:
	\begin{enumerate}
		\item 
		$\n^h J=0$ and $h$ is of type (3,0)+(0,3).
		
		\item 
		$(J,\eta)$ is of class $\mathcal{W}_0$.
		
		\item 
		$\lc J = 0$.
	\end{enumerate}
\end{corollary}

\section{Para-hyperHermitian and Born structures}
In this section, we consider pairs $J,K$ of almost product structures on a manifold $\mathcal{P}$ that anti-commute: 
\[ J^2 = K^2 = \id\] 
and 
\[ \{J,K\}:=JK+KJ=0.\] 
Note that since $J$ and $K$ anti-commute, they in fact both have to be almost complex structures. Indeed, by anti-commutativity, $J$ maps the $\pm 1$-eigenbundle of $K$ isomorphically onto the $\mp 1$-eigenbundle of $K$, implying that the eigenbundles of $K$ have the same rank. A similar argument applies to $J$.

Furthermore, the product 
\[ I\coloneqq JK=-KJ\]
defines an almost complex structure since 
\[ I^2=-\id,\] 
and any pair from the triple $I,J,K$ anticommutes: 
\[ \{ I,J\} = \{ I,K \} = \{ J,K \} = 0.\]
Such a triple of endomorphisms is called a \textbf{para-hypercomplex} or 
\textbf{para-quaternionic structure} on the manifold $\mathcal{P}$ \cite{Cortes:2010ykx,hitchin1990hypersymplectic}. 

In the following, we explore the possibility of adding a metric $\eta$ that is compatible with the three endomorphisms $I,J,K$ of a para-hypercomplex structure. 
We assume that each endomorphism is either orthogonal or anti-orthogonal with respect to $\eta$. Since $K = JI$, we have three cases:
\begin{itemize}
	\item If $I$ is orthogonal with respect to $\eta$, then $J$ and $K$ are either both orthogonal or both anti-orthogonal, giving us the two options:
	\begin{enumerate}
		\item $\eta(I\cdot,I\cdot)=\eta(J\cdot,J\cdot)=\eta(K\cdot,K\cdot)=\eta(\cdot,\cdot)$.
		\item $\eta(I\cdot,I\cdot)=-\eta(J\cdot,J\cdot)=-\eta(K\cdot,K\cdot)=\eta(\cdot,\cdot)$.
	\end{enumerate}
	\item If $I$ is anti-orthogonal with respect to $\eta$, then $J$ and $K$ have different orthogonalities. Without loss of generality, this mean that $J$ is orthogonal and $K$ is anti-orthogonal:
	\begin{enumerate}
		\item[3.] $\eta(I\cdot,I\cdot)=-\eta(J\cdot,J\cdot)=\eta(K\cdot,K\cdot)=-\eta(\cdot,\cdot)$.
	\end{enumerate}
\end{itemize} 
Option 2. gives rise to {\bf para-hyperHermitian geometry} whereas options 1. and 3. give rise to {\bf Born geometry}, which we describe in the next two sections.

\subsection{Para-hyperHermitian geometry}\label{app:parahyperhermitian}
Para-hyperHermitian geometry is the para-complex analogue to hyperHermitian geometry. It has been studied for example in \cite{hitchin1990hypersymplectic,kamada1999neutral} and other works, including in physics \cite{Ooguri:1990ww,Ooguri:1991fp}.

\begin{Def}
	An \textbf{(almost) para-hyperHermitian structure} on a manifold $\mathcal{P}$ is a quadruple $(\eta,I,J,K)$ consisting of an (almost) para-hypercomplex triple $I,J,K$:
	\begin{align*}
	-I^2=J^2=K^2=\id,\quad \{I,J\}=\{J,K\}=\{K,I\}=0,\quad I=JK,
	\end{align*}
	and a pseudo-Riemannian metric $\eta$ such that
	\begin{align*}
	\eta(I\cdot,I\cdot)=-\eta(J\cdot,J\cdot)=-\eta(K\cdot,K\cdot)=\eta(\cdot,\cdot).
	\end{align*}
\end{Def}

The definition of an (almost) para-hyperHermitian structure $(\eta,I,J,K)$ implies that $(J,\eta)$ and $(K,\eta)$ are both (almost) para-Hermitian structures with respect to the same metric $\eta$. Consequently, $\eta$ is necessarily of signature $(n,n)$ and $(I,\eta)$ is an (almost) pseudo-Hermitian structure. Finally, the contractions of $\eta$ with each of the three structures $I,J,K$ are almost symplectic two-form, called the \textbf{fundamental forms} of the (almost) para-Hermitian structure, which we denote:
\begin{align*}
\omega_{I} := \eta I,\quad \omega_J := \eta J,\quad \omega_K := \eta K.
\end{align*}
When $\rd \omega_I = \rd \omega_J = \rd \omega_K = 0$, we call the structure $(\eta,I,J,K)$ {\bf (almost) para-hyperK\"ahler} or {\bf hypersymplectic}. Para-hyperK\"ahler geometry has been extensively studied by both mathematicians and physicists (see \cite{Cortes:2010ykx, hitchin1990hypersymplectic,kamada1999neutral} and the references therein). 

\begin{Ex}[Linear structures]\label{linear_pHK} 
	We first consider para-hyperHermitian structures on vector spaces, which give us the canonical local forms of such structures on any manifold.
	Let $V$ be a $2n$-dimensional real vector space and let $(\eta,I,J,K)$ be a para-hyperHermitian structure on $V$. Choose the form of $(\eta,K)$ to be the standard form of a para-Hermitian structure diagonalizing $K$:
	\begin{align*}
	\eta=
	\begin{pmatrix}
	0 & \id \\
	\id & 0
	\end{pmatrix},\quad
	K=
	\begin{pmatrix}
	\id & 0 \\
	0 & -\id
	\end{pmatrix},
	\end{align*}
	where the blocks are $n\times n$ matrices. Because $I$ anti-commutes with $K$, squares to $-\id$ and is orthogonal with respect to $\eta$, it is of the form
	\begin{align*}
	I=
	\begin{pmatrix}
	0	& -\Omega^{-1} \\
	\Omega & 0
	\end{pmatrix}
	\end{align*}
	with $\Omega$ a skew-symmetric, non-degenerate $n\times n$ matrix. This implies, in particular, that $n=2k$ for some $k\in \mathbb{N}$, so that para-hyperHermitian structures only exist in dimension $4k$. Moreover, $J$ takes the form
	\begin{align*}
	J=
	\begin{pmatrix}
	0 & \Omega^{-1} \\
	\Omega & 0
	\end{pmatrix}.
	\end{align*}
	We can also start from a frame where $\eta$ is diagonal, in which case
	\begin{align*}
	\eta=
	\begin{pmatrix}
	\id & 0 \\
	0 & -\id
	\end{pmatrix},\quad
	K=
	\begin{pmatrix}
	0 & \id \\
	\id & 0
	\end{pmatrix},\quad
	I=
	\begin{pmatrix}
	\hat{I} & 0 \\
	0 & -\hat{I}
	\end{pmatrix},\quad
	J=
	\begin{pmatrix}
	0	 & \hat{I} \\
	-\hat{I} & 0
	\end{pmatrix},
	\end{align*}
	where $\hat{I}$ is a skew-symmetric $n\times n$ matrix such that $\hat{I}^2=-\id$.
\end{Ex}

The above discussion then tells us the following:
\begin{proposition}
	Let $(\mathcal{P}, \eta, I, J, K)$ be a para-hyperHermitian manifold. Then, $\mathcal{P}$ has real dimension $4k$ for some $k \in \mathbb{N}$, in which case $\eta$ has signature $(2k,2k)$. Moreover, there exist local frames of $T\mathcal{P}$ with respect to which
	\begin{align*}
	\eta=
	\begin{pmatrix}
	0 & \id \\
	\id & 0
	\end{pmatrix},\quad
	I=
	\begin{pmatrix}
	0	& -\Omega^{-1} \\
	\Omega & 0
	\end{pmatrix}, \quad
	J=
	\begin{pmatrix}
	0 & \Omega^{-1} \\
	\Omega & 0
	\end{pmatrix}, \quad
	K=
	\begin{pmatrix}
	\id & 0 \\
	0 & -\id
	\end{pmatrix}
	\end{align*}
	where $\Omega$ is a non-degenerate skew-symmetric $n\times n$ matrix, or
	\begin{align*}
	\eta=
	\begin{pmatrix}
	\id & 0 \\
	0 & -\id
	\end{pmatrix},\quad
	I=
	\begin{pmatrix}
	\hat{I} & 0 \\
	0 & -\hat{I}
	\end{pmatrix},\quad
	J=
	\begin{pmatrix}
	0	 & \hat{I} \\
	-\hat{I} & 0
	\end{pmatrix}, \quad
	K=
	\begin{pmatrix}
	0 & \id \\
	\id & 0
	\end{pmatrix}
	\end{align*}
	where $\hat{I}$ is a skew-symmetric $n\times n$ matrix such that $\hat{I}^2=-\id$.
\end{proposition}

\begin{Ex}
	If $\mathcal{P} = \mathbb{R}^{4k}$ and $(\eta,I,J,K)$ is one of the linear para-hyperHermitian structures described in example \ref{linear_pHK}, then it is para-hyperK\"ahler. Moreover, this linear para-hyperK\"ahler structure also descends to a para-hyperK\"ahler structure on the $4k$-torus $T^{4k} = \mathbb{R}^{4k}/\Lambda$, where $\Lambda$ is a $4k$-dimensional lattice in $\mathbb{R}^{4k}$. It is also known that para-hyperK\"ahler structures exist on Kodaira surfaces \cite{kamada1999neutral}.
\end{Ex}

\begin{Ex}
	We now give an example of a para-hyperHermitian structure that is {\em not} para-hyperK\"ahler. Let $\mathcal{P} = \mathbb{R}^4$ with coordinates $(x_1,x_2,x_3,x_4)$. Consider the coordinate frames $\{ \partial/\partial x_i \}_{i=1}^4$ and $\{ dx_i \}_{i=1}^4$ of $T\mathbb{R}^4$ and $T^*\mathbb{R}^4$, respectively. Then, with respect to these frames,
	\[ \eta = dx_1 \otimes dx_1 + dx_2 \otimes dx_2 - dx_3 \otimes dx_3 - dx_4 \otimes dx_4,\]
	\[ J(\partial/\partial x_1) = -\partial/\partial x_4, \quad
	J(\partial/\partial x_2) = \partial/\partial x_3, \]
	and
	\[ K(\partial/\partial x_1) = \partial/\partial y_1, \quad
	K(\partial/\partial x_2) = \partial/\partial y_2, \]
	determine a para-hyperK\"ahler structure on $\mathbb{R}^4$ with $I = JK$.
	
	Let $X = \mathbb{R}^4\backslash \{ 0 \}/\sim$, where $x \sim 2x$ for all 
	$x = (x_1,x_2,x_3,x_4) \in \mathbb{R}^4\backslash \{ 0 \}$. 
	Note that $X$ is a compact 4-manifold known as a {\em Hopf surface} that diffeomorphic to $S^3 \times S^1$, and can therefore not admit K\"ahler structures (because $b_1(X)$ is odd).
	The para-hypercomplex structure $(I,J,K)$ we defined on $\mathbb{R}^4$ nonetheless descends to $X$.
	Let
	\[ |x|^2 := x_1^2 + x_2^2 + x_3^2 + x_4^2\]
	for all 
	$x = (x_1,x_2,x_3,x_4) \in \mathbb{R}^4\backslash \{ 0 \}$. Then, if we set $\tilde{\eta} = \eta/| \cdot|^2$, the quadruple $(\tilde{\eta},I,J,K)$ is a para-hyperHermitian structure on $X$ whose fundamental forms $\tilde{\omega}_I := \tilde{\eta}I, \tilde{\omega}_J := \tilde{\eta}J, \tilde{\omega}_K := \tilde{\eta}_K$
	are such that 
	\[ \rd \tilde{\omega}_I \neq 0 \]
	and  
	\[ \rd^p_J \tilde{\omega}_J = \rd^p_K \tilde{\omega}_K \neq 0,\]
	implying that $(\tilde{\eta},I,J,K)$ is not para-hyperK\"ahler. 
	
	For other examples of para-hyperHermitian structures that are not para-hyperK\"ahler, we refer the reader to \cite{Davidov:2009uc,Davidov2012Surfaces} and the references therein.
\end{Ex}

\subsection{Born Geometry}\label{sec:born}
Born geometry was introduced in physics \cite{freidel2014born,freidel2015metastring} and its properties later discussed in \cite{freidel2017generalised,Freidel:2018tkj,Szabo-paraherm}. Here we present Born geometry from a slightly different point of view, keeping the analogy with para-hyperHermitian geometry explicit:

\begin{Def}\label{def:born_new}
	An \textbf{(almost) Born structure} on a manifold $\mathcal{P}$ is a quadruple $(\eta,I,J,K)$ consisting of an (almost) para-hypercomplex triple $I,J,K$:
	\begin{align*}
	-I^2=J^2=K^2=\id,\quad \{I,J\}=\{J,K\}=\{K,I\}=0,\quad I=JK,
	\end{align*}
	and a pseudo-Riemannian metric $\eta$ such that
	\begin{align*}
	\eta(I\cdot,I\cdot)=-\eta(J\cdot,J\cdot)=\eta(K\cdot,K\cdot)=-\eta(\cdot,\cdot).
	\end{align*}
\end{Def}
This again implies that $\eta$ is of signature $(n,n)$. However, $(\eta,I)$ is now an anti-Hermitian pair and the symmetry between $J$ and $K$ is seemingly broken: $(\eta,K)$ is (almost) para-Hermitian while $(\eta,J)$ is (almost) chiral. This means that the contractions between the individual pairs are
\begin{align*}
\eta I=g_I,\quad \eta J=g_J,\quad \eta K=\omega_K,
\end{align*}
where $g_I$ and $g_J$ are metric structures and $\omega_K$ is an almost-symplectic form. We now compare the definition \ref{def:born_new} with the one found in literature\footnote{The definition is slightly tweaked so that it fits the language in the present discussion.}

\begin{Def}[\cite{Freidel:2018tkj}, Def. 10]\label{def:born_old}
	Let $(\PS,\eta,\omega)$ be an almost para-Hermitian manifold and let $\HH$ be a Riemannian metric satisfying
	\begin{align}
	\eta^{-1}\HH=\HH^{-1}\eta,\quad \omega^{-1}\HH=-\HH^{-1}\omega .
	\end{align}
	Then we call the triple $(\eta,\omega,\HH)$ an (almost) Born structure on $\PS$ where $\PS$ is called an (almost) Born manifold and $(\PS,\eta,\omega,\HH)$ a Born geometry.
\end{Def}

\begin{corollary}
	The Definition \ref{def:born_old} is equivalent to the Definition \ref{def:born_new} with $\eta J=g_J$ a Riemannian metric.
\end{corollary}
\begin{proof}
	Assume the properties in Definition \ref{def:born_new} with $g_J$ Riemannian. Then upon identifying $\omega_K\coloneqq \eta K$ and $g_J=\eta J$ with $\omega$ and $\HH$ in Definition \ref{def:born_new}, we get $(\eta,\omega)$ is indeed almost para-Hermitian and $\HH$ is Riemannian by assumption. The condition $\eta^{-1}\HH=\HH^{-1}\eta$ follows directly from the fact that the pair $(\eta,J)$ is chiral. Now, because $(\eta,K)$ is para-Hermitian, $K=\omega^{-1}\eta=\eta^{-1}\omega$ and because $(\eta,J)$ is chiral, $J=g_J^{-1}\eta=\eta^{-1}g_J$. The relation $JK=-KJ$ then implies $\omega^{-1}\HH=-\HH^{-1}\omega$.
	
	The converse statement Definition \ref{def:born_old} $\Rightarrow$ Definition \ref{def:born_new} follows similarly.
\end{proof}

We see that apart from the restriction on the signature of $g_J/\HH$, the two definitions are equivalent. From now on, we will be using the customary notations $\HH$ for $g_J$ and $\omega$ for $\omega_K$.

We will now show that the metric $g_I\coloneqq \eta'$ is in fact another signature $(n,n)$ metric such that $(\eta',J)$ is para-Hermitian and $(\eta',K)$ is chiral, exchanging the roles of $J$ and $K$.

\begin{proposition}\label{prop:born=bichiral=bihermitian}
	Let $(\eta,I,J,K)$ be an (almost) Born geometry, such that $(\eta,I)$, $(\eta,J)$ and $(\eta,K)$ are (almost) anti-Hermitian, chiral and para-Hermitian, respectively. Then $(\eta',I)$,  $(\eta',J)$ and $(\eta',K)$, where $\eta'=\eta I$ are (almost) anti-Hermitian, para-Hermitian and chiral, respectively. In particular, $(\eta,J)$ and $(\eta',K)$ is a pair of chiral structures sharing the same metric $\HH=\eta J=\eta' K$ and $(\eta,K)$ and $(\eta',J)$ is a pair of para-Hermitian structures sharing the same fundamental form $\omega=\eta K=\eta' J$.
\end{proposition}
\begin{proof}
	Let us check the required orthogonality properties. Using $JK=I$ and the orthogonality properties with $\eta$,
	\begin{align*}
	\eta'&(IX,IY)=-\eta(X,IY)=-\eta(IX,Y)=-\eta'(X,Y),\\
	\eta'&(JX,JY)=-\eta(KX,JY)=-\eta(IX,Y)=-\eta'(X,Y),\\
	\eta'&(KX,KY)=\eta(JX,KY)=\eta(IX,Y)=\eta'(X,Y).
	\end{align*}
	The equalities $\HH=\eta J=\eta' K$ and $\omega=\eta K=\eta' J$ again follow from  $JK=I$.
\end{proof}

\begin{remark}
	Proposition \ref{prop:born=bichiral=bihermitian} also shows that apart from para-hyper-Hermitian geometry, Born Geometry is the only other option for the choice of orthogonality for the para-hypercomplex structure $(I,J,K)$ with respect to some metric. Denoting orthogonal by $+$ and anti-orthogonal by $-$ and keeping the notation consistent with the above discussion, the options for $(I,J,K)$ are:
	\begin{itemize}
		\item $(+,-,-)$: para-Hyper-Hermitian.
		\item $(+,+,+)$: Born with respect to $\HH$.
		\item $(-,+,-)$: Born with respect to $\eta$.
		\item $(-,-,+)$: Born with respect to $\eta'$.
	\end{itemize}	
\end{remark}

\begin{Ex}[Linear structures]\label{linear_born} 
	We now discuss Born structures on a vector space. This can also be understood as the canonical local form of the geometry on manifolds.
	Let $V$ be a $2n$-dimensional real vector space and $(\eta, I, J, K)$ be a para-hyperHermitian structure on $V$. Pick a basis of $V$ with respect to which $(\eta,K)$ has the standard form of a para-Hermitian structure diagonalizing $K$ (see example \ref{local_paraH}):
	\begin{align*}
	\eta=
	\begin{pmatrix}
	0 & \id \\
	\id & 0
	\end{pmatrix},\quad
	K=
	\begin{pmatrix}
	\id & 0 \\
	0 & -\id
	\end{pmatrix}.
	\end{align*}
	The anti-Hermitian and chiral structures then have the form
	\begin{align*}
	I=
	\begin{pmatrix}
	0 & -g^{-1} \\
	g & 0
	\end{pmatrix},\quad
	J=
	\begin{pmatrix}
	0 & g^{-1} \\
	g & 0
	\end{pmatrix}
	\end{align*}
	with respect to this basis,
	for some non-degenerate symmetric $n\times n$ matrix $g$. The tensors $\eta'=\eta I$, $\HH=\eta J$ and $\omega=\eta K$ then take the form
	\begin{align*}
	\eta'=
	\begin{pmatrix}
	g & 0 \\
	0 & -g^{-1}
	\end{pmatrix},\quad
	\HH=
	\begin{pmatrix}
	g & 0 \\
	0 & g^{-1}
	\end{pmatrix},\quad
	\omega=
	\begin{pmatrix}
	0 & -\id \\
	\id & 0
	\end{pmatrix}.
	\end{align*}
\end{Ex}

The above discussion then tells us the following:
\begin{proposition}
	Let $(\mathcal{P}, \eta, I, J, K)$ be an almost Born manifold. Then, there exist local frames of $T\mathcal{P}$ with respect to which
	\begin{align*}
	\eta=
	\begin{pmatrix}
	0 & \id \\
	\id & 0
	\end{pmatrix},\quad
	I=
	\begin{pmatrix}
	0 & -g^{-1} \\
	g & 0
	\end{pmatrix}, \quad
	J=
	\begin{pmatrix}
	0 & g^{-1} \\
	g & 0
	\end{pmatrix}, \quad
	K=
	\begin{pmatrix}
	\id	& 0 \\
	0 & -\id
	\end{pmatrix},
	\end{align*}
	where $g$ is a non-degenerate symmetric $n\times n$ matrix. 
\end{proposition}

\begin{Ex}\label{ex:BornHK}
	If $\mathcal{P} = \mathbb{R}^{4k}$ and $(\eta,I,J,K)$ is one of the linear Born structures described in example \ref{linear_born}, then $(\HH,I)$ is pseudo-K\"ahler and $(\HH,J)$, $(\HH,K)$ are both chiral of class $\mathcal{W}_0$. In other words, $(\eta,I,J,K)$ is a Born structure of \textbf{hyperK\"ahler-type}. Furthermore, this linear structure descends to a Born structure of hyperK\"ahler-type on the $4k$-torus $T^{4k} = \mathbb{R}^{4k}/\Lambda$, where $\Lambda$ is a $4k$-dimensional lattice in $\mathbb{R}^{4k}$. 
\end{Ex}

\section{Anti-hyperHermitian structures}\label{sec:anti-HH}
Finally, we consider pairs $J,K$ of almost complex structures on a manifold $\mathcal{P}$ that anti-commute: 
\[ J^2 = K^2 = - \id\] 
and 
\[ \{J,K\}:=JK+KJ=0.\] 
The product 
\[ I\coloneqq JK=-KJ\]
then  defines an almost complex structure since 
\[ I^2=-\id,\] 
and any pair from the triple $I,J,K$ anticommutes: 
\[ \{ I,J\} = \{ I,K \} = \{ J,K \} = 0.\]
In other words, the triple $I, J, K$ is an \textbf{almost hypercomplex structure} on $\mathcal{P}$. This implies, in particular, that $\mathcal{P}$ must have dimension $4k$ for some $k\in \mathbb{N}$ (because hypercomplex structures only exist in dimension $4k$).

Let $\eta$ be a metric that is compatible with the three endomorphisms $I,J,K$ of an almost hypercomplex structure $(\eta,I,J,K)$. 
We assume that each endomorphism is either orthogonal or anti-orthogonal with respect to $\eta$. Since $I = JK$, we obtain three cases:
\begin{itemize}
	\item If $I$ is orthogonal with respect to $\eta$, then $J$ and $K$ are either both orthogonal or both anti-orthogonal, giving us the two options:
	\begin{enumerate}
		\item $\eta(I\cdot,I\cdot)=\eta(J\cdot,J\cdot)=\eta(K\cdot,K\cdot)=\eta(\cdot,\cdot)$.
		\item $\eta(I\cdot,I\cdot)=-\eta(J\cdot,J\cdot)=-\eta(K\cdot,K\cdot)=\eta(\cdot,\cdot)$.
	\end{enumerate}
	\item If $I$ is anti-orthogonal with respect to $\eta$, then $J$ and $K$ have different orthogonalities. Without loss of generality, this mean that $J$ is orthogonal and $K$ is anti-orthogonal:
	\begin{enumerate}
		\item[3.] $\eta(I\cdot,I\cdot)=-\eta(J\cdot,J\cdot)=\eta(K\cdot,K\cdot)=-\eta(\cdot,\cdot)$.
	\end{enumerate}
\end{itemize} 
Option 1. gives rise to {\bf hyperHermitian geometry} whereas options 2. and 3. give rise to {\bf anti-hyperHermitian geometry}, which we describe in the next two sections.

\begin{Def}
	An \textbf{(almost) anti-hyperHermitian structure} on a manifold $\mathcal{P}$ is a quadruple $(\eta,I,J,K)$ consisting of an (almost) hypercomplex triple $I,J,K$:
	\begin{align*}
	I^2=J^2=K^2=-\id,\quad \{I,J\}=\{J,K\}=\{K,I\}=0,\quad I=JK,
	\end{align*}
	and a pseudo-Riemannian metric $\eta$ such that
	\begin{align*}
	\eta(I\cdot,I\cdot)=-\eta(J\cdot,J\cdot)=-\eta(K\cdot,K\cdot)=\eta(\cdot,\cdot).
	\end{align*}
\end{Def}

The definition of an (almost) anti-hyperHermitian structure $(\eta,I,J,K)$ implies that $(J,\eta)$ and $(K,\eta)$ are both (almost) anti-Hermitian structures with respect to the same metric $\eta$. Consequently, $\eta$ is necessarily of signature $(n,n)$ and $(I,\eta)$ is an (almost) pseudo-Hermitian structure. Moreover, if $J$ and $K$ are both integrable, then so is $I$, and if $J$ and $K$ are both parallel with respect the Levi-Civita connection of $\eta$, then so is $I$. Consequently, if $(J,\eta)$ and $(K,\eta)$ are both of class $\mathcal{W}_0$, then $(I,\eta)$ is pseudo-K\"ahler. An anti-hyperHermitian structure $(\eta, I , J, K)$ is thus said to be \textbf{anti-hyperK\"ahler} if $J$ and $K$ are both $(J,\eta)$ and $(K,\eta)$ are both of class $\mathcal{W}_0$.

\begin{Ex}[Linear structures]\label{linear_antiHH} 
	We first consider anti-hyperHermitian structures on vector spaces, which give us the canonical local forms of such structures on any manifold.
	Let $V$ be a $4n$-dimensional real vector space and let $(\eta,I,J,K)$ be a para-hyperHermitian structure on $V$. 
	Choose a basis of $V$ with respect to which $(\eta,J)$ has the standard form of an anti-Hermitian structure (see \ref{local_AntiH}):
	\begin{align*}
		J=
	\begin{pmatrix}
	0 & -\id_n \\
	\id_n & 0
	\end{pmatrix}
	,\quad
	\eta=
	\begin{pmatrix}
	\id_n & A \\
	A & -\id_n
	\end{pmatrix},
	\end{align*}
	where the blocks are $n\times n$ matrices. Because $K$ anti-commutes with $J$, squares to $\id$ and is anti-orthogonal with respect to $\eta$, it is of the form
	\begin{align*}
	K=
	\begin{pmatrix}
	0	& B \\
	B & 0
	\end{pmatrix},
	\end{align*}
	with $B$ a skew-symmetric $n\times n$ matrix such that $B^2=-\id$. Moreover, $I = JK$ takes the form
	\begin{align*}
	I=
	\begin{pmatrix}
	-B & 0 \\
	0 & B
	\end{pmatrix}.
	\end{align*}
\end{Ex}

The above discussion then tells us the following:
\begin{proposition}
	Let $(\mathcal{P}, \eta, I, J, K)$ be an almost anti-hyperHermitian manifold. Then, $\mathcal{P}$ has real dimension $4k$, $k \in \mathbb{N}$, and $\eta$ has signature $(2k,2k)$. In addition, there exist local frames of $T\mathcal{P}$ with respect to which
	\begin{align*}
	\eta= 
	\begin{pmatrix}
	\id & A \\
	A & -\id
	\end{pmatrix},\quad
	I=
	\begin{pmatrix}
	-B & 0 \\
	0 & B
	\end{pmatrix}, \quad
	J=
	\begin{pmatrix}
	0 & -\id \\
	\id & 0
	\end{pmatrix}, \quad
	K=
	\begin{pmatrix}
	0	& B \\
	B & 0
	\end{pmatrix},
	\end{align*}
	where $B$ is a skew-symmetric $n\times n$ matrix such that $B^2=-\id$. 
\end{proposition}

\begin{Ex}\label{ex:antiHK}
	If $\mathcal{P} = \mathbb{R}^{4k}$ and $(\eta,I,J,K)$ is one of the linear anti-hyperHermitian structures described in example \ref{linear_antiHH}, then it is anti-hyperK\"ahler and descends to an anti-hyperK\"ahler structure on the $4k$-torus $T^{4k} = \mathbb{R}^{4k}/\Lambda$, where $\Lambda$ is a $4k$-dimensional lattice in $\mathbb{R}^{4k}$. 
\end{Ex}

\bibliographystyle{JHEP}
\bibliography{mybib}

\providecommand{\href}[2]{#2}\begingroup\raggedright\begin{thebibliography}{10}

\bibitem{Courant-Weinstein}
T.~J. Courant and A.~Weinstein, \emph{Beyond poisson structures, seminare
  sud-rhodanien de}, {\emph{Seminare sud-rhodanien de geometrie VIII. Travaux
  en Cours 27, Hermann, Paris} (1988) }.

\bibitem{Dorfman}
I.~Dorfman, \emph{Dirac structures of integrable evolution equations},
  {\emph{Phys. Lett. A} {\bfseries 125} (1987) 240--246}.

\bibitem{Gualtieri:2003dx}
M.~Gualtieri, \emph{{Generalized complex geometry}}, Ph.D. thesis, Oxford U.,
  2003.
\newblock \href{https://arxiv.org/abs/math/0401221}{{\ttfamily math/0401221}}.

\bibitem{Gates-hull-rocek-biherm}
S.~J. Gates, Jr., C.~M. Hull and M.~Rocek, \emph{{Twisted Multiplets and New
  Supersymmetric Nonlinear Sigma Models}},
  \href{https://doi.org/10.1016/0550-3213(84)90592-3}{\emph{Nucl. Phys.}
  {\bfseries B248} (1984) 157--186}.

\bibitem{Apostolov1999}
V.~Apostolov, P.~Gauduchon and G.~Grantcharov, \emph{Bi-hermitian structures on
  complex surfaces}, {\emph{Proc. London Math. Soc. (3)} {\bfseries 79} (1999)
  414--428}.

\bibitem{Hitchin2006}
N.~Hitchin, \emph{Instantons, poisson structures and generalized k\"ahler
  geometry}, {\emph{Comm. Math. Phys.} {\bfseries 265} (2006) 131--164}.

\bibitem{vaisman2015generalized}
I.~Vaisman, \emph{Generalized para-k{\"a}hler manifolds}, {\emph{Differential
  Geometry and its Applications} {\bfseries 42} (2015) 84--103}.

\bibitem{Gualtieri:2007bq}
M.~Gualtieri, \emph{{Branes on Poisson varieties}},
\newblock 2007.
\newblock \href{https://arxiv.org/abs/0710.2719}{{\ttfamily 0710.2719}}.
\newblock
  \href{https://doi.org/10.1093/acprof:oso/9780199534920.003.0018}{DOI}.

\bibitem{naviera_classification}
A.~M. Naviera, \emph{A classification of riemannian almost product manifolds},
  {\emph{Rendiconti di Matematica e delle sue Applicazioni} {\bfseries 3}
  (1983) 577--592}.

\bibitem{Hull_twisted_susy}
M.~Abou-Zeid and C.~M. Hull, \emph{{The Geometry of sigma models with twisted
  supersymmetry}},
  \href{https://doi.org/10.1016/S0550-3213(99)00528-3}{\emph{Nucl. Phys.}
  {\bfseries B561} (1999) 293--315},
  [\href{https://arxiv.org/abs/hep-th/9907046}{{\ttfamily hep-th/9907046}}].

\bibitem{SUSY_non-integrable}
G.~W. Delius, M.~Rocek, A.~Sevrin and P.~van Nieuwenhuizen,
  \emph{{Supersymmetric $\sigma$ Models With Nonvanishing Nijenhuis Tensor and
  Their Operator Product Expansion}},
  \href{https://doi.org/10.1016/0550-3213(89)90478-1}{\emph{Nucl. Phys.}
  {\bfseries B324} (1989) 523--531}.

\bibitem{toappear}
B.~W. Williams and D.~Svoboda, \emph{{Topological Sigma Models for Generalized
  Para-Complex Geometry}}, {\emph{in preparation} }.

\bibitem{Svoboda:2018rci}
D.~Svoboda, \emph{{Algebroid Structures on Para-Hermitian Manifolds}},
  \href{https://doi.org/10.1063/1.5040263}{\emph{J. Math. Phys.} {\bfseries 59}
  (2018) 122302}, [\href{https://arxiv.org/abs/1802.08180}{{\ttfamily
  1802.08180}}].

\bibitem{Szabo-paraherm}
V.~E. Marotta and R.~J. Szabo, \emph{{Para-Hermitian Geometry, Dualities and
  Generalized Flux Backgrounds}},
  \href{https://arxiv.org/abs/1810.03953}{{\ttfamily 1810.03953}}.

\bibitem{Hassler:2019wvn}
F.~Hassler, D.~Lüst and F.~J. Rudolph, \emph{{Para-Hermitian Geometries for
  Poisson-Lie Symmetric $\sigma$-models}},
  \href{https://arxiv.org/abs/1905.03791}{{\ttfamily 1905.03791}}.

\bibitem{Stojevic:2009ub}
V.~Stojevic, \emph{{Two-Dimensional Supersymmetric Sigma Models on
  Almost-Product Manifolds and Non-Geometry}},
  \href{https://doi.org/10.1088/0264-9381/27/23/235005}{\emph{Class. Quant.
  Grav.} {\bfseries 27} (2010) 235005},
  [\href{https://arxiv.org/abs/0906.2028}{{\ttfamily 0906.2028}}].

\bibitem{Liu:1995lsa}
Z.-J. Liu, A.~Weinstein and P.~Xu, \emph{{Manin Triples for Lie Bialgebroids}},
  {\emph{J. Diff. Geom.} {\bfseries 45} (1997) 547--574},
  [\href{https://arxiv.org/abs/dg-ga/9508013}{{\ttfamily dg-ga/9508013}}].

\bibitem{Severa:2017oew}
P.~Ševera, \emph{{Letters to Alan Weinstein about Courant algebroids}},
  \href{https://arxiv.org/abs/1707.00265}{{\ttfamily 1707.00265}}.

\bibitem{courant1990dirac}
T.~J. Courant, \emph{Dirac manifolds}, {\emph{Transactions of the American
  Mathematical Society} {\bfseries 319} (1990) 631--661}.

\bibitem{wade2004dirac}
A.~Wade, \emph{Dirac structures and paracomplex manifolds}, {\emph{Comptes
  Rendus Mathematique} {\bfseries 338} (2004) 889--894}.

\bibitem{Zabzine:2006uz}
M.~Zabzine, \emph{{Lectures on Generalized Complex Geometry and
  Supersymmetry}}, {\emph{Archivum Math.} {\bfseries 42} (2006) 119--146},
  [\href{https://arxiv.org/abs/hep-th/0605148}{{\ttfamily hep-th/0605148}}].

\bibitem{Crainic:2004ic}
M.~Crainic, \emph{{Generalized complex structures and Lie brackets}},
  \href{https://arxiv.org/abs/math/0412097}{{\ttfamily math/0412097}}.

\bibitem{Courant_1994}
T.~J. Courant, \emph{Tangent lie algebroids},
  \href{https://doi.org/10.1088/0305-4470/27/13/026}{\emph{Journal of Physics
  A: Mathematical and General} {\bfseries 27} (jul, 1994) 4527--4536}.

\bibitem{vaisman2012geometry}
I.~Vaisman, \emph{{On the geometry of double field theory}},
  \href{https://doi.org/10.1063/1.3694739}{\emph{J. Math. Phys.} {\bfseries 53}
  (2012) 033509}, [\href{https://arxiv.org/abs/1203.0836}{{\ttfamily
  1203.0836}}].

\bibitem{vaisman2013towards}
I.~Vaisman, \emph{{Towards a double field theory on para-Hermitian manifolds}},
  \href{https://doi.org/10.1063/1.4848777}{\emph{J. Math. Phys.} {\bfseries 54}
  (2013) 123507}, [\href{https://arxiv.org/abs/1209.0152}{{\ttfamily
  1209.0152}}].

\bibitem{Chatzistavrakidis:2018ztm}
A.~Chatzistavrakidis, L.~Jonke, F.~S. Khoo and R.~J. Szabo, \emph{{Double Field
  Theory and Membrane Sigma-Models}},
  \href{https://doi.org/10.1007/JHEP07(2018)015}{\emph{JHEP} {\bfseries 07}
  (2018) 015}, [\href{https://arxiv.org/abs/1802.07003}{{\ttfamily
  1802.07003}}].

\bibitem{Gualtieri:2010fd}
M.~Gualtieri, \emph{{Generalized Kahler geometry}},
  \href{https://arxiv.org/abs/1007.3485}{{\ttfamily 1007.3485}}.

\bibitem{Hitchin:2005cv}
N.~Hitchin, \emph{{Instantons, Poisson structures and generalized Kahler
  geometry}}, \href{https://doi.org/10.1007/s00220-006-1530-y}{\emph{Commun.
  Math. Phys.} {\bfseries 265} (2006) 131--164},
  [\href{https://arxiv.org/abs/math/0503432}{{\ttfamily math/0503432}}].

\bibitem{freidel2017generalised}
L.~Freidel, F.~J. Rudolph and D.~Svoboda, \emph{{Generalised Kinematics for
  Double Field Theory}},
  \href{https://doi.org/10.1007/JHEP11(2017)175}{\emph{JHEP} {\bfseries 11}
  (2017) 175}, [\href{https://arxiv.org/abs/1706.07089}{{\ttfamily
  1706.07089}}].

\bibitem{deligne1999quantum}
P.~Deligne, P.~I. Etingof, D.~S. Freed and A.~M. Society, \emph{Quantum fields
  and strings: a course for mathematicians}, vol.~1.
\newblock American Mathematical Society Providence, 1999.

\bibitem{susy_topological_thesis}
J.~van~der Leer~Dur{\'a}n, \emph{Supersymmetric sigma models and generalized
  complex geometry},  Master's thesis, 2012.

\bibitem{Freidel:2018tkj}
L.~Freidel, F.~J. Rudolph and D.~Svoboda, \emph{{A Unique Connection for Born
  Geometry}},  \href{https://arxiv.org/abs/1806.05992}{{\ttfamily 1806.05992}}.

\bibitem{Bredthauer:2006sz}
A.~Bredthauer, \emph{{Generalized Hyperkaehler Geometry and Supersymmetry}},
  \href{https://doi.org/10.1016/j.nuclphysb.2007.03.004}{\emph{Nucl. Phys.}
  {\bfseries B773} (2007) 172--183},
  [\href{https://arxiv.org/abs/hep-th/0608114}{{\ttfamily hep-th/0608114}}].

\bibitem{Cortes:2010ykx}
V.~Cortés, ed., \emph{{Handbook of Pseudo-Riemannian Geometry and
  Supersymmetry}}, vol.~16.
\newblock 2010, \href{https://doi.org/10.4171/079}{10.4171/079}.

\bibitem{Cortes:2003zd}
V.~Cortes, C.~Mayer, T.~Mohaupt and F.~Saueressig, \emph{{Special geometry of
  Euclidean supersymmetry. 1. Vector multiplets}},
  \href{https://doi.org/10.1088/1126-6708/2004/03/028}{\emph{JHEP} {\bfseries
  03} (2004) 028}, [\href{https://arxiv.org/abs/hep-th/0312001}{{\ttfamily
  hep-th/0312001}}].

\bibitem{Cruceanu}
V.~Cruceanu, P.~Fortuny and P.~Gadea, \emph{A survey on paracomplex geometry},
  \href{https://doi.org/10.1216/rmjm/1181072105}{\emph{Rocky Mountain J. Math.}
  {\bfseries 26} (03, 1996) 83--115}.

\bibitem{moroianu2007lectures}
A.~Moroianu, \emph{Lectures on K{\"a}hler geometry}, vol.~69.
\newblock Cambridge University Press, 2007.

\bibitem{bejan1993existence}
C.-L. Bejan, \emph{The existence problem of hyperbolic structures on vector
  bundles}, {\emph{Publ. Inst. Mat. Beograd} {\bfseries 53} (1993) 133--138}.

\bibitem{norden2}
M.~Iscan and A.~Salimov, \emph{On k{\"a}hler-norden manifolds},
  {\emph{Proceedings-Mathematical Sciences} {\bfseries 119} (2009) 71--80}.

\bibitem{norden1}
V.~Oproiu and N.~Papaghiuc, \emph{Some examples of almost complex-manifolds
  with norden metric}, {\emph{PUBLICATIONES MATHEMATICAE-DEBRECEN} {\bfseries
  41} (1992) 199--211}.

\bibitem{B-manifolds}
K.~Gribachev, D.~Mekerov and G.~Djelepov, \emph{Generalized b-manifold},
  {\emph{DOKLADI NA BOLGARSKATA AKADEMIYA NA NAUKITE} {\bfseries 38} (1985)
  299--302}.

\bibitem{chiral1}
M.~Staikova and K.~Gribachev, \emph{Canonical connections and their conformal
  invariants on riemannian almost-product manifolds}, {\emph{Serdica}
  {\bfseries 18} (1992) 150--161}.

\bibitem{hitchin1990hypersymplectic}
N.~Hitchin, \emph{Hypersymplectic quotients}, {\emph{Acta Acad. Sci.
  Tauriensis} {\bfseries 124} (1990) 169--180}.

\bibitem{kamada1999neutral}
H.~Kamada et~al., \emph{Neutral hyperkahler structures on primary kodaira
  surfaces}, {\emph{Tsukuba journal of mathematics} {\bfseries 23} (1999)
  321--332}.

\bibitem{Ooguri:1990ww}
H.~Ooguri and C.~Vafa, \emph{{Selfduality and $N=2$ String {MAGIC}}},
  \href{https://doi.org/10.1142/S021773239000158X}{\emph{Mod. Phys. Lett.}
  {\bfseries A5} (1990) 1389--1398}.

\bibitem{Ooguri:1991fp}
H.~Ooguri and C.~Vafa, \emph{{Geometry of N=2 strings}},
  \href{https://doi.org/10.1016/0550-3213(91)90270-8}{\emph{Nucl. Phys.}
  {\bfseries B361} (1991) 469--518}.

\bibitem{Davidov:2009uc}
J.~Davidov, G.~Grantcharov, O.~Mushkarov and M.~Yotov,
  \emph{{Para-hyperhermitian surfaces}},
  \href{https://arxiv.org/abs/0906.0546}{{\ttfamily 0906.0546}}.

\bibitem{Davidov2012Surfaces}
J.~Davidov, G.~Grantcharov, O.~Mushkarov and M.~Yotov, \emph{Compact complex
  surfaces with geometric structures related to split quaternions},
  \href{https://doi.org/10.1016/j.nuclphysb.2012.07.024}{\emph{Nuclear Phys. B}
  {\bfseries 865} (2012) 330--352}.

\bibitem{freidel2014born}
L.~Freidel, R.~G. Leigh and D.~Minic, \emph{{Born Reciprocity in String Theory
  and the Nature of Spacetime}},
  \href{https://doi.org/10.1016/j.physletb.2014.01.067}{\emph{Phys. Lett.}
  {\bfseries B730} (2014) 302--306},
  [\href{https://arxiv.org/abs/1307.7080}{{\ttfamily 1307.7080}}].

\bibitem{freidel2015metastring}
L.~Freidel, R.~G. Leigh and D.~Minic, \emph{{Metastring Theory and Modular
  Space-time}}, \href{https://doi.org/10.1007/JHEP06(2015)006}{\emph{JHEP}
  {\bfseries 06} (2015) 006},
  [\href{https://arxiv.org/abs/1502.08005}{{\ttfamily 1502.08005}}].

\end{thebibliography}\endgroup

\end{document}